\documentclass[11pt]{article}
\usepackage{amssymb,amsmath,latexsym,a4wide}

\newcommand{\tree}[1]{\ensuremath{\mathcal{#1}}}
\newcommand{\arc}{\ensuremath{\leadsto}}
\newcommand{\movearrow}{\ensuremath{\stackrel{\move{}}{\longrightarrow}}}

\newcommand{\cclass}[1]{\ensuremath{\mathbf{#1}}}
\newcommand{\bigo}[1]{\ensuremath{\mathcal{O}(#1)}}
\newcommand{\ecl}[1]{\ensuremath{\mathsf{ecl}(#1)}}

\newcommand{\D}{\ensuremath{\Delta}}
\newcommand{\G}{\ensuremath{\Gamma}}

\newcommand{\vp}{\ensuremath{\varphi}}

\newcommand{\fm}[1]{\emph{#1}}
\newcommand{\de}[1]{\emph{#1}}

\newcommand{\power}[1]{\ensuremath{\mathcal{P}(#1)}}
\newcommand{\union}{\, \cup \,}
\newcommand{\bigunion}{\bigcup \,}

\newcommand{\inter}{\, \cap \,}
\newcommand{\crh}[2]{\ensuremath{\{\, #1 \mid \, #2\, \}}}
\newcommand{\set}[1]{\ensuremath{ \{ #1 \} }}
\newcommand{\card}[1]{\ensuremath{|#1|}}
\newcommand{\nat}{\ensuremath{\mathbb{N}}}

\newcommand{\CTL}{\textbf{CTL}}
\newcommand{\LTL}{\ensuremath{\textbf{LTL}}}

\newcommand{\ext}[2]{\ensuremath{\|#1\|_{#2}}}
\newcommand{\con}{\wedge}

\newcommand{\dis}{\vee}

\newcommand{\imp}{\rightarrow}

\newcommand{\equivalence}{\leftrightarrow}
\newcommand{\ap}{\textbf{\texttt{AP}}}
\newcommand{\truth}{\ensuremath{\top}}
\newcommand{\truthset}[2]{\ensuremath{\| #2 \|_{\mmodel{#1}}}}
\newcommand{\next}{\!\raisebox{-.2ex}{ %possibly add a little space before
            \mbox{\unitlength=0.9ex
            \begin{picture}(2,2)
            \linethickness{0.06ex}
            \put(1,1){\circle{2}} % Draws circle with
            \end{picture}}}       % diameter 2 at centre 1,1
            \,}
          \newcommand{\until}{\ensuremath{\hspace{2pt}\mathcal{U}}}

\newcommand{\kframe}[1]{\ensuremath{\mathfrak{#1}}}
\newcommand{\mmodel}[1]{\ensuremath{\mathcal{#1}}}
\newcommand{\hintikka}[1]{\ensuremath{\mathcal{#1}}}
\newcommand{\sat}[3]{\ensuremath{\mmodel{#1}, #2 \Vdash #3}}
\newcommand{\notsat}[3]{\ensuremath{\mmodel{#1}, #2 \nVdash #3}}
\newcommand{\cl}[1]{\ensuremath{\mathsf{cl}(#1)}}
\newcommand{\Rule}[1]{\textbf{(#1)}}

\newcommand{\ATL}{\textbf{ATL}}

\newcommand{\coal}[1]{\ensuremath{\langle \hspace{-2.2pt} \langle #1
    \rangle \hspace{-2.5pt} \rangle}}
\newcommand{\coalnext}[1]{\ensuremath{\langle \hspace{-2.5pt} \langle
    #1 \rangle \hspace{-2.5pt} \rangle \next}}
\newcommand{\coalbox}[1]{\ensuremath{\langle \hspace{-2.5pt} \langle
    #1 \rangle \hspace{-2.5pt} \rangle \Box}}
\newcommand{\coaldiam}[1]{\ensuremath{\langle \hspace{-2.5pt} \langle
    #1 \rangle \hspace{-2.5pt} \rangle \Diamond}}
\newcommand{\move}[1]{\ensuremath{\sigma_{\hspace{-2.0pt}#1}}}
\newcommand{\nmoves}[2]{\ensuremath{d_{\hspace{-0.5pt}#1}(#2)}}
\newcommand{\moves}[2]{\ensuremath{D_{\hspace{-1.0pt}#1}(#2)}}
\newcommand{\comoves}[2]{\ensuremath{D^c_{\hspace{-2.5pt}#1}(#2)}}
\newcommand{\merge}[2]{\ensuremath{#1 \sqcup #2}}
\newcommand{\supmove}{\ensuremath{\sqsupseteq}}
\newcommand{\comove}[1]{\ensuremath{\sigma^c_{\hspace{-2.0pt}#1}}}
\newcommand{\agents}{\ensuremath{\Sigma}}

\newcommand{\sucpr}[1]{\ensuremath{\mathbf{prestates}(#1)}}

\newcommand{\suctab}[2]{\ensuremath{\mathbf{succ}_{#2}(#1)}}

\newcommand{\str}[1]{\ensuremath{F_{\hspace{-2.5pt}#1}}}
\newcommand{\costr}[1]{\ensuremath{{F^{c}}_{\hspace{-3.2pt}#1}}}
\newcommand{\tableau}[1]{\ensuremath{\mathcal{#1}}}
\newtheorem{theorem}{Theorem}[section]
\newtheorem{lemma}[theorem]{Lemma}

\newtheorem{corollary}[theorem]{Corollary}
\newenvironment{proof}{\noindent
  \textbf{Proof.}}{\hfill$\Box$\\}
\newtheorem{notation}[theorem]{Notational convention}
\newtheorem{definition}[theorem]{Definition}
\newtheorem{example}{Example}

\newtheorem{remark}[theorem]{Remark}

\newcommand{\cut}[1]{}

\newcommand{\illegalmove}{\sharp}
\newcommand{\st}[1]{\textbf{states}(#1)}
\newcommand{\brancharrow}{\ensuremath{\Longrightarrow}}
\renewcommand{\movearrow}{\ensuremath{\stackrel{\move{}}{\longrightarrow}}}

\begin{document}

\title{Tableau-based decision procedures \\ for
  logics of strategic ability in multiagent systems}

\author{Valentin Goranko\footnote{School of Mathematics, University of
    the Witwatersrand, South Africa,
    \texttt{goranko@maths.wits.ac.za}} and Dmitry
  Shkatov\footnote{School of Computer Science, University of the
    Witwatersrand, South Africa, \texttt{dmitry@cs.wits.ac.za}}}
\maketitle

\begin{abstract}
  We develop an incremental tableau-based decision procedures for the
  Alternating-time temporal logic \ATL\ and some of its variants.
  While running within the theoretically established complexity upper
  bound, we claim that our tableau is practically more efficient in
  the average case than other decision procedures for \ATL\ known so
  far.  Besides, the ease of its adaptation to variants of \ATL\
  demonstrates the flexibility of the proposed procedure.
\end{abstract}


%%%%%%%%%%%%%%%%%%%%%%%%%%%%%%%%%%
%%%%%%%%%%%%%%%%%%%%%%%%%%%%%%%%%%
\section{Introduction}
\label{sec:intro}
%%%%%%%%%%%%%%%%%%%%%%%%%%%%%%%%%%
%%%%%%%%%%%%%%%%%%%%%%%%%%%%%%%%%%

Multiagent systems (\cite{Fagin95knowledge}, \cite{Weiss99},
\cite{Wooldridge02}, \cite{Shoham08}) are an increasingly important
and active area of interdisciplinary research on the border of
computer science, artificial intelligence, and game theory, as they
model a wide variety of phenomena in these fields, including open and
interactive systems, distributed computations, security protocols,
knowledge and information exchange, coalitional abilities in games,
etc. Not surprisingly, a number of logical formalisms have been
proposed for specification, verification, and reasoning about
multiagent systems. These formalisms, broadly speaking, fall into two
categories: those for reasoning about \emph{knowledge of agents} and
those for reasoning about \emph{abilities of agents}. In the present
paper, we deal with the latter variety of logics, the most influential
among them being the so-called Alternating-time temporal logic (\ATL),
introduced in \cite{AHK97} and further developed in \cite{AHK98} and
\cite{AHK02}.

\ATL\ and its modifications can be applied to multiagent systems in a
similar way as temporal logics, such as \LTL\ and \CTL, are applied to
reactive systems.  First, since \ATL-models can be viewed as
abstractions of multiagent systems, \ATL\ can be used to verify and
specify properties of such systems. Given a model \mmodel{M} and an
\ATL-formula \vp, the task of verifying \mmodel{M} with respect to the
property expressed by \vp\ is, in logical terms, the model checking
problem for \ATL, extensively discussed in \cite{AHK02}; a
model-checker for \ATL\ has also been developed, see \cite{AHMQR98}.
Second, \ATL\ can be used to design multiagent systems conforming to a
given specification; then, \ATL-formulae are viewed as specifications
to be realized rather than verified.  In logical terms, this is the
\emph{constructive satisfiability problem} for \ATL: given a formula
\vp, check if it is satisfiable and, if so, construct a model of
\vp.

In the temporal logic tradition, in which \ATL\ is rooted, two
approaches to constructive satisfiability are predominant:
\emph{tableau-based} and \emph{automata-based}.  The relationship
between the two is not, in our view, sufficiently well understood
despite being widely acknowledged.  The automata-based approach to
\ATL-satisfiability was developed in \cite{vanDrimmelen03} and
\cite{GorDrim06}.

The aim of the present paper is to develop practically useful
``incremental'' (also called ``goal-driven'') tableau-based
decision procedures (in the style of \cite{Wolper85}) for the
constructive satisfiability problem for the ``standard'' \ATL\ and
some of its modifications. Incremental tableaux form one of the two
most popular types of tableau-based decision procedures for modal and
temporal logics with fixpoint-defined operators (the most widely known
examples being \LTL\ and \CTL). It should be noted that, while
tableaux for logics with such operators employ all common features of the
``traditional'' tableaux for modal logics, comprehensively covered
in~\cite{Fitting83}, \cite{Gore98}, and \cite{Fitting07}, they differ
substantially from the latter, because they involve a loop-detecting
(or equivalent) procedure that checks for the satisfaction of
formulas containing fixpoint operators.

As already mentioned, the alternative to the incremental tableaux for
logics with fixpoint-definable operators are the ``top-down''
tableaux, developed, for the case of \CTL\ and some closely related
logics, in \cite{EmHal85} (see also \cite{Emerson90}) and essentially
applied to \ATL\ in \cite{WLWW06}. A major practical drawback of the top-down
tableaux is that, while they run within the same worst-case
complexity bound as the corresponding incremental tableaux, their
performance matches the worst-case upper bound for \emph{every}
formula to be tested for satisfiability.  The reason for this
``practical inefficiency'' of the top-down tableaux is that they
invariably involve the construction of all maximally consistent
subsets of the so-called ``extended closure'' of the formula to be
tested, which in itself requires the number of steps of the order of
the theoretical upper bound\footnote{It should be stressed that the
  top-down tableaux for \ATL\ presented in \cite{WLWW06} were not
  meant to serve as a practically efficient method of checking
  \ATL-satisfiability, but rather were used as a tool for establishing
  the \cclass{ExpTime} upper bound for \ATL, in particular, for the
  case when the number of agents is not fixed, as assumed
  in~\cite{vanDrimmelen03} and \cite{GorDrim06}, but taken as a
  parameter.}. Some authors consider it to be so great a disadvantage
of the top-down tableaux that they propose non-optimal complexity
tableaux for such logics, which they claim to perform better in
practice (see \cite{AGF07}).

We believe that the incremental tableaux developed in the present
paper are intuitively more appealing, practically more efficient, and
therefore more suitable both for manual and for computerized execution
than the top-down tableaux, not least because checking satisfiability
of a formula using incremental tableaux takes, on average, much less
time than predicted by the worst-case complexity
upper-bound. Furthermore, incremental tableaux are quite flexible
and amenable to modifications and extensions covering not only
variants of \ATL\ considered in this paper, but also a number of other
logics for multiagent systems, such as multiagent epistemic logics
(see~\cite{Fagin95knowledge}), for which analogous tableau-based
decision procedures have recently been developed in \cite{GorShkat2}
and \cite{GorShkat3}. Lastly, it should be noted, that our tableau
method naturally reduces (in the one-agent case) to incremental
tableaux for \CTL, which is practically more efficient (again, on
average) than Emerson and Halpern's top-down tableaux
from~\cite{EmHal85}.

We should also mention that yet another type of tableau-based decision
procedure for \ATL, the so-called ``tableau games'', has been
considered in \cite{Hansen04}. Even though neither soundness nor
completeness of the tableau games for the full \ATL\ has been
established in~\cite{Hansen04}, sound and complete tableau games for
the ``Next-time fragment of \ATL'', namely, the Coalition Logic
\textbf{CL}, introduced in \cite{Pauly01} (see also \cite{Pauly01a}
and \cite{Pauly02}), have been presented in~\cite{Hansen04}.

The structure of the present paper is as follows: after introducing
the syntactic and semantic basics of \ATL\ in section \ref{sec:atl},
we introduce, in section \ref{sec:atl-hintikka}, concurrent game
Hintikka structures and show that they provide semantics for \ATL\
that is, satisfiability-wise, equivalent to the one based on
concurrent game models described in section~\ref{sec:atl}. In
section~\ref{sec:atl-tableau-procedure}, we develop the tableau
procedure for \ATL\ and analyze its complexity, while in section
\ref{sec:atl-tableaux-sound-and-complete} we prove its soundness and
completeness using concurrent game Hintikka structures introduced in
section~\ref{sec:atl-hintikka}. In section~\ref{sec:flexibility}, we
briefly discuss adaptations of our tableau method for some
modifications of \ATL.

%%%%%%%%%%%%%%%%%%%%%%%%%%%%%%%%%%
%%%%%%%%%%%%%%%%%%%%%%%%%%%%%%%%%%
\section{Preliminaries: the multiagent logic \ATL}
\label{sec:atl}
%%%%%%%%%%%%%%%%%%%%%%%%%%%%%%%%%%
%%%%%%%%%%%%%%%%%%%%%%%%%%%%%%%%%%

\ATL\ was introduced in \cite{AHK97}, and further developed in
\cite{AHK98} and \cite{AHK02}, as a logical formalism to reason about
open systems (\cite{Hewitt90}), but it naturally applies to the more
general case of multiagent systems. Technically, \ATL\ is an
extension of the multiagent coalition logics \textbf{CL} and
\textbf{ECL} studied in \cite{Pauly01}, \cite{Pauly01a}, and
\cite{Pauly02} (for a comparison of the logics, see \cite{Goranko01}
and \cite{GorJam04}).

%%%%%%%%%%%%%%%%%%%%%%%%%%%%%%%%%%
\subsection{ATL syntax}
\label{sec:atl-syntx}
%%%%%%%%%%%%%%%%%%%%%%%%%%%%%%%%%%

\ATL\ is a multimodal logic with \CTL-style modalities indexed by
subsets, commonly called \fm{coalitions}, of the finite, non-empty set
of (names of) \fm{agents}, or players, that can be referred to in the
language.  Thus, formulae of \ATL\ are defined with respect to a
finite, non-empty set $\agents$ of agents, usually denoted by the
natural numbers $1$ through $\card{\agents}$ (the cardinality of
$\agents$), and a finite or countably infinite set \ap\ of atomic
propositions.

\begin{definition}
  \label{def:atl-syntax}
  \ATL-formulae are defined by the following grammar: $$\vp :=
  p \mid \neg \vp \mid (\vp_1 \imp \vp_2) \mid \coalnext{A} \vp \mid
  \coalbox{A} \vp \mid \coal{A} \vp_1 \until \vp_2,$$ where $p$ ranges
  over $\ap$ and $A$ ranges over $\power{\agents}$, the power-set of
  $\agents$.
\end{definition}

Notice that we allow (countably) infinitely many propositional
parameters, but in line with traditional presentations of \ATL\ (see,
for example, \cite{AHK02}), only finitely many names of agents.  We
will show, however, after introducing \ATL-semantics, that this latter
restriction is not essential (see Remark~\ref{rm:num_of_agents} below)
and thus does not result in a loss of generality.

The other boolean connectives and the propositional constant $\top$
(``truth'') can be defined in the usual way.  Also, $\coaldiam{A} \vp$
can be defined as $\coal{A} \top \until \vp$.  As will become
intuitively clear from the semantics of \ATL, $\coaldiam{A} \vp$ and
$\coalbox{A} \vp$ are not interdefinable\footnote{A formal proof of
  this claim would require a suitable semantic argument, e.g., one
  involving bisimulations between models for \ATL. As such an argument
  would take up quite a lot of space and is not immediately relevant
  to the contents of the present paper, we do not pursue it in this
  paper.}.

The expression $\coal{A}$, where $A \subseteq \agents$, is a
\fm{coalition quantifier} (also referred to as ``path quantifier'' in
the literature), while $\next$ (``next''), $\Box$ (``always''), and
$\until$ (``until'') are \fm{temporal operators}.  Like in \CTL, where
every temporal operator has to be preceded by a path quantifier, in
\ATL\ every temporal operator has to be preceded by a coalition
quantifier.  Thus, \de{modal operators} of \ATL\ are pairs made up of
a coalition quantifier and a temporal operator.

We adopt the usual convention that unary connectives have a stronger
binding power than binary ones; when this convention helps
disambiguate a formula, we usually omit the parentheses associated
with binary connectives.

Formulae of the form $\coal{A} \vp \until \psi$ and $\neg \coalbox{A}
\vp$ are called \fm{eventualities}, for the reason explained later on.

%%%%%%%%%%%%%%%%%%%%%%%%%%%%%%%%%%
\subsection{\ATL\ semantics}
\label{sec:semantics}
%%%%%%%%%%%%%%%%%%%%%%%%%%%%%%%%%%

While the syntax of \ATL\ remained unchanged from \cite{AHK97} to
\cite{AHK02}, the semantics, originally based on ``alternating
transition systems'', was revised in~\cite{AHK02}, where the notion of
``concurrent game structures'' was introduced. The latter are
essentially equivalent to ``multi-player game models''
(\cite{Pauly01}, \cite{Pauly02}) and are more general than, yet
yielding the same set of validities as, alternating transition
systems---see \cite{Goranko01},\cite{GorJam04}.

In the present paper, we use the term ``concurrent game models'' to
refer to the ``concurrent game structures'' from~\cite{AHK02} and, in
keeping with the long-established tradition in modal logic, the term
``concurrent game frames'' to refer to the structures resulting from
those by abstracting away from the meaning of atomic propositions.

%%%%%%%%%%%%%%%%%%%%%%%%%%%%%%%%%%
\subsubsection{Concurrent game frames}
\label{sec:CGS}
%%%%%%%%%%%%%%%%%%%%%%%%%%%%%%%%%%

Concurrent game frames are to \ATL\ what Kripke frames are to standard
modal logics.

\begin{definition}
  \label{def:cgf}
  A \de{concurrent game frame} (for short, CGF) is a tuple $\kframe{F}
  = (\agents, S, d, \delta)$, where
  \begin{itemize}
  \item $\agents$ is a finite, non-empty set of \de{agents}, referred
    to by the numbers $1$ through $\card{\agents}$; subsets of
    $\agents$ are called \de{coalitions};
  \item $S \ne \emptyset$ is a set of \de{states};
  \item $d$ is a function assigning to every agent $a \in \agents$ and
    every state $s \in S$ a natural number $d_a(s) \geq 1$ of
    \de{moves}, or actions, available to agent $a$ at state $s$; these
    moves are identified with the numbers $0$ through $\nmoves{a}{s} -
    1$.  For every state $s \in S$, a \de{move vector} is a $k$-tuple
    $(\move{1}, \ldots, \move{k})$, where $k = \card{\agents}$, such
    that $0 \leq \move{a} < \nmoves{a}{s}$ for every $1 \leq a \leq k$
    (thus, $\move{a}$ denotes an arbitrary action of agent $a \in
    \agents$). Given a state $s \in S$, we denote by \moves{a}{s} the
    set \set{0, \ldots, \nmoves{a}{s}-1} of all moves available to
    agent $a$ at $s$, and by $D(s)$ the set $\prod_{a \in \agents}
    \moves{a}{s}$ of all move vectors at $s$; with $\move{}$ we denote
    an arbitrary member of $D(s)$.
  \item $\delta$ is a \de{transition function} assigning to every $s
    \in S$ and $\move{} \in D(s)$ a state $\delta (s, \move{})
    \in S$ that results from $s$ if every agent $a \in \agents$ plays
    move $\move{a}$.
  \end{itemize}
\end{definition}

All definitions in the remainder of this section refer to an
arbitrarily fixed CGF.

\begin{definition}
  For two states $s, s' \in S$, we say that $s'$ is a \de{successor}
  of $s$ (or, for brevity, an $s$-successor) if $s' = \delta (s,
  \move{})$ for some $\move{} \in D(s)$.
\end{definition}

\begin{definition}
  A \de{run} in $\kframe{F}$ is an infinite sequence $\lambda = s_0,
  s_1, \ldots$ of elements of $S$ such that, for all $i \geq 0$, the
  state $s_{i+1}$ is a successor of the state $s_i$. Elements of the
  domain of $\lambda$ are called \fm{positions}.  For a run $\lambda$
  and positions $i,j \geq 0$, we use $\lambda[i]$ and $\lambda[j, i]$
  to denote the $i$th state of $\lambda$ and the finite segment $s_j,
  s_{j+1} \ldots, s_{i}$ of $\lambda$, respectively.  A run with
  $\lambda[0] = s$ is referred to as an \de{$s$-run}.
\end{definition}

Given a tuple $\tau$, we interchangeably use $\tau_n$ and $\tau(n)$ to
refer to the $n$th element of $\tau$. We use the symbol $\illegalmove$
as a placeholder for an arbitrarily fixed move of a given agent.

\begin{definition}
  \label{def:A-moves}
  Let $s \in S$ and let $A \subseteq \agents$ be a coalition of
  agents, where $\card{\agents} = k$.  An \de{$A$-move} \move{A} at
  state $s$ is a $k$-tuple $\move{A}$ such that $\move{A} (a) \in D_a
  (s)$ for every $a \in A$ and $\move{A} (a') = \illegalmove$ for
  every $a' \notin A$. We denote by \moves{A}{s} the set of all
  $A$-moves at state $s$.
\end{definition}

Alternatively, $A$-moves at $s$ can be defined as equivalence classes
on the set of all move vectors at $s$, where each equivalence class
is determined by the choices of moves of agents in $A$.

\begin{definition}
  We say that a \de{move vector $\move{}$ extends an $A$-move
    $\move{A}$} and write $\move{A} \sqsubseteq \move{}$, or $\move{}
  \sqsupseteq \move{A}$, if $\move{} (a) = \move{A} (a)$ for every $a
  \in A$.
\end{definition}

Given a coalition $A \subseteq \agents$, an $A$-move $\move{A} \in
\moves{A}{s}$, and a $(\agents \setminus A)$-move $\move{\agents
  \setminus A} \in D_{\agents \setminus A} (s)$, we denote by
$\merge{\move{A}}{\move{\agents \setminus A}}$ the unique $\move{} \in
D(s)$ such that both $\move{A} \sqsubseteq \move{}$ and $\move{\agents
  \setminus A} \sqsubseteq \move{}$.

\begin{definition}
  Let $\move{A} \in \moves{A}{s}$.  The \de{outcome} of \move{A} at
  $s$, denoted by $out(s, \move{A})$, is the set of all states $s'$
  for which there exists a move vector $\move{} \in D(s)$ such that
  $\move{A} \sqsubseteq \move{}$ and $\delta (s, \move{}) = s'$.
\end{definition}

Concurrent game frames are meant to model coalitions of agents
behaving strategically in pursuit of their goals.  Given a coalition
$A$, a strategy for $A$ is, intuitively, a rule determining at a given
state what $A$-move the agents in $A$ should play.  Given a state as a
component of a run, the strategy for agents in $A$ at that state may
depend on some part of the history of the run\footnote{In general, we
  might consider the case when an agent can remember \emph{any} part
  of the history of the run; it suffices, however, for our purposes in
  this paper to consider only those parts that are made up of
  consecutive states of a run.}, the length of this ``remembered''
history being a parameter formally represented by an ordinal $\gamma
\leq \omega$.  Intuitively, players using a $\gamma$-recall strategy
can ``remember'' any number $n < \gamma$ of the previous
\emph{consecutive} states of the run.  If $\gamma$ is a natural
number, then $\gamma$ can be thought of as a number of the consecutive
states, including the current state, on which an agent is basing its
decision of what move to play.  If, however, $\gamma = \omega$, then
an agent can remember any number of the previous consecutive states of
the run.

Given a natural number $n$, by $S^n$ we denote the set of sequences of
elements of $S$ of length $n$; the length of a sequence $\kappa$ is
denoted by $|\kappa|$ and the last element of $\kappa$ by $l(\kappa)$.

\begin{definition}
  \label{def:A-strategies}
  Let $A \subseteq \agents$ be a coalition and $\gamma$ an ordinal
  such that $1 \leq \gamma \leq \omega$.  A \de{$\gamma$-recall
    strategy for $A$} (or, a \de{$\gamma$-recall $A$-strategy}) is a
  mapping $\str{A}[\gamma] : \bigunion_{1 \leq n < 1+\gamma} S^n
  \mapsto \bigunion \crh{\moves{A}{s}}{s \in S}$ such that
  $\str{A}[\gamma] (\kappa) \in \moves{A}{l(\kappa)}$ for every
  $\kappa \in \bigunion_{1 \leq n < 1+\gamma} S^n$.
\end{definition}

\begin{remark}
  Given that $1 + \omega = \omega$, the condition of
  Definition~\ref{def:A-strategies} for the case of $\omega$-recall
  strategies can be rephrased in a simpler form as follows:
  $\str{A}[\omega] : \bigunion_{1 \leq n < \omega} S^n \mapsto
  \bigunion \crh{\moves{A}{s}}{s \in S}$ such that $\str{A}[\omega]
  (\kappa) \in \moves{A}{l(\kappa)}$ for every $\kappa \in
  \bigunion_{1 \leq n < \omega} S^n$.
\end{remark}

\begin{definition}
  Let $\str{A}[\gamma]$ be a $\gamma$-recall $A$-strategy. If $\gamma
  = \omega$, then $\str{A}[\gamma]$ is referred to as a
  \de{perfect-recall $A$-strategy}; otherwise, $\str{A}[\gamma]$ is
  referred to as a \de{bounded-recall $A$-strategy}. Furthermore, if
  $\gamma = 1$, then $\str{A}[\gamma]$ is referred to as a
  \de{positional $A$-strategy}.
\end{definition}

Thus, agents using a perfect-recall strategy have potentially
unlimited memory; those using positional strategies have none ($\gamma
= 1$ means that an agent bases its decisions on one state only, i.e.,
the current one); in between, agents using $n$-recall strategies, for
$1 < n < \omega$, can base their decisions on the $n-1$ previous
consecutive states of the run as well as the current state.  We
usually write $\str{A}$ instead of $\str{A}[\gamma]$ when $\gamma$ is
understood from the context.

\begin{remark}
  Even though the concept of $n$-recall strategies, for $1 < n <
  \omega$ is of some interest in itself, in the present paper it is
  introduced for purely technical reasons, to be used in the proof of
  the satisfiability-wise equivalence (see
  Theorem~\ref{thr:from_hintikka_to_models} below) of the semantics of
  \ATL\ based on concurrent game models and the one based on
  concurrent game Hintikka structures as well as in the completeness
  proof for our tableau procedure.

  We note, however, that a more realistic notion of finite-memory
  strategy is the one allowing a strategy to be computed by a finite
  automaton reading a sequence of states in the history of a run and
  producing a move to be played, as proposed in \cite{Thomas95}.
\end{remark}

\begin{definition}
  \label{def:outcomes_of_strategies}
  Let $\str{A}[\gamma]$ be an $A$-strategy.  The \de{outcome of
    $\str{A}[\gamma]$ at state $s$}, denoted by $out(s,
  \str{A}[\gamma])$, is the set of all $s$-runs $\lambda$ such that
\[
  \begin{array}{lc}
    (\gamma) & \lambda[i+1] \in out (\lambda[i],
    \str{A}[\gamma](\lambda[j, i]))
    \text{ holds for all } i \geq 0, \\
    & \mbox{ where } j = \max (i - \gamma + 1,0).
  \end{array}
 \]

\end{definition}

Note that for positional strategies condition $(\mathbf{\gamma})$
reduces to
$$
\begin{array}{ll}
  (\mathbf{P}) & \lambda[i+1] \in out (\lambda[i],
  \str{A}(\lambda[i])), \text{ for all } i \geq 0,
\end{array}
$$
whereas for perfect-recall strategies it reduces to
$$
\begin{array}{ll}
  (\mathbf{PR}) & \lambda[i+1] \in out (\lambda[i],
  \str{A}(\lambda[0, i])),
  \text{ for all } i \geq 0.
\end{array}
$$

%%%%%%%%%%%%%%%%%%%%%%%%%%%%%%%%%%
\subsubsection{Truth of \ATL-formulae}
\label{sec:atl-semantics-informally}
%%%%%%%%%%%%%%%%%%%%%%%%%%%%%%%%%%

We are now ready to define the truth of \ATL-formulae in terms of
concurrent game models and perfect-recall strategies.

\begin{definition}
  \label{def:cgm}
  A \de{concurrent game model} (for short, CGM) is a tuple $\mmodel{M}
  = (\kframe{F}, \ap, L)$, where
  \begin{itemize}
  \item \kframe{F} is a concurrent game frame;
  \item \ap\ is a set of atomic propositions;
  \item $L$ is a labeling function $L: S \rightarrow \power{\ap}$.
    Intuitively, the set $L(s)$ contains the atomic propositions that
    are true at state $s$.
  \end{itemize}
\end{definition}

\begin{definition}
  \label{def:atl-satisfaction}
  Let $\mmodel{M} = (\agents, S, d, \delta, \ap, L)$ be a concurrent
  game model.  The satisfaction relation $\Vdash$ is inductively
  defined for all $s \in S$ and all \ATL-formulae as follows:
  \begin{itemize}
  \item \sat{M}{s}{p} iff $p \in L(s)$, for all $p \in \ap$;
  \item \sat{M}{s}{\neg \vp} iff \notsat{M}{s}{\vp};
  \item \sat{M}{s}{\vp \imp \psi} iff \sat{M}{s}{\vp} implies
    \sat{M}{s}{\psi};
  \item \sat{M}{s}{\coalnext{A} \vp} iff there exists an $A$-move
    $\move{A} \in \moves{A}{s}$ such that \sat{M}{s'}{\vp} for all
    $s' \in out(s, \move{A})$;
  \item \sat{M}{s}{\coalbox{A} \vp} iff there exists a
    perfect-recall $A$-strategy \str{A} such that
    \sat{M}{\lambda[i]}{\vp} holds for all $\lambda \in out(s,
    \str{A})$ and all positions $i \geq 0$;
  \item \sat{M}{s}{\coal{A} \vp \until \psi} iff there exists a
    perfect-recall $A$-strategy \str{A} such that, for all $\lambda
    \in out(s, \str{A})$, there exists a position $i \geq 0$ with
    \sat{M}{\lambda[i]}{\psi} and \sat{M}{\lambda[j]}{\vp} holds
    for all positions $0 \leq j < i$.
  \end{itemize}
\end{definition}

\begin{definition}
  \label{def:atl-truth}
  Let $\theta$ be an \ATL-formula and $\G$ be a set of \ATL-formulae.
  \begin{itemize}
  \item $\theta$ is \de{true at a state $s$ of a CGM} \mmodel{M} if
    \sat{M}{s}{\theta}; $\G$ is \de{true at $s$}, denoted
    \sat{M}{s}{\G}, if \sat{M}{s}{\vp} holds for every $\vp \in
    \Gamma$;
  \item $\theta$ is \de{satisfiable in a CGM} \mmodel{M} if
    \sat{M}{s}{\theta} holds for some $s \in \mmodel{M}$; $\G$ is
    \de{satisfiable in} \mmodel{M} if \sat{M}{s}{\G} holds for some $s
    \in \mmodel{M}$;
  \item $\theta$ is \de{true in a CGM} \mmodel{M} if
    \sat{M}{s}{\theta} holds for every $s \in \mmodel{M}$.
\end{itemize}
\end{definition}

As the clauses for the modal operators $\coalbox{A} $ and $\coal{A}
\until $ in Definition~\ref{def:atl-satisfaction} involve strategies,
these will henceforth be referred to as \de{strategic operators}.

\begin{remark}
  \label{rm:num_of_agents}
  As in the present paper we are only concerned with satisfiability of
  single formulae (or, equivalently, finite sets of formulae), and a
  formula can only contain finitely many atomic propositions, the size
  of \ap\ is of no real significance for our purposes here. The issue
  of the cardinality of the set of agents $\agents$ is more involved,
  however, as infinite coalitions can be \emph{named} within a single
  formula, which would imply certain technical complications.
  Nevertheless, when interested in satisfiability of single formulae,
  the finiteness of $\agents$ does not result in a loss of generality.
  Indeed, as every formula $\varphi$ mentions only \emph{finitely many
    coalitions}, we can definite an equivalence relation of finite
  index on the set of agents that is naturally induced by $\varphi$;
  to wit, two agents are considered ``equivalent'' if they always
  occur (or not) together in all the coalitions mentioned in $\varphi$ (i.e. $a   \cong_{\vp} b$ if $a \in A$ iff $b \in A$ holds for every coalition
  $A$ mentioned in $\vp$). Then, $\varphi$ can be rewritten into a
  formula $\varphi'$ in which equivalence classes with respect to
  $\cong_{\vp}$ are treated as single agents. It is not hard to show
  that $\vp'$ is satisfiable iff $\vp$ is, and thus the satisfiability
  of the latter can be reduced to the satisfiablity of the former.
\end{remark}

%%%%%%%%%%%%%%%%%%%%%%%%%%%%%%%%%%
\subsection{Fixpoint characterization of strategic operators}
\label{sec:fixpoint_characterisation}
%%%%%%%%%%%%%%%%%%%%%%%%%%%%%%%%%%

In the tableau procedure described later on in the paper and in the
proofs of a number of results concerning \ATL, we will make use of the
fact that the strategic operators $\coalbox{A} $ and $\coal{A} \until$
can be given neat fixpoint characterizations, as shown in
\cite{GorDrim06}.  In this respect, \ATL\ turns out to be not much
different from \LTL\ and \CTL, whose ``long-term'' modalities are
well-known to have similar fixpoint characterizations.

The following definitions introduce set theoretic operators
corresponding to the semantics of the respective coalitional
modalities in a sense made precise in Theorem \ref{thr:fixpoint}.

\begin{definition}
  Let $(\agents, S, d, \delta)$ be a CGF and let $X \subseteq S$.
  Then, $[\coalnext{A}]$ is an operator $\power{S} \mapsto \power{S}$
  defined by the following condition: $s \in [\coalnext{A}] (X)$ iff
  there exists $\move{A} \in \moves{A}{s}$ such that $out(s, \move{A})
  \subseteq X$.
\end{definition}

\begin{definition}
  Let $(\agents, S, d, \delta)$ be a CGF and let $X, Y \subseteq S$.  Then,
  we define operators $[Y \inter \coalnext{A}]$ and $[Y \union
  \coalnext{A}]$ from $\power{S}$ to $\power{S}$ as expected:
  \begin{itemize}
  \item $[Y \inter \coalnext{A}](X) = Y \inter [\coalnext{A}] (X)$;
 \item $[Y \union \coalnext{A}](X) = Y \union [\coalnext{A}](X)$.

  \end{itemize}
\end{definition}

Given a formula $\vp$ and a model $\mmodel{M}$, we denote by
$\truthset{M}{\vp}$ the set \crh{s}{\sat{M}{s}{\vp}}; we simply write
$\truthset{}{\vp}$ when $\mmodel{M}$ is clear from the context.

Given a monotone operator $[\Omega]: \power{S} \mapsto \power{S}$, we
denote by $\mu X.  [\Omega] (X)$ and $\nu X. [\Omega] (X)$ the least
and greatest fixpoints of $[\Omega]$, respectively.

\begin{theorem}[Goranko, van Drimmelen \cite{GorDrim06}]
  \label{thr:fixpoint}
  Let $(\agents, S, d, \delta, \ap, L)$ be a CGM. Then, for any
  formulae $\vp,\psi$:
  \begin{itemize}
  \item $\ext{\coalnext{A} \vp}{} = [\coalnext{A}](\ext{\vp}{})$
  \item $\ext{\coalbox{A} \vp}{} = \nu X. [\ext{\vp}{} \inter
    \coalnext{A}] (X)$;
  \item $\ext{\coal{A} \vp \until \psi}{} = \mu X. [\ext{\psi}{}
    \union [\ext{\vp}{} \inter \coalnext{A}]] (X)$.
  \end{itemize}
\end{theorem}

\begin{corollary}
  \label{cor:fixpoint}
  The following equivalences hold at every state of every CGM with $A
  \subseteq \agents$:
  \begin{itemize}
  \item $\coalbox{A} \vp \equivalence \vp \con \coalnext{A} \coalbox{A}
    \vp$;
  \item $\coal{A} \vp \until \psi \equivalence \psi \dis (\vp \con
    \coalnext{A} \coal{A} \vp \until \psi)$;
  \end{itemize}
\end{corollary}

%%%%%%%%%%%%%%%%%%%%%%%%%%%%%%%%%%
\subsection{Tight, general, and loose \ATL-satisfiability}
\label{sec:satisfiability}
%%%%%%%%%%%%%%%%%%%%%%%%%%%%%%%%%%

Unlike the case of standard modal logics, it is natural to think of
several apparently different notions of \ATL-satisfiability.  The
differences lie along two dimensions: the types of strategies used in
the definition of the satisfaction relation and the relationship
between the set of agents mentioned in a formula and the set of agents
referred to in the language. We consider these issues in turn.

The notion of strategy, as introduced above, is dependent on the
amount of memory used to prescribe it. At one end of the spectrum are
\emph{positional} (or \emph{memoryless}) strategies, which only take
into consideration the current state of, but not any part of the
history of, the run; and at the other---\emph{perfect recall}
strategies, which take into account the entire history of the run. It
turns out, however, that these both ``extreme'' types of
strategy---and, hence, all those in between---yield equivalent
semantics in the case of \ATL\ (they, however, differ in the case of
the more expressive logic \ATL*, considered
in~\cite{AHK02}). Therefore, the above definition of truth of
\ATL-formulae (Definition~\ref{def:atl-satisfaction}) could have been
couched in terms of positional, rather than perfect-recall, strategies
without any changes in what formulae are satisfiable at which states.
This equivalence, first mentioned in~\cite{AHK02}, can be proved using
a model-theoretic argument; independently, it follows as a corollary
of the soundness and completeness theorems for the tableau procedure
presented below (see Corollary~\ref{cor:positionality}).

Now, assuming the type of strategies being fixed, one can consider
three different, at least on the face of it, notions of satisfiability
and validity for \ATL, depending on the relationship between the set
of agents mentioned in a formula and the set of agents referred to in
the language, as introduced in \cite{WLWW06}.

For every \ATL-formula $\theta$, we denote by $\agents_{\theta}$ the
set of agents occurring in $\theta$. When considering an \ATL-formula
$\theta$ in isolation, we may assume, without a loss of generality,
that the names of the agents occurring in $\theta$ are the numbers $1$
through $\card{\agents_{\theta}}$; hence, the following definitions.

\begin{definition}
  An \ATL-formula $\theta$ is \de{\agents-satisfiable}, for some
  $\agents \supseteq \agents_{\theta}$, if $\theta$ is satisfiable in
  a CGM $\mmodel{M} = (\agents, S, d, \delta, \ap, L)$; $\theta$ is
  \de{\agents-valid} if $\theta$ is true in every such CGM.
\end{definition}

\begin{definition}
  An \ATL-formula $\theta$ is \de{tightly satisfiable} if $\theta$ is
  satisfiable in a CGM $\mmodel{M} = (\agents_{\theta}, S, d, \delta,
  \ap, L)$; $\theta$ is \de{tightly valid} if $\theta$ is true in
  every such CGM.
\end{definition}

Clearly, $\theta$ is tightly satisfiable iff it is
$\agents_{\theta}$-satisfiable.

\begin{definition}
  \label{def:general_sat}
  An \ATL-formula $\theta$ is \de{generally satisfiable} if $\theta$
  is satisfiable in a CGM $\mmodel{M} = (\agents', S, d, \delta, \ap,
  L)$ for some $\agents'$ with $\agents_{\theta} \subseteq \agents'$;
  $\theta$ is \de{generally valid} if $\theta$ is true in every such
  CGM.
\end{definition}

To see that tight satisfiability (validity) is different from general
satisfiability (validity), consider the formula $\neg \coalnext{1} p
\con \neg \coalnext{1} \neg p$; it is easy to see that this formula is
generally, but not tightly satisfiable (accordingly, its negation is
tightly, but not generally, valid).  Obviously, tight satisfiability
implies general satisfiability, and it is not hard to notice that it
also implies $\agents$-satisfiability (in a model where any agent $a'
\in \agents \setminus \agents_{\theta}$ plays a dummy role by
having exactly one action available at every state).

We now show that testing for both $\agents$-satisfiability and general
satisfiability for $\theta$ can be reduced to testing for tight
satisfiability and a special case of $\agents$-satisfiability where
$\agents = \agents_{\theta} \union \set{a'}$ for some $a' \notin
\agents_{\theta}$ (more precisely, $a' = \card{\agents_{\theta}}+ 1
$)---in other words, only one new agent suffices to witness
satisfiability of $\theta$ over CGFs involving agents not in
$\agents_{\theta}$. This result, proved below, was first stated, with
a proof sketch, for satisfiability in the more restricted (but
equivalent with respect to satisfiability, see \cite{Goranko01})
semantics based on ``alternating transition systems'',
in~\cite{WLWW06}.

\begin{theorem}
  \label{thr:one_fresh_agent_is_enough}
  Let $\theta$ be an \ATL-formula, $\agents_{\theta} \subsetneq
  \agents$, and $a' \notin \agents_{\theta}$.  Then, $\theta$ is
  $\agents$-satisfiable iff $\theta$ is $(\agents_{\theta} \union
  \set{a'})$-satisfiable.
\end{theorem}

\begin{proof}
  Suppose, first, that $\theta$ is $\agents$-satisfiable.  Let
  $\mmodel{M} = (\agents, S, d, \delta, \ap, L)$ be a CGM and $s \in
  S$ be a state such that \sat{M}{s}{\theta}.  To obtain a
  $(\agents_{\theta} \union \set{a'})$-model $\mmodel{M}'$ for
  $\theta$, first, let, for every $s \in S$:
    \begin{itemize}
  \item $d'_a (s) = d_a (s)$ for every $a \in \agents_{\theta}$;
  \item $d'_{a'} (s) = \card{\prod_{b \in (\agents -
        \agents_{\theta})} d_b (s)}$;
  \end{itemize}
  then, define $\delta'$ in the following way:
  $\delta'(\merge{\move{\agents_{\theta}}}{\move{a'}}) =
  \delta(\merge{\sigma_{\agents_{\theta}}}{\sigma_{\agents -
      \agents_{\theta}}}$), where $\sigma_{a'}$ is the place of
  $\sigma_{\agents - \agents_{\theta}}$ in the lexicographic ordering
  of $D_{\agents - \agents_{\theta}} (s)$. Finally, put $\mmodel{M}' =
  (\agents_{\theta} \union \set{a'}, S, d', \delta', \ap, L)$.

  Notice that the above definition immediately implies that $out (s,
  \move{A})$ is the same set in both \mmodel{M} and $\mmodel{M}'$ for
  every $s \in S$ and every $\move{A} \in \moves{A}{s}$ with $A
  \subseteq \agents_{\theta}$, and therefore, in both models,
  $[\coalnext{A}] (X)$ is the same set for every $X \subseteq S$ and
  every $A \subseteq \agents_{\theta}$. It can then be shown, by a
  routine induction on the structure of subformulae $\chi$ of
  $\theta$, using Theorem~\ref{thr:fixpoint}, that \sat{M}{s}{\chi}
  iff \sat{M'}{s}{\chi} for every $s \in S$.

  Suppose, next, that $\theta$ is $(\agents_{\theta} \union
  \set{a'})$-satisfiable.  Let $\mmodel{M}$ be the model witnessing
  the satisfaction and let $b$ be an arbitrary agent in $\agents -
  \agents_{\theta}$.  To obtain a $\agents$-model $\mmodel{M}'$ for
  $\theta$, first, let, for every $s \in S$:
  \begin{itemize}
  \item $d'_a (s) = d_a (s)$ for every $a \in \agents_{\theta}$;
  \item $d'_b (s) = d_{a'} (s)$;
  \item $d'_{b'} (s) = 1$ for any $b' \in \agents \setminus (\set{b}
    \union \agents_{\theta}$);
  \end{itemize}
  then, define $\delta'$ in the following way:
  $\delta'(\merge{\sigma_{\agents_{\theta}}}{\sigma_{\agents -
      \agents_{\theta}}}) =
  \delta(\merge{\sigma_{\agents_{\theta}}}{\sigma_{a'}})$, where
  $\sigma_{a'} = \sigma_{b}$.  Finally, put $\mmodel{M}' = (\agents,
  S, d', \delta', \ap, L)$.  The rest of the argument is identical to
  the one for the opposite direction.
\end{proof}

\begin{corollary}
  \label{cor:general_satisfiability}
  Let $\theta$ be an \ATL-formula. Then, $\theta$ is generally
  satisfiable iff $\theta$ is either tightly satisfiable or
  $(\agents_{\theta} \union \set{a'})$-satisfiable for any $a' \notin \agents_{\theta}$.
\end{corollary}

\begin{proof}
  Straightforward.
\end{proof}

Theorem~\ref{thr:one_fresh_agent_is_enough} and
Corollary~\ref{cor:general_satisfiability} essentially mean that it
suffices to consider two distinct notions of satisfiability for
\ATL-formulae: tight satisfiability and satisfiability in CGMs with
one fresh agent, which we will henceforth refer to as \fm{loose
  satisfiability}.

%%%%%%%%%%%%%%%%%%%%%%%%%%%%%%%%%%
\subsection{Alternative semantic characterization of negated modal
operators} \label{sec:goranko-drimmelen-semantics}
%%%%%%%%%%%%%%%%%%%%%%%%%%%%%%%%%%

Under Definition~\ref{def:atl-satisfaction}, truth conditions for
negated modal operators, such as $\neg \coal{A} \until$, involve
claims about the non-existence of moves or strategies.  In
\cite{GorDrim06}, an alternative semantic characterization of such
formulae has been proposed; this alternative characterization involves
claims about the existence of so-called in \cite{GorDrim06}
\fm{co-moves} and \fm{co-strategies}.

\begin{definition}
  \label{def:co-moves}
  Let $s \in S$ and $A \subseteq \agents$. A \de{co-$A$-move} at state
  $s$ is a function $\comove{A}: D_{A}(s) \mapsto D(s)$ such that
  $\move{A} \sqsubseteq \comove{A} (\move{A})$ for every $\move{A} \in
  \moves{A}{s}$. We denote the set of all co-$A$-moves at $s$ by
  $D^{c}_{A} (s)$.
\end{definition}

Intuitively, given an $A$-move $\move{A} \in \moves{A}{s}$, which
represents a collective action of agents in $A$, a co-$A$-move assigns
to $\move{A}$ a ``countermove'' $\move{\agents \setminus A}$ of the
complement coalition $\agents \setminus A$; taken together, these two
moves produce a unique move vector $\merge{\move{A}}{\move{\agents
    \setminus A}} \in D(s)$.

\begin{definition}
  \label{def:comove_outcome}
  Let $\comove{A} \in D^{c}_{A} (s)$.  The outcome of $\comove{A}$ at
  $s$, denoted by $out(s, \comove{A})$, is the set $\bigunion
  \crh{\delta(s, \comove{A}(\move{A}))}{\move{A} \in
    \moves{A}{s}}$. (Thus, $out(s, \comove{A})$ is the range of
  \comove{A}).

\end{definition}

We next define co-strategies, which are related to co-moves in the
same way as strategies are related to moves.

\begin{definition}
  \label{def:co-A-strategies}
  Let $A \subseteq \agents$ be a coalition and $\gamma$ an ordinal
  such that $1 \leq \gamma \leq \omega$.  A \de{$\gamma$-recall
    co-$A$-strategy} is a mapping $\costr{A}[\gamma] : \bigunion_{1
    \leq n < 1+\gamma} S^n \mapsto \bigunion \crh{\comoves{A}{s}}{s
    \in S}$ such that $\costr{A}[\gamma] (\kappa) \in
  \comoves{A}{l(\kappa)}$ for every $\kappa \in \bigunion_{1 \leq n <
    1+\gamma} S^n$.
\end{definition}

Note that the coalition following a co-$A$-strategy is $\agents
\setminus A$.

\begin{remark}
  Given that $1 + \omega = \omega$, the condition of the
  Definition~\ref{def:co-A-strategies} for the case of $\omega$-recall
  strategies can be rephrased in a simpler form as follows:
  $\costr{A}[\omega] : \bigunion_{1 \leq n < \omega} S^n \mapsto
  \bigunion \crh{\comoves{A}{s}}{s \in S}$ such that
  $\costr{A}[\omega] (\kappa) \in \comoves{A}{l(\kappa)}$ for every
  $\kappa \in \bigunion_{1 \leq n < \omega} S^n$.
\end{remark}

\begin{remark}
  A $\gamma$-recall co-strategy can be defined equivalently as a
  mapping from pairs ($\kappa \in S^n$; $\gamma$-recall strategy
  $\str{A}[\gamma]$) to the set of outcome states $out
  (l(\kappa),\str{A}[\gamma](\kappa))$.
\end{remark}

We will write $\costr{A}$ instead of $\costr{A}[\gamma]$ when $\gamma$
is understood from the context.

\begin{definition}
  Let $\costr{A}[\gamma]$ be a $\gamma$-recall co-$A$-strategy.  If
  $\gamma = \omega$, then $\costr{A}[\gamma]$ is referred to as a
  \de{perfect-recall co-$A$-strategy}; otherwise, $\str{A}[\gamma]$ is
  referred to as a \de{bounded-recall co-$A$-strategy}. Furthermore,
  if $\gamma = 1$, then $\costr{A}[\gamma]$ is referred to as a
  \de{positional co-$A$-strategy}.
\end{definition}

\begin{definition}
  \label{def:outcomes_of_co-strategies}
  Let $\costr{A}[\gamma]$ be a co-$A$-strategy.  The \de{outcome of
    $\costr{A}[\gamma]$ at state $s$}, denoted by $out(s,
  \costr{A}[\gamma])$, is the set of all $s$-runs $\lambda$ such that
 \[
  \begin{array}{lc}
    (\gamma^c) & \lambda[i+1] \in out (\lambda[i],
    \costr{A}(\lambda[j,i]))
    \text{ holds for all } i \geq 0, \\
    & \mbox{ where } j = \max (i - \gamma + 1,0).
  \end{array}\]
\end{definition}

For positional co-strategies, condition $(\mathbf{\gamma^c})$ reduces to
$$
\begin{array}{ll}
  (\mathbf{CP}) & \lambda[i+1] \in out (\lambda[i],
  \costr{A}(\lambda[i])), \text{ for all } i \geq 0,
\end{array}
$$
whereas for perfect-recall co-strategies, it reduces to
$$
\begin{array}{ll}
  (\mathbf{CPR}) & \lambda[i+1] \in out (\lambda[i],
  \costr{A}(\lambda[0, i])),
  \text{ for all } i \geq 0.
\end{array}
$$

Now, we can give alternative truth conditions for negated modalities, couched in terms of co-moves and co-strategies.

\begin{theorem}[Goranko, Drimmelen \cite{GorDrim06}]
  \label{thr:GorDrim}
  Let $\mmodel{M}$ be a CGM and $s \in \mmodel{M}$.  Then,
  \begin{enumerate}
  \item \sat{M}{s}{\neg \coalnext{A} \vp} iff there exists a
    co-$A$-move $\comove{A} \in D^{c}_{A} (s)$ such that
    $\sat{M}{s'}{\neg \vp}$ for every $s' \in out(s, \comove{A})$;
  \item \sat{M}{s}{\neg \coalbox{A} \vp} iff there exists a perfect
    recall co-$A$-strategy \costr{A} such that, for every $\lambda
    \in out(s, \costr{A})$, there exists position $i \geq 0$ with
    \sat{M}{\lambda[i]}{\neg \vp};
  \item \sat{M}{s}{\neg \coal{A} \vp \until \psi} iff there exists a
    perfect recall co-$A$-strategy \costr{A} such that, for every
    $\lambda \in out(s, \costr{A})$ and every position $i \geq 0$
    with $\sat{M}{\lambda[i]}{\psi}$, there exists a position $0 \leq
    j < i$ with $\sat{M}{\lambda[j]}{\neg \vp}$.
  \end{enumerate}
\end{theorem}

\begin{remark}
  Since both types of strategies yield the same semantics for \ATL, in
  the last two clauses of Theorem~\ref{thr:GorDrim}, ``perfect
  recall'' can be replaced with ``positional''.
\end{remark}

%%%%%%%%%%%%%%%%%%%%%%%%%%%%%%%%%%
%%%%%%%%%%%%%%%%%%%%%%%%%%%%%%%%%%
\section{Hintikka structures for \ATL}
\label{sec:atl-hintikka}
%%%%%%%%%%%%%%%%%%%%%%%%%%%%%%%%%%
%%%%%%%%%%%%%%%%%%%%%%%%%%%%%%%%%%

When proving completeness of the tableau procedure described in the
next section, we will make use of a new kind of semantic structures
for \ATL---namely, Hintikka structures.  The basic difference between
models and Hintikka structures is that while models specify the truth
or otherwise of every formula of the language at every state, Hintikka
structures only provide truth values of the formulae relevant to the
evaluation of a fixed formula $\theta$. Before defining Hintikka
structures for \ATL, which we, for the sake of terminological
consistency, call \fm{concurrent game Hintikka structures}, we
introduce, with a view to simplifying the subsequent presentation,
$\alpha$- and $\beta$-notation for \ATL-formulae.

%%%%%%%%%%%%%%%%%%%%%%%%%%%%%%%%%%
\subsection{$\alpha$- and $\beta$-notation for \ATL}
\label{sec:alpha-beta-atl}
%%%%%%%%%%%%%%%%%%%%%%%%%%%%%%%%%%

We divide all \ATL-formulae into primitive and non-primitive ones.

\begin{definition}
  Let $\vp$ be an \ATL-formula. Then, \vp\ is \de{primitive} if it is
  one of the following:
  \begin{itemize}
  \item $\top$;
  \item $p \in \ap$;
  \item $\neg p$ for some $p \in \ap$;
  \item $\coalnext{A} \psi$ for some formula $\psi$;
  \item $\neg \coalnext{A} \psi$ for some formula $\psi$ and $A \ne
    \agents$.
  \end{itemize}
  Otherwise, \vp\ is \de{non-primitive}.
\end{definition}

Intuitively, $\vp$ is primitive if the truth of $\vp$ at a state $s$
of a CGM cannot be reduced to the truth of any ``semantically
simpler'' formulae at $s$; otherwise, $\vp$ is non-primitive.  Note,
in particular, that $\neg p$ is not considered ``semantically
simpler'' then $p$, as the truth of the former can not be reduced to
the truth, as opposed to the falsehood, of the latter.

Following \cite{Smullyan68}, we classify all non-primitive formulae
into $\alpha$-ones and $\beta$-ones.  Intuitively, $\alpha$-formulae
are ``conjunctive'' formulae: an $\alpha$-formula is true at a state
$s$ iff two other formulae, ``conjuncts'' of $\alpha$, denoted by
$\alpha_1$ and $\alpha_2$, are true at $s$.  By contrast,
$\beta$-formulae are ``disjunctive'' formulae, true at a state $s$ iff
either of their ``disjuncts'', denoted by $\beta_1$ and $\beta_2$, is
true at $s$.  For neatness of classification, if the truth of a
non-primitive formula $\psi$ at $s$ can be reduced to the truth of
only \emph{one} simpler formula at $s$, then $\psi$ is treated as an
$\alpha$-formula; for such formulae, $\alpha_1 = \alpha_2$. The
following tables list $\alpha$- and $\beta$-formulae together with
their respective ``conjuncts'' and ``disjuncts''.

\begin{center}
  \begin{tabular}{|c|c|c|}
    \hline
    $\alpha$ & $\alpha_1$ & $\alpha_2$ \\ \hline
    $\neg \neg \vp$ & $\vp$ & $\vp$ \\
    $\neg (\vp \imp \psi)$ & $\vp$ & $\neg \psi$ \\
    $\neg \coalnext{\agents} \vp$ & $\coalnext{\emptyset} \neg \vp$ &
    $\coalnext{\emptyset} \neg \vp$ \\
    $\coalbox{A} \vp$ & $\vp$ & $\coalnext{A} \coalbox{A} \vp$ \\
    \hline
\end{tabular}
\end{center}

\begin{center}
  \begin{tabular}{|c|c|c|}
    \hline
    $\beta$ & $\beta_1$ & $\beta_2$ \\ \hline
    $\vp \imp \psi$ & $\neg \vp$ & $\psi$ \\
    $\coal{A} (\vp \until \psi)$ & $\psi$ & $\vp \con \coalnext{A} \coal{A} (\vp
    \until \psi)$ \\
    $\neg \coal{A} (\vp \until \psi)$ & $\neg \psi \con \neg \vp$ &
    $\neg \psi \con \neg \coalnext{A} \coal{A} (\vp
    \until \psi)$ \\
    $\neg \coalbox{A} \vp$ & $\neg \vp$ & $\neg \coalnext{A}
    \coalbox{A} \vp$
    \\ \hline
\end{tabular}
\end{center}

The entries for the non-modal connectives in the above tables are
motivated by the well-known classical validities.  The entries for the
strategic operators are motivated by
Corollary~\ref{cor:fixpoint}. Lastly, it can be easily checked that
\sat{M}{s}{\neg \coalnext{\agents} \vp} iff
\sat{M}{s}{\coalnext{\emptyset} \neg \vp} for every CGM \mmodel{M} and
$s \in \mmodel{M}$.

%%%%%%%%%%%%%%%%%%%%%%%%%%%%%%%%%%
\subsection{Concurrent game Hintikka structures}
\label{sec:hintikka-atl}
%%%%%%%%%%%%%%%%%%%%%%%%%%%%%%%%%%

We are now ready to define concurrent game Hintikka structures (CGHSs,
for short).  Like concurrent game models, CGHSs are based on
concurrent game frames, where different kinds of strategies may be
used, ranging from positional to perfect-recall.  As it will become
evident from the forthcoming completeness proof, in the case of basic
\ATL, which we primarily focus on in this paper, it suffices to
consider only \emph{positional} Hintikka structures. Nevertheless, we
consider, in this section, the most general case of CGHSs, based on
\emph{perfect-recall} strategies\footnote{Our reason for doing so is
  that we intend to consider, in a follow-up work, adaptations of the
  tableau procedure described herein to some important variations and
  extensions of \ATL, such as $\ATL$ with incomplete information,
  $\ATL^*$, and Game Logic (\cite{AHK02}), where positional strategies
  only do not suffice; then, the results in this section will be put
  to full use.}.

\begin{definition}
  \label{def:cghs}
  A (perfect-recall) \de{concurrent game Hintikka structure} (for
  short, CGHS) is a tuple $\hintikka{H} = (\agents, S, d, \delta, H)$,
  where
  \begin{itemize}
  \item $(\agents, S, d, \delta)$ is a concurrent game frame;
  \item $H$ is a labeling of the elements of $S$ with sets of
    \ATL-formulae that satisfy the following constraints:
    \begin{description}
    \item[H1] If $\neg \vp \in H(s)$, then $\vp \notin H(s)$;
    \item[H2] if $\alpha \in H(s)$, then $\alpha_1 \in H(s)$ and
      $\alpha_2 \in H(s)$;
    \item[H3] if $\beta \in H(s)$, then $\beta_1 \in H(s)$ or $\beta_2
      \in H(s)$;
    \item[H4] if $\coalnext{A} \vp \in H(s)$, then there exists an
      $A$-move $\move{A} \in \moves{A}{s}$ such that $\vp \in H(s')$
      for all $s' \in out(s, \move{A})$;
    \item[H5] if $\neg \coalnext{A} \vp \in H(s)$, then there exists a
      co-$A$-move $\comove{A} \in D^c_A (s)$ such that $\neg \vp \in
      H(s')$ for all $s' \in out(s, \comove{A})$;
    \item[H6] if $\coal{A} \vp \until \psi \in H(s)$, then there
      exists a perfect-recall $A$-strategy \str{A} such that, for
      all $\lambda \in out(s, \str{A})$, there exists a position $i
      \geq 0$ such that $\psi \in H(\lambda[i]$) and $\vp \in
      H(\lambda[j])$ holds for all positions $0 \leq j < i$;
    \item[H7] if $\neg \coalbox{A} \vp \in H(s)$, then there exists a
      perfect-recall co-$A$-strategy \costr{A} such that, for every
      $\lambda \in out(s, \costr{A})$, there exists position $i
      \geq 0$ with $\neg \vp \in H(\lambda[i])$.
    \end{description}
  \end{itemize}
\end{definition}

\begin{remark}
  To obtain the definition of positional CGHS, all one has to do is
  replace ``perfect-recall'' with ``positional'' in clauses (H6) and
  (H7) of Definition~\ref{def:cghs}.
\end{remark}

\begin{definition}
  Let $\theta$ be an \ATL-formula and $\hintikka{H} = (\agents, S, d,
  \delta, H)$ be a CGHS.  We say that \hintikka{H} is a \de{concurrent
    game Hintikka structure for $\theta$} if $\theta \in H(s)$ for
  some $s \in S$.
\end{definition}

Hintikka structures can be thought of as representing a class of
models on the set of states $S$ that, for every $s \in S$, agree on
the formulae in $H(s)$ (that is, make exactly the same formulae in
$H(s)$ true). Models themselves can be thought of as \emph{maximal}
Hintikka structures, whose states are labeled with maximally
consistent sets of formulae.  More precisely, given a CGM $\mmodel{M}
= (\agents, S, d, \delta, \ap, L)$, we can define the extended
labeling function $L^+_{\mmodel{M}}$ by $L^+_{\mmodel{M}} (s) =
\crh{\vp}{\sat{M}{s}{\vp}}$, where $\vp$ ranges over all
\ATL-formulae, and the resulting structure $(\agents, S, d, \delta,
L^+_{\mmodel{M}})$ will be a Hintikka structure. This immediately
gives rise to the following theorem.

\begin{theorem}
  \label{from_models_to_hintikka}
  Let $\theta$ be an \ATL-formula.  Every CGM $\mmodel{M} = (\agents,
  S, d, \delta, \ap, L)$ satisfying $\theta$ induces a CGHS
  $\hintikka{H} = (\agents, S, d, \delta, L^+_{\mmodel{M}})$ for
  $\theta$, where $L^+_{\mmodel{M}}$ is the extended labeling function
  on \mmodel{M}.
\end{theorem}

\begin{proof}
  Straightforward, using Theorem~\ref{thr:GorDrim} for (H5) and (H7).
\end{proof}

Conversely, every Hintikka structure for a formula $\theta$ can be
expanded to a maximal one---that is, a model---by declaring, for every
$s \in S$, all atomic propositions outside $H(s)$ to be false at $s$.
To prove this claim, however, we need a few auxiliary definitions.

\begin{definition}
  Let $\hintikka{H} = (\agents, S, d, \delta, H)$ be a CGHS.  A
  \de{run} of length $m$, where $1 \leq m < \omega$, in \hintikka{H}
  is a sequence $\lambda = s_0, \ldots, s_{m-1}$ of elements of $S$
  such that, for all $0 \leq i < m - 1$, the state $s_{i+1}$ is a
  successor of the state $s_i$.  Numbers $0$ through $m-1$ are called
  \de{positions} of $\lambda$. The length of $\lambda$, defined as the
  number of positions in $\lambda$, is denoted by $|\lambda|$.  For
  each position $0 \leq i < m$, we denote by $\lambda[i]$ the $i$th
  state of $\lambda$.  A \de{finite run} in \hintikka{H} is a run of
  length $m$ for some $m$ with $1 \leq m < \omega$. A finite run with
  $\lambda[0] = s$ is a \de{finite $s$-run}.
\end{definition}

\begin{definition}
  \label{def:compliant-finite-s-runs}
  Let \hintikka{H} be a CGHS, $\lambda$ be a finite
  $s$-run in \hintikka{H}, and $\costr{A}[m]$ be an $m$-recall
  co-$A$-strategy on the frame of \hintikka{H}, where $1 \leq m <
  \omega$.  We say that $\lambda$ is \de{compliant with
    $\costr{A}[m]$} if
  \begin{itemize}
  \item $|\lambda| = m + 1$;
  \item $\lambda[i+1] \in out (\lambda[i],
    \costr{A}[m](\lambda[0,i]))$ holds for all $0 \leq i < m$.
  \end{itemize}
\end{definition}

\begin{definition}
  Let \hintikka{H} be a CGHS, let $\lambda$ be an (infinite) $s$-run
  in \hintikka{H} and let \costr{A} be a perfect-recall
  co-$A$-strategy on the frame of \hintikka{H}. We say that $\lambda$
  is \de{compliant with $\costr{A}$} if $\lambda \in out (s,
  \costr{A})$.
\end{definition}

\begin{theorem}
  \label{thr:from_hintikka_to_models}
  Let $\theta$ be an \ATL-formula.  Every CGHS $\hintikka{H} =
  (\agents, S, d, \delta, H)$ for $\theta$ can be expanded to a CGM
  satisfying $\theta$.
\end{theorem}

\begin{proof}
  Let $\hintikka{H} = (\agents, S, d, \delta, H)$ be a CGHS for
  $\theta$.  To obtain a CGM $\mmodel{M} = (\agents, S, d, \delta,
  \ap, L)$, we define the labeling function $L$ as follows: $L(s) =
  H(s) \inter \ap$, for every $s \in S$.

  To establish the statement of the theorem, we prove, by induction on
  the structure of formula $\chi$ that, for every $s \in S$ and every
  $\chi$, the following claim holds:
  \begin{displaymath}
    \chi \in H(s) \text{ implies } \sat{M}{s}{\chi} \textit{ and } \neg
    \chi \in H(s) \text{ implies } \sat{M}{s}{\neg \chi}.
  \end{displaymath}
  Let $\chi$ be some $p \in \ap$.  Then, $p \in H(s)$ implies $p \in
  L(s)$ and, thus, \sat{M}{s}{p}; if, on the other hand, $\neg p \in
  H(s)$, then due to (H1), $p \notin H(s)$ and thus $p \notin L(s)$;
  hence, \sat{M}{s}{\neg p}.

  Assume that the claim holds for all subformulae of $\chi$; then, we have
  to prove that it holds for $\chi$, as well.

  Suppose that $\chi$ is $\neg \vp$.  If $\neg \vp \in H(s)$, then the
  inductive hypothesis immediately gives us $\sat{M}{s}{\neg \vp}$;
  if, on the other hand, $\neg \neg \vp \in H(s)$, then by virtue of
  (H2), $\vp \in H(s)$ and hence, by inductive hypothesis,
  $\sat{M}{s}{\vp}$ and thus $\sat{M}{s}{\neg \neg \vp}$.

  The cases of $\chi = \vp \imp \psi$ and $\chi = \coalnext{A} \psi$
  and are straightforward, using (H2)--(H5).

  Suppose that $\chi = \coal{A} \vp \until \psi$.  If $\coal{A} \vp
  \until \psi \in H(s)$, then the desired conclusion immediately
  follows from (H6) and the inductive hypothesis.

  Assume now that $\neg \coal{A} \vp \until \psi \in H(s)$.  In view
  of the inductive hypothesis and Theorem~\ref{thr:GorDrim}, it
  suffices to show that there exists a perfect-recall co-$A$-strategy
  $\costr{A}$ such that $\lambda \in out(s, \costr{A})$ implies that,
  if there exists $i \geq 0$ with $\psi \in H(\lambda[i])$, then there
  exists $0 \leq j < i$ with $\neg \vp \in H(\lambda[j])$.

  We define the required $\costr{A}$ by induction on the length of
  sequences in its domain.  This amounts to defining finite prefixes
  of $\costr{A}$ for every $1 \leq n < \omega$---the restrictions of
  $\costr{A}$ to sequences of states of length $\leq n$.  Clearly, the
  finite prefix of $\costr{A}$ of length $n$ is an $n$-recall
  co-$A$-strategy.  We only explicitly define the value of
  $\costr{A}[n](\lambda)$, where $|\lambda| = n$, if $\lambda$ is a
  finite $s$-run compliant with $\costr{A}[n-1]$ (recall
  Definition~\ref{def:compliant-finite-s-runs}), where
  $\costr{A}[n-1]$ is a strategy defined at the previous step of the
  induction.  The values of $\costr{A}[n](\lambda)$ for any other
  sequences of length $n$ are immaterial. The only other constraint
  that we have to take into account when defining $\costr{A}[n]$ is
  that, if $\costr{A}[n]$ extends $\costr{A}[m]$, then the values of
  $\costr{A}[m]$ and $\costr{A}[n]$ should agree on all the sequences
  of length $m$.  Alongside defining $\costr{A}[n]$ for every $1 \leq
  n < \omega$, we prove that the following invariant property holds:
  If $\lambda \in out(s, \costr{A}[n])$, then
  \begin{center}
    \(
    \begin{array}{lll}

      &(i) & \text{\emph{Either} there exists a  position } 0 \leq i \leq
      n, \text{ such that } \\
      (\dag) \hspace{1cm} & & \neg \vp \in H(\lambda[i])
      \text{ and } \neg \psi \in H(\lambda[j])
      \text{ for all } 0 \leq j \leq i, \\
      & (ii) & \textit {or } \neg  \psi,
      \neg \coalnext{A} \coal{A} \vp \until \psi \in H(\lambda[i])
      \text{ for all } 0 \leq i \leq n.
    \end{array}
    \)
  \end{center}
  Clearly, if every finite prefix of $\costr{A}$ satisfies (\dag),
  $\costr{A}$ is the required co-$A$-strategy.

  We start by defining $\costr{A}[1]$.  There is only one $s$-run of
  length $1$, namely $(s)$.  As $\neg \coal{A} \vp \until \psi \in
  H(s)$, in view of (H3) and (H2), either $\neg \psi, \neg \vp \in
  H(s)$ or $\neg \psi, \neg \coalnext{A} \coal{A} \vp \until \psi \in
  H(s)$. In the former case any co-$A$-move will produce a
  co-$A$-strategy $\costr{A}[1]$ such that, if $\lambda \in out(s,
  \costr{A}[1])$, then $\lambda$ satisfies (\dag) (i). In the latter
  case, (H5) guarantees that there exists a co-$A$-move $\comove{A}
  \in D^c_A (s)$ such that $\neg \coal{A} \vp \until \psi \in H(s')$
  for all $s' \in out(s, \comove{A})$.  This, together with (H3) and
  (H2) guarantees that $\neg \psi, \neg \coalnext{A} \coal{A} \vp
  \until \psi \in H(s')$ for every $s' \in out(s, \comove{A})$, which,
  as $\neg \psi \in H(s)$, ensures that (\dag) (ii) holds for any
  $\lambda \in out(s, \costr{A}[1])$.  Thus, in either case, (\dag)
  holds for every $\lambda \in out(s, \costr{A}[1])$.

  Next, inductively assume that, if $\lambda$ is an $s$-run compliant
  with $\costr{A}[n]$, then (\dag) holds for $\lambda$.  We need to
  show how to extend $\costr{A}[n]$ to $\costr{A}[n+1] \supset
  \costr{A}[n]$ in the (\dag)-preserving way.  If (\dag) (i) holds for
  every $\lambda$ satisfying the condition of the inductive
  hypothesis, then obviously, any co-$A$-move will do.  Otherwise,
  (\dag) (ii) holds for every such $\lambda$; then, $\costr{A}[n+1]$
  can be obtained from $\costr{A}[n]$ as in the second part of the
  ``basis case'' argument.  For all other sequences $\kappa$ of length
  $n+1$ (i.e., those that do not start with $s$ or are not compliant
  with $\costr{A}[n]$), the value $\costr{A}[n](\kappa)$ can be
  defined arbitrarily.  For all sequences $\kappa$ of length $\leq n$,
  we stipulate $\costr{A}[n+1](\kappa) = \costr{A}[n](\kappa)$.  This
  completes the definition of $\costr{A}[n+1]$.  As we have seen, if
  $\lambda$ is an $s$-run compliant with $\costr{A}[n+1]$, then (\dag)
  holds for $\lambda$.

  The case of $\neg \coalbox{A} \vp \in H(s)$ is straightforward using
  (H7), while the case of $\coalbox{A} \vp \in H(s)$ can be proved in
  a way analogous to the case of $\neg \coal{A} \vp \until \psi$,
  using suitable definitions of compliancy of (finite and infinite)
  runs with strategies.
\end{proof}

Theorems~\ref{from_models_to_hintikka} and
\ref{thr:from_hintikka_to_models} taken together mean that, from the
point of view of a single \ATL-formula, satisfiability in a
(perfect-recall) model and in a (perfect-recall) Hintikka structure
are equivalent.

%%%%%%%%%%%%%%%%%%%%%%%%%%%%%%%%%%
%%%%%%%%%%%%%%%%%%%%%%%%%%%%%%%%%%
\section{Terminating tableaux for tight \ATL-satisfiability}
\label{sec:atl-tableau-procedure}
%%%%%%%%%%%%%%%%%%%%%%%%%%%%%%%%%%
%%%%%%%%%%%%%%%%%%%%%%%%%%%%%%%%%%

In the current section, we present a tableau method for testing
\ATL-formulae for tight satisfiability.

Traditionally, tableau techniques work by decomposing the formula
whose satisfiability is being tested into ``semantically simpler''
formulae. In the classical propositional case (\cite{Smullyan68}),
``semantically simpler'' implies ``smaller'', which by itself
guarantees termination of the procedure in a finite number of
steps. Another feature of the tableau method for the classical
propositional logic is that this decomposition into semantically
simpler formulae results in a tree representing an exhaustive search
for a model---or, to be more precise, a Hintikka set (the classical
analogue of Hintikka structures)---for the input formula.  If at least
one branch of the tree produces a Hintikka set for the input formula,
the search has succeeded and the formula is pronounced
satisfiable\footnote{Even though this tree is usually built step-by-step by decomposing one formula at a time (see~\cite{Smullyan68} and
  \cite{Wolper85}), it can be compressed into a simple
  tree---i.e., a tree with a single interior node---whose root is the
  set containing only the input formula and whose leaves are
  all minimal downward-saturated extensions (to be defined later
  on; see Definitions~\ref{def:downward-saturated} and
  \ref{def:min-downward-saturated-extension}) of the root. We will use
  this, more compact, approach in our tableaux.}.

These two defining features of the classical tableau method do not
emerge unscathed when the method is applied to logics containing
fixpoint operators, such as \ATL\ (in this respect, the case of \ATL\
is similar to those of \LTL\ and \CTL).

Firstly, decomposition of \ATL-formulae into ``semantically simpler''
ones, which, just as in the classical case, is carried out by breaking
up $\alpha$- and $\beta$-formulae into their respective ``conjuncts''
and ``disjuncts,'' does not always produce smaller formulae, as can be
seen from the tables given in section \ref{sec:alpha-beta-atl}.
Therefore, we will have to take special precautions to ensure that the
procedure terminates (in our case, as in~\cite{Wolper85}, this will
involve the use of the so-called \emph{prestates}).

Secondly, in the classical case the only reason why it might turn out
to be impossible to produce a Hintikka set for the input formula is
that every attempt to build such a set results in a collection of
formulae containing a patent inconsistency (from here on, by
\de{patent inconsistency} we mean a pair of formulas of the form
$\vp$, $\neg \vp$)\footnote{Notice that this condition implies but is
  not, in general, equivalent to propositional inconsistency.}.  In the
case of \ATL, there are two other reasons for a tableau not to
correspond to any Hintikka structure for the input formula. First,
applying decomposition rules to eventualities---formulae whose truth
conditions require that some formula ($\psi$ in the case of the
eventuality $\coal{A} \vp \until \psi$, and $\neg \vp$ in the case of
the eventuality $\neg \coalbox{A} \vp$) ``eventually'' becomes true;
the tableau analog of this we will refer to as \fm{realization of an
  eventuality},---one can indefinitely postpone their realization by
always choosing the ``disjunct'' (notice that both eventualities are
$\beta$-formulas) ``promising'' that the realization will happen
further down the line, this ``promise'' never being
fulfilled. Therefore, in addition to not containing patent
inconsistencies, ``good'' \ATL\ tableaux should not contain sets with
unrealized eventualities.  Yet another reason for the resultant
tableau not to represent a Hintikka structure is that some sets do not
have all the successors they would be required to have in a
corresponding Hintikka structure.

Coming back to the realization of eventualities, it should be noted
that, in a Hintikka structure for the input formula, all the
eventualities belonging to the labels of its states have to be
realized, and different eventualities can place different demands on
the labels of states of a Hintikka structure. Fortunately, in the case
of \ATL\ (just like in those of \LTL\ and \CTL\ and unlike, for
example, those of Parikh's game logic~\cite{PaulyPar03} and
propositional $\mu$-calculus~\cite{BradStir07}), eventualities can be
``taken on'' one at a time: we can ensure, and this lies at the heart
of our completeness proof, that having realized eventualities one by
one, we can then assemble a Hintikka structure out of the ``building
blocks'' realizing single eventualities.  This technique resembles the
mosaic method used to prove decidability of a variety of modal and
temporal logics (see, for example, \cite{MMR00}).

%%%%%%%%%%%%%%%%%%%%%%%%%%%%%%%%%%
\subsection{Brief description of the tableau procedure}
\label{sec:procedure}
%%%%%%%%%%%%%%%%%%%%%%%%%%%%%%%%%%

In essence, the tableau procedure for testing an \ATL-formula $\theta$
for satisfiability is an attempt to construct a non-empty graph
$\tableau{T}^{\theta}$, called a \fm{tableau}, representing all
possible concurrent game Hintikka structures for $\theta$.  If the
attempt is successful, $\theta$ is pronounced satisfiable; otherwise,
it is declared unsatisfiable. (As this whole section is exclusively
concerned with tight satisfiability, whenever we use the word
``satisfiable'' or any derivative thereof, we mean the tight variety;
another reason to keep the language generic is that---as we shall see
later on---the basic ideas transfer smoothly over to other species of
satisfiability).

The tableau procedure consists of three major phases: \fm{construction
  phase}, \fm{prestate elimination phase}, and \fm{state elimination
  phase}.  Accordingly, we have three types of tableau rules:
construction rules, a prestate elimination rule, and state elimination
rules.  The procedure itself essentially specifies---apart from the
starting point of the whole process---in what order and under what
circumstances these rules should be applied.

During the construction phase, the construction rules are used to
produce a directed graph $\tableau{P}^{\theta}$---referred to as the
\emph{pretableau} for $\theta$---whose set of nodes properly contains
the set of nodes of the tableau $\tableau{T}^{\theta}$ that we are
ultimately building.  Nodes of $\tableau{P}^{\theta}$ are sets of
\ATL-formulae, some of which---referred to as
\fm{states}\footnote{From here on, the term ``state'' is used in two
  different meanings: as ``state'' of the (pre)tableaux---which is a
  set of \ATL-formulas satisfying certain conditions, to be stated
  shortly, ---and as a ``state'' of a semantic structure (frame,
  model, or Hintikka structure).  Usually, the context will determine
  explicitly which of these we mean; when the context leaves room for
  ambiguity, we will explicitly mention what kind of states we mean.}
--- are meant to represent states (whence the name) of a Hintikka
structure, while others---referred to as \fm{prestates}---fulfill a
purely technical role of helping to keep $\tableau{P}^{\theta}$
finite.  During the prestate elimination phase, we create a smaller
graph $\tableau{T}_0^{\theta}$ out of
$\tableau{P}^{\theta}$---referred to as the \fm{initial tableau for
  $\theta$}---by eliminating all the prestates of
$\tableau{P}^{\theta}$ (and tweaking with its edges) since prestates
have already fulfilled their function: as we are not going to add any
more nodes to the graph built so far, the possibility of producing an
infinite structure is no longer a concern.  Lastly, during the state
elimination phase, we remove from $\tableau{T}_0^{\theta}$ all the
states, if any, that cannot be satisfied in any CGHS, for one of the
following three reasons: either they are inconsistent, or contain
unrealizable eventualities, or do not have all the successors needed
for their satisfaction.  This results in a (possibly empty) subgraph
$\tableau{T}^{\theta}$ of $\tableau{T}_0^{\theta}$, called the
\de{final tableau for $\theta$}.  Then, if we have some state $\Delta$
in $\tableau{T}^{\theta}$ containing $\theta$, we pronounce $\theta$
satisfiable; otherwise, we declare $\theta$ unsatisfiable.

%%%%%%%%%%%%%%%%%%%%%%%%%%%%%%%%%%
\subsection{Construction phase} \label{sec:atl-construction_rules}
%%%%%%%%%%%%%%%%%%%%%%%%%%%%%%%%%%

As already mentioned, at the construction phase, we build the
pretableau $\tableau{P}^{\theta}$ --- a directed graph whose nodes are
sets of \ATL-formulae, coming in two varieties: \fm{states} and
\fm{prestates}.  Intuitively, states are meant to represent states of
CGHSs, while prestates are ``embryo states'', which will in the course
of the construction be ``unwound'' into states.  Technically, states
are downward saturated, while prestates do not have to be so.

\begin{definition}
  \label{def:downward-saturated}
  Let $\D$ be a set of \ATL-formulae.  We say that $\D$ is
  \de{downward saturated} if the following conditions are satisfied:
  \begin{itemize}
  \item if $\alpha \in \D$, then $\alpha_1 \in \D$ and $\alpha_2 \in
    \D$;
  \item if $\beta \in \D$, then $\beta_1 \in \D$ or $\beta_2 \in \D$.
  \end{itemize}
\end{definition}

Moreover, $\tableau{P}^{\theta}$ will contain two types of edge.  As
has already been mentioned, tableau techniques usually work by setting
in motion an exhaustive search for a Hintikka structure for the input
formula; one type of edge, depicted by unmarked double arrows
$\brancharrow$, will represent this exhaustive search dimension of our
tableaux.  Exhaustive search looks for all possible alternatives, and
in our tableaux the alternatives will arise when we unwind prestates
into states; thus, when we draw an unmarked arrow from a prestate \G\
to states $\D$ and $\D'$ (depicted as $\G \brancharrow \D$ and $\G
\brancharrow \D'$, respectively), this intuitively means that, in any
CGHS, a state satisfying \G\ has to satisfy at least one of $\D$ and
$\D'$.

Another type of edge represents transitions in CGHSs effected by move
vectors.  Accordingly, this type of edge will be represented in
pretableaux by single arrows marked with
$\card{\agents_{\theta}}$-tuples $\move{}$ of numbers, each number
intuitively representing an $a$-move for some $a \in
\agents_{\theta}$.  Intuitively, we think of these
$\card{\agents_{\theta}}$-tuples as move vectors.  Thus, if we draw an
arrow marked by $\move{}$ from a state $\D$ to a prestate $\G$
(depicted as $\D \movearrow \G$), this intuitively means that, in any
CGHS represented by the tableau we are building, from a state
satisfying $\D$ we can move along $\move{}$ to a state satisfying
$\G$.

It should be noted that, in the pretableau, we never create in one go
full-fledged successors for states, which is to say we never draw a
marked arrow from state to state; such arrows always go from states to
prestates.  On the other hand, unmarked arrows connect prestates to
states.  Thus, the whole construction of the pretableau alternates
between going from prestates to states along edges represented by
double unmarked arrows and going from states to prestates along the
edges represented by single arrows marked by ``move vectors''. This
cycle has, however, to start somewhere.

The tableau procedure for testing satisfiability of $\theta$ starts
off with the creation of a single prestate $\set{\theta}$. Thereafter,
a pair of construction rules are applied to the part of the pretableau
created thus far: one of the rules, \Rule{SR}, specifies how to unwind
prestates into states; the other, \Rule{Next},---how to obtain
``successor'' prestates from states.  To state \Rule{SR}, we need the
following definition.

\begin{definition}
  \label{def:min-downward-saturated-extension}
  Let $\G$ and $\D$ be sets of \ATL-formulae.  We say that $\D$ is a
  \de{minimal downward saturated extension of $\G$} if the following
  holds:
  \begin{itemize}
  \item $\G \subseteq \D$;
  \item $\D$ is downward saturated;
  \item there is no downward saturated set $\D'$ such that $\G
    \subseteq \D' \subset \D$.
  \end{itemize}
\end{definition}

Note that $\G$ can be a minimal downward saturated extension of
itself.

We now state the first construction rule.

\bigskip

\Rule{SR} Given a prestate $\G$, do the following:
\begin{enumerate}
\item add to the pretableau all the minimal downward saturated
  extensions $\D$ of $\G$ as \de{states};

\item for each of the so obtained states $\D$, if $\D$ does not contain
  any formulae of the form $\coalnext{A} \vp$ or $\neg \coalnext{A}
  \vp$, add the formula $\coalnext{\agents_{\theta}} \truth$ to $\D$;

\item for each state $\D$ obtained at steps 1 and 2, put $\G
  \brancharrow \D$;

\item if, however, the pretableau already contains a state $\D'$ that
  coincides with $\D$, do not create another copy of $\D'$, but only
  put $\G \brancharrow \D'$.
\end{enumerate}

We denote the finite set of states that have outgoing edges from a
prestate $\G$ by $\st{\G}$.  These include genuinely ``new'' states
created by applying of \Rule{SR} to $\G$ as well as the states that
had already been in the pretableau and got identified with a state
that would otherwise have been created by applying \Rule{SR} to $\G$.

\begin{example}
  \label{ex:sr-rule}
  As a running example illustrating our tableau procedure, we will be
  constructing a tableau for the formula $\theta_1 = \neg \coalbox{1}
  p \con \coalnext{1,2} p \con \neg \coalnext{2} \neg p$.  The
  construction of the tableau for this formula starts off with the
  creation of a prestate $\G_1 = \set{\neg \coalbox{1} p \con
    \coalnext{1,2} p \con \neg \coalnext{2} \neg p}$.  Next, \Rule{SR}
  is applied to $\G_1$, which produces two states, which we call, for
  future reference, $\D_1$ and $\D_2$ (in the diagram below, as well
  as in the following examples, we omit the customary set-theoretic
  curly brackets around states and prestates of the (pre)tableaux):

\medskip

  \begin{picture}(200,50)(0,450)
    \footnotesize
    \thicklines

    \put(67,487.5){$(\G_1)$}

    \put(165,490){\makebox(0,0){
        $\begin{array}{c}
          \neg \coalbox{1} p \con
          \coalnext{1,2} p \con \neg \coalnext{2} \neg p = \theta_1
        \end{array}$
      }}
    \put(120,483){\line(-1,-1){10}}
    \put(122,483){\line(-1,-1){10}}
    \put(113,475){\vector(-1,-1){5}}

    \put(207,483){\line(1,-1){10}}
    \put(209,483){\line(1,-1){10}}
    \put(216.3,475){\vector(1,-1){5}}

    % node 2; state
    \put(23,457.5){$(\D_1)$}

    \put(95,458){\makebox(0,0){
        $\begin{array}{c}
          \theta_1, \neg \coalbox{1} p, \coalnext{1,2} p, \\
          \neg \coalnext{2} \neg p, \neg \coalnext{1} \coalbox{1} p
        \end{array}$
      }}

    % node 3; state
    \put(162,457.5){$(\D_2)$}

    \put(230,458){\makebox(0,0){
        $\begin{array}{c}
          \theta_1, \neg \coalbox{1} p, \coalnext{1,2} p, \\
          \neg \coalnext{2} \neg p, \neg p
        \end{array}$
      }}
  \end{picture}
\end{example}

In general, if at least one subformula of a non-primitive member of a
prestate $\G$ is a $\beta$-formula, $\G$ will have more than one
minimal downward saturated extension; hence, for such a $\G$, the set
\st{\G} will contain more than one state.  The only exception to this
general rule may occur when we come across $\beta$-formulae for which
$\beta_1 = \beta_2$, such as $(\vp \imp \neg \vp)$.

\medskip

We now turn to our second construction rule, \Rule{Next}, which
creates ``successor'' prestates from states. The rule has to ensure
that a sufficient supply of successor prestates is created to enforce
the truth of all ``next-time formulae'' (see below) at the current
state. Unlike the case of logics whose models are sets of states
connected by edges of binary relations, such as \LTL\ and \CTL, in
\ATL\ successor prestates cannot be created by simply removing the
``next-time'' modality from a formula and creating an edge associated
with that formula.  On the contrary, in \ATL, transitions are effected
by move vectors, with which we, then, associate formulae made true by
actions of agents making up that particular move vector. Thus, the
rule \Rule{Next} needs to provide each agent mentioned in the input
formula with a sufficient number of actions available at the current
state, and then ``populate'' prestates associated each resultant move
vector $\sigma$ with appropriate formulae.

Before formally introducing the rule, we provide some intuition behind
it.  The rule is applicable to a state, say $\D$; more precisely, it
is applicable to the formulae of the form $\coalnext{A} \vp$---which
we refer to as \emph{positive next-time formulae}---and $\neg
\coalnext{A} \psi$, where $A \ne \agents$---which we refer to as
\emph{proper negative next-time formulae}---belonging to
$\D$. Positive and proper negative next-time formulae are referred to
collectively as \emph{next-time formulae}. These formulae are arranged
in a list $\mathbb{L}$ and, thus, numbered; all the positive formulae
in $\mathbb{L}$ precede all the negative ones; otherwise, the ordering
is immaterial. The agents mentioned in the input formula $\theta$ can
be thought of as having to decide which formulae from $\D$ appearing
under the ``next-time'' coalition modalities $\coalnext{\ldots}$ and
$\neg \coalnext{\ldots}$ should be included into a successor prestate
associated with each move vector $\sigma$ (inclusion into a prestate
intuitively corresponds to satisfiability in the successor states of a
Hintikka structure, as prestates eventually get unwound into tableau
states).  Therefore, the number of ``actions'' each agent mentioned in
$\theta$ is given at $\D$ equals the number of the next-time formulae
in $\D$ (= length of $\mathbb{L}$).  These actions are combined into
``move vectors'' $\sigma$ leading to successor prestates.  The
inclusion of formulae into the prestate $\G_{\sigma}$ created as a
successor of $\D$ by an arrow labeled with $\sigma$ is then decided as
follows.  A formula $\vp$ for which $\coalnext{A} \vp \in \mathbb{L}$
is included into $\G_{\sigma}$, if every agent in $A$ ``votes'' in
$\sigma$ for this formula (i.e. every $i$th slot in $\sigma$ with $i
\in A$ contains the number representing the position of $\coalnext{A}
\vp$ in $\mathbb{L}$).  On the other hand, $\neg \psi$ for which $\neg
\coalnext{A} \psi \in \mathbb{L}$ is included into $\G_{\sigma}$ (for
technical reasons, at most one such formula can be included into any
prestate) if every agent \emph{not in} $A$ votes, in the sense
explained above for the positive case, for a negative formula from
$\mathbb{L}$ (not necessarily $\neg \coalnext{A} \psi$) and, moreover,
$\neg \coalnext{A} \psi$ is the formula decided on by the
\emph{collective} (negative) vote of agents in $\agents \setminus A$.
Technically, this collective vote is represented by the number
$\mathbf{neg}(\move{})$, which is computed using all negative votes of
$\sigma$, which allows it to represent a truly collective decision.

We now turn to the technical presentation of \Rule{Next}.  The rule
does not apply to the states containing patent inconsistencies since
such states, obviously, cannot be part of any CGHS (so, we are not
wasting time creating ``junk'' that will have to be removed anyway).

\bigskip

\Rule{Next} Given a state $\D$ such that for no $\chi$ we have $\chi,
\neg \chi \in \D$, do the following:
\begin{enumerate}
\item Order linearly all positive and proper negative next-time formulae of
  $\D$ in such a way that all the positive next-time formulae precede all
  the negative ones; suppose the result is the list
   \[\mathbb{L} = \coalnext{A_0} \vp_0, \ldots, \coalnext{A_{m-1}}
   \vp_{m-1},\lnot \coalnext{A'_0} \psi'_0, \ldots, \lnot
   \coalnext{A'_{l-1}} \psi_{l-1}.\] (Note that, due to step 2 of
   \Rule{SR}, $\mathbb{L}$ is always non-empty.) Let $r_{\D} = m + l$;
   denote by $D (\D)$ the set $\set{0,
     \ldots,r_{\D}-1}^{\card{\agents_{\theta}}}$; lastly, for every
   $\move{} \in D(\D)$, denote by $N(\move{})$ the set
   \crh{i}{\move{i} \geq m}, where \move{i} stands for the $i$th
   component of the tuple $\sigma$, and by $\mathbf{neg}(\move{})$ the
   number $[\sum_{i \in N(\move{})} (\move{i} - m)] \mod l$.

 \item Consider the elements of $D(\D)$ in the lexicographic order and
   for each $\sigma \in D(\D)$ do the following:
  \begin{enumerate}
  \item Create a prestate
     \begin{eqnarray*}
       \G_{\move{}} & = & \crh{\vp_p}{\coalnext{A_p} \vp_p \in \D
         \text{ and }  \move{a} = p \text{ for all } a \in A_p}\\
       & \union & \crh{\neg \psi_q}{\neg \coalnext{A'_q}
         \psi_q \in \D, \ \mathbf{neg}(\move{}) = q, \text { and }
         \agents_{\theta} - A'_q \subseteq N(\move{})};
     \end{eqnarray*}

     put $\G_{\move{}} := \{\top\}$ if the sets on both sides of the
     union sign above are empty.

   \item Connect $\D$ to $\G_{\move{}}$ with $\movearrow$;
   \end{enumerate}
   If, however, $\G_{\move{}} = \G$ for some prestate $\G$ that has
   already been added to the pretableau, only connect $\D$ to $\G$
   with $\movearrow$.
 \end{enumerate}

 We denote the finite set of prestates \crh{\G}{\D \movearrow \G\
   \text{ for some } \move{} \in D(\Delta) } by $\sucpr{\D}$.  Note
 that a state $\D$ may get connected to some $\G \in \sucpr{\D}$ by
 arrows labeled by distinct $\move{}, \move{}' \in \moves{}{\D}$.  In
 such cases, we ``glue together'' arrows labeled by $\sigma$ and
 $\sigma'$, in effect creating an arrow marked by a set of labels
 rather than a label (in examples below, in such cases, we attach
 several labels to a single arrow).

\setcounter{example}{0}

\begin{example}[continued]
  \label{ex:next_rule}
  Let us apply the \Rule{Next} rule to the state \linebreak $\D_1 =
  \set{\theta_1, \neg \coalbox{1} p, \coalnext{1,2} p, \neg
    \coalnext{2} \neg p, \neg \coalnext{1} \coalbox{1} p}$ from our
  running example.  We arrange all the positive and proper negative
  next-time formulae of this state in the list $\mathbb{L} =
  \coalnext{1,2} p, \neg \coalnext{2} \neg p, \neg \coalnext{1}
  \coalbox{1} p$.  Then, at $\D_1$, each of the two agents from
  $\theta_1$ is going to have 3 actions, denoted by numbers 0, 1, and
  2.  To decide what formulae are to be included in the prestates
  resulting from tuples of those actions, we also need to separately
  number all the negative next-time formulae from $\mathbb{L}$: $\neg
  \coalnext{2} \neg p$ will be numbered $0$, while $\neg \coalnext{1}
  \coalbox{1} p$ will be numbered $1$ ($\mathbf{neg}(\move{})$ in the
  table below will refer to these numbers).  The following table
  illustrates which formulae are included into prestates associated
  with what move vectors at $\D$:

  \begin{center}
    \footnotesize
  $\begin{array}{|c|c|l|}
    \hline
    \move{} & \mathbf{neg}(\move{}) & \text{formulae} \\
    \hline
    0,0 & 0 & p \\
    0,1 & 0 & \truth \\
    0,2 & 1 & \neg \coalbox{1} p \\
    1,0 & 0 & \neg \neg p \\
    1,1 & 0 & \neg \neg p \\
    1,2 & 1 & \neg \coalbox{1} p \\
    2,0 & 1 & \truth \\
    2,1 & 1 & \neg \coalbox{1} p \\
    2,2 & 0 & \neg \neg p \\
    \hline
  \end{array}$
  \end{center}

  In the table above, it so happens that only one formula is included
  into each prestate; in general, however, this does not have to be
  the case.  Based on the above table, by applying \Rule{Next} to
  $\D_1$, we produce the following set of its prestate successors:

\medskip
  
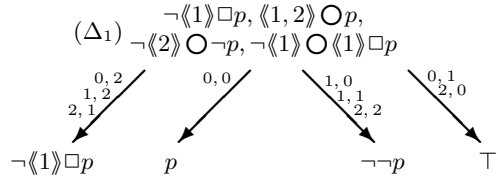
\begin{picture}(200,65)(0,405)
    \footnotesize
    \thicklines

    % node 2; state

    \put(78,457){$(\D_1)$}

    \put(150,460){\makebox(0,0){
        $\begin{array}{c}
          \neg \coalbox{1} p, \coalnext{1,2} p, \\
          \neg \coalnext{2} \neg p, \neg \coalnext{1} \coalbox{1} p
        \end{array}$
      }}

    \put(86, 440){\tiny $0,2$}
    \put(81.5, 434){\tiny $1,2$}
    \put(76, 428){\tiny $2,1$}

    \put(70,410){\makebox(0,0){
        $\begin{array}{c}
          \neg \coalbox{1} p
        \end{array}$
      }}

    \put(127, 440){\tiny $0,0$}

    \put(115,410){\makebox(0,0){
        $\begin{array}{c}
          p
        \end{array}$
      }}

    \put(173, 439){\tiny $1,0$}
    \put(177.5, 433.5){\tiny $1,1$}
    \put(184, 428){\tiny $2,2$}

    \put(105, 445){\vector(-1,-1){27}}
    \put(145, 445){\vector(-1,-1){27}}

    \put(165, 445){\vector(1,-1){27}}
    \put(205, 445){\vector(1,-1){27}}

    \put(212, 441){\tiny $0, 1$}
    \put(216, 435.5){\tiny $2,0$}

    \put(195,410){\makebox(0,0){
        $\begin{array}{c}
          \neg \neg p
        \end{array}$
      }}

    \put(235,410){\makebox(0,0){
        $\begin{array}{c}
          \truth
        \end{array}$
      }}
  \end{picture}
\end{example}

\begin{remark}
  \label{rem:at_most_one_neg_fma}
  Technically, \Rule{Next} ensures that every $\G_{\move{}} \in
  \sucpr{\D}$ satisfies the following properties:
  \begin{itemize}
  \item if $\set{\coalnext{A_i} \vp_i, \coalnext{A_j} \vp_j} \subseteq
    \D$ and $\set{\vp_i, \vp_j} \subseteq \G_{\move{}}$, then $A_i
    \inter A_j = \emptyset$;

    \item $\G_{\move{}}$ contains at most one formula of the form
      $\neg \psi$ such that $\neg \coalnext{A} \psi \in \D$, since the
      number $\mathbf{neg}(\move{})$ is uniquely determined for every
      $\move{} \in \moves{}{\D}$;

    \item if $\set{\coalnext{A_i} \vp_i, \neg \coalnext{A'} \psi}
      \subseteq \D$ and $\set{\vp_i, \neg \psi} \subseteq
      \G_{\move{}}$, then $A_i \subseteq A'$.
  \end{itemize}
\end{remark}

Note that there is a connection between the above properties and the
basic properties of ``next-time'' coalition modalities, such as
monotonicity and superadditivity (see \cite{Pauly01}, \cite{Pauly02},
\cite{GorDrim06}).

The construction phase, starting with a single prestate \set{\theta},
consists of alternately applying the rule \Rule{SR} to the prestates
created as a result of the last application of \Rule{Next} (or, if we
are at the beginning of the whole construction, to $\set{\theta}$) and
applying \Rule{Next} to the states created as a result of the last
application of \Rule{SR}. This cycle continues until any application
of \Rule{Next} does not produce any new prestates; after adding the
relevant arrows, if any, the construction stage is over.  As we show
in the next subsection, this is bound to happen in a finite number of
steps---more precisely, in the number of steps exponential in the
length of $\theta$.

\setcounter{example}{0}

\begin{example}[continued]
  \label{ex:atl_pretab_1}
  Here is a complete pretableau for the formula $\theta_1 = \neg
  \coalbox{1} p \con \coalnext{1,2} p \con \neg \coalnext{2} \neg p$:

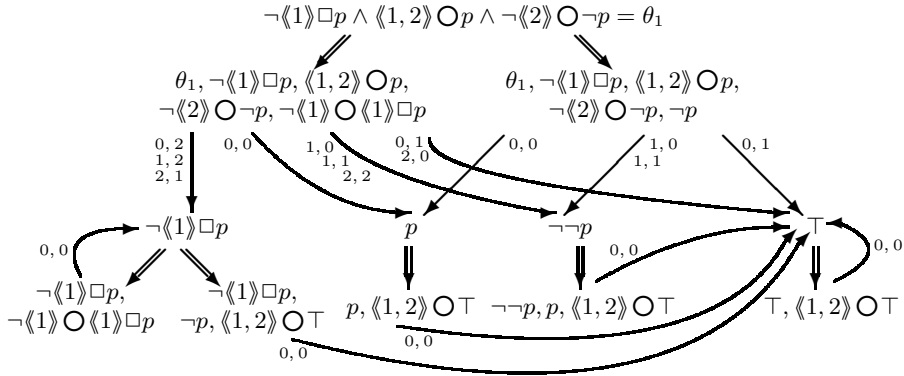
\medskip
  \begin{picture}(200,145)(0,355)
    \footnotesize
\thicklines
%node 1; prestate
\put(165,490){\makebox(0,0){
    $\begin{array}{c}
      \neg \coalbox{1} p \con
      \coalnext{1,2} p \con \neg \coalnext{2} \neg p = \theta_1
    \end{array}$
}}
\put(120,483){\line(-1,-1){10}}
\put(122,483){\line(-1,-1){10}}
\put(113,475){\vector(-1,-1){5}}

\put(207,483){\line(1,-1){10}}
\put(209,483){\line(1,-1){10}}
\put(216.3,475){\vector(1,-1){5}}

% node 2; state
\put(100,460){\makebox(0,0){
  $\begin{array}{c}
      \theta_1, \neg \coalbox{1} p, \coalnext{1,2} p, \\
      \neg \coalnext{2} \neg p, \neg \coalnext{1} \coalbox{1} p
  \end{array}$
}}

% node 3; state
\put(225,460){\makebox(0,0){
    $\begin{array}{c}
      \theta_1, \neg \coalbox{1} p, \coalnext{1,2} p, \\
      \neg \coalnext{2} \neg p, \neg p
  \end{array}$
}}

\put(62,446){\vector(0,-1){30}}
\put(48, 440){\tiny $0,2$}
\put(48, 434){\tiny $1,2$}
\put(48, 428){\tiny $2,1$}

% node 4; prestate
\put(60,410){\makebox(0,0){
    $\begin{array}{c}
      \neg \coalbox{1} p
    \end{array}$
}}

\qbezier(85, 446)(110, 418)(140, 416)
\put(140,416){\vector(1,0){5}}
\put(74, 440){\tiny $0,0$}

% node 5; prestate
\put(145,410){\makebox(0,0){
    $\begin{array}{c}
      p
    \end{array}$
}}

\qbezier(115, 446)(130, 428)(197, 416)
\put(195,416){\vector(1,0){5}}
\put(105, 439){\tiny $1, 0$}
\put(111, 433.5){\tiny $1,1$}
\put(119, 428){\tiny $2,2$}

\qbezier(152, 444)(150, 428)(283, 416)
\put(283,416){\vector(1,0){5}}
\put(139, 441){\tiny $0, 1$}
\put(141, 435.5){\tiny $2,0$}

\put(180,445){\vector(-1,-1){30}}
\put(182, 440){\tiny $0,0$}

\put(233,445){\vector(-1,-1){30}}
\put(235, 440){\tiny $1, 0$}
\put(229, 433.5){\tiny $1,1$}

% node 10; prestate
\put(205,410){\makebox(0,0){
  $\begin{array}{c}
   \neg \neg p
  \end{array}$
}}

% node 10; prestate
\put(298,410){\makebox(0,0){
  $\begin{array}{c}
   \truth
  \end{array}$
}}

\put(263,445){\vector(1,-1){30}}
\put(270, 440){\tiny $0, 1$}

\put(50,402){\line(-1,-1){10}}
\put(52,402){\line(-1,-1){10}}
\put(42,393){\vector(-1,-1){5}}

\put(20,380){\makebox(0,0){
    $\begin{array}{c}
      \neg \coalbox{1} p, \\
      \neg \coalnext{1} \coalbox{1} p
  \end{array}$
}}

\put(58,402){\line(1,-1){10}}
\put(60,402){\line(1,-1){10}}
\put(68,393){\vector(1,-1){5}}

\put(85,380){\makebox(0,0){
    $\begin{array}{c}
      \neg \coalbox{1} p, \\
   \neg p, \coalnext{1,2} \truth
  \end{array}$
}}

\put(143,403){\line(0,-1){11}}
\put(145,403){\line(0,-1){11}}
\put(144,392){\vector(0,-1){5}}

\put(145,380){\makebox(0,0){
  $\begin{array}{c}
   p, \coalnext{1,2} \truth
  \end{array}$
}}

\put(208,403){\line(0,-1){11}}
\put(210,403){\line(0,-1){11}}
\put(209,392){\vector(0,-1){5}}

\put(210,380){\makebox(0,0){
    $\begin{array}{c}
      \neg \neg p, p, \coalnext{1,2} \truth
  \end{array}$
}}

\put(297,403){\line(0,-1){11}}
\put(299,403){\line(0,-1){11}}
\put(298,392){\vector(0,-1){5}}

\put(304,380){\makebox(0,0){
    $\begin{array}{c}
      \truth, \coalnext{1,2} \truth
  \end{array}$
}}

\qbezier(305, 390)(330, 402)(307, 412)
\put(307,412){\vector(-1,0){5}}
\put(320, 402){\tiny $0,0$}

\qbezier(215, 390)(250, 410)(282, 410.5)
\put(283,410.5){\vector(1,0){5}}
\put(220, 400){\tiny $0,0$}

\qbezier(140, 373)(260, 355)(287, 406)
\put(286,405){\vector(1,1){5}}
\put(142.5, 365.5){\tiny $0,0$}

\qbezier(100, 368)(250, 330)(292, 405)
\put(291,404){\vector(1,1){5}}
\put(95, 361){\tiny $0,0$}

\qbezier(20, 392)(10, 410)(39, 410)
\put(37,410){\vector(1,0){5}}
\put(4.5, 400){\tiny $0,0$}

\end{picture}
\end{example}

\begin{example}
  \label{ex:atl_pretab_2}
  For yet another demonstration of our procedure, let us build a pretableau for
  the formula $\theta_2 = \coalbox{1} \neg q \con \coal{2} p \until
  q$:

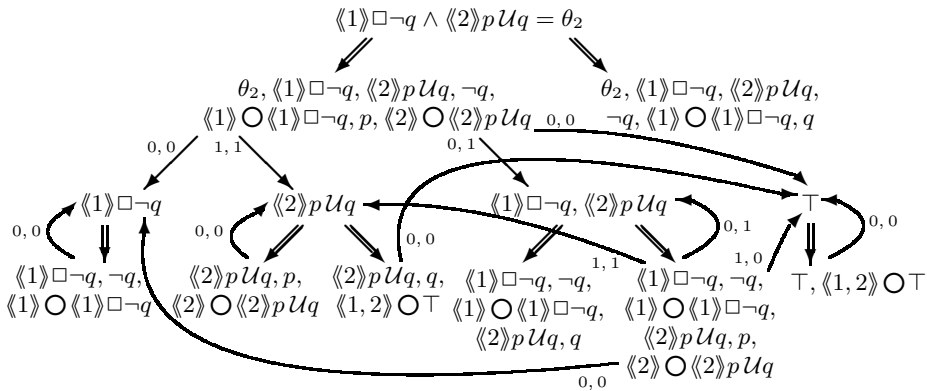
\medskip
  \begin{picture}(200,145)(0,355)
    \footnotesize
    \thicklines
    % node 1; prestate
    \put(165,490){\makebox(0,0){
        $\begin{array}{c}
          \coalbox{1} \neg q \con \coal{2} p \until q = \theta_2
        \end{array}$
      }}
    \put(130,483){\line(-1,-1){10}}
    \put(132,483){\line(-1,-1){10}}
    \put(123,475){\vector(-1,-1){5}}

    \put(207,483){\line(1,-1){10}}
    \put(209,483){\line(1,-1){10}}
    \put(216.3,475){\vector(1,-1){5}}

    \qbezier(194, 448)(233, 450)(292.5, 431)
    \put(292,432){\vector(1,-1){5}}
    \put(198, 450){\tiny $0,0$}

    % node 2; state
    \put(130,458){\makebox(0,0){
        $\begin{array}{c}
          \theta_2, \coalbox{1} \neg q, \coal{2} p \until q, \neg q, \\
          \coalnext{1} \coalbox{1} \neg q, p, \coalnext{2} \coal{2} p
          \until q
        \end{array}$
      }}

    % node 3; state
    \put(260,458){\makebox(0,0){
        $\begin{array}{c}
          \theta_2, \coalbox{1} \neg q, \coal{2} p \until q, \\
          \neg q, \coalnext{1} \coalbox{1} \neg q, q
        \end{array}$
      }}

    \put(66,446){\vector(-1,-1){20}}
    \put(47, 440){\tiny $0,0$}

    % node 4; prestate
    \put(37,420){\makebox(0,0){
        $\begin{array}{c}
          \coalbox{1} \neg q
        \end{array}$
      }}

    \put(30,412){\line(0,-1){10}}
    \put(32,412){\line(0,-1){10}}
    \put(31,403){\vector(0,-1){5}}

    \put(21,387){\makebox(0,0){
        $\begin{array}{c}
          \coalbox{1} \neg q, \neg q, \\
          \coalnext{1} \coalbox{1} \neg q
        \end{array}$
      }}

    \qbezier(19, 400)(0, 410)(18.7, 419)
    \put(15.5,416){\vector(1,1){5}}
    \put(-4, 408){\tiny $0, 0$}

    \put(82,446){\vector(1,-1){20}}
    \put(72.2, 440){\tiny $1,1$}

    % node 5; prestate
    \put(110,420){\makebox(0,0){
        $\begin{array}{c}
          \coal{2} p \until q
        \end{array}$
      }}

    \put(104,412){\line(-1,-1){10}}
    \put(106,412){\line(-1,-1){10.5}}
    \put(96,403){\vector(-1,-1){5}}

    \put(84,387){\makebox(0,0){
        $\begin{array}{c}
          \coal{2} p \until q, p, \\
          \coalnext{2} \coal{2} p \until q
        \end{array}$
      }}

    \qbezier(85, 400)(70, 410)(91, 418)
    \put(88,416){\vector(1,1){5}}
    \put(65, 408){\tiny $0, 0$}

    \put(173,445){\vector(1,-1){18}}
    \put(160.5, 441){\tiny $0,1$}

    \put(122,412){\line(1,-1){10}}
    \put(124,412){\line(1,-1){10.5}}
    \put(132,403){\vector(1,-1){5}}

    \put(138,387){\makebox(0,0){
        $\begin{array}{c}
          \coal{2} p \until q, q, \\
          \coalnext{1,2} \truth
        \end{array}$
      }}

    \qbezier(143, 399)(130, 460)(289.5, 423.5)
     \put(288,424){\vector(1,0){5}}
    \put(145, 405){\tiny $0,0$}

    % node 5; prestate
    \put(210,420){\makebox(0,0){
        $\begin{array}{c}
          \coalbox{1} \neg q, \coal{2} p \until q
        \end{array}$
      }}

    \put(202,412){\line(-1,-1){10}}
    \put(204,412){\line(-1,-1){10.5}}
    \put(194,403){\vector(-1,-1){5}}

    \put(191,380){\makebox(0,0){
        $\begin{array}{c}
          \coalbox{1} \neg q, \neg q, \\
          \coalnext{1} \coalbox{1} \neg q, \\
          \coal{2} p \until q, q
        \end{array}$
      }}

    \put(232,412){\line(1,-1){10}}
    \put(234,412){\line(1,-1){10.5}}
    \put(242,403){\vector(1,-1){5}}

    \put(256,375){\makebox(0,0){
        $\begin{array}{c}
          \coalbox{1} \neg q, \neg q, \\
          \coalnext{1} \coalbox{1} \neg q, \\
          \coal{2} p \until q, p, \\
          \coalnext{2} \coal{2} p \until q
        \end{array}$
      }}

    \qbezier(235, 398)(175, 422)(132, 420)
    \put(134,420){\vector(-1,0){5}}
    \put(215, 396){\tiny $1,1$}

    \qbezier(250, 400)(275, 412)(252, 422)
    \put(252,422){\vector(-1,0){5}}
    \put(265, 410){\tiny $0,1$}

    \qbezier(282, 395)(286, 405)(290, 412)
    \put(290,412){\vector(1,1){5}}
    \put(269, 398){\tiny $1,0$}

    \qbezier(225, 360)(34, 340)(46, 410)
    \put(46,410){\vector(0,1){5}}
    \put(210, 351){\tiny $0,0$}

    % node 10; prestate
    \put(298,420){\makebox(0,0){
        $\begin{array}{c}
          \truth
        \end{array}$
      }}

    \put(297,413){\line(0,-1){11}}
    \put(299,413){\line(0,-1){11}}
    \put(298,402){\vector(0,-1){5}}

    \put(316,390){\makebox(0,0){
        $\begin{array}{c}
          \truth, \coalnext{1,2} \truth
        \end{array}$
      }}

    \qbezier(305, 400)(330, 412)(307, 422)
    \put(307,422){\vector(-1,0){5}}
    \put(320, 412){\tiny $0,0$}
\end{picture}
\end{example}

%%%%%%%%%%%%%%%%%%%%%%%%%%%%%%%%%%
\subsection{Termination and complexity of the construction phase}
\label{sec:term_cosntr}
%%%%%%%%%%%%%%%%%%%%%%%%%%%%%%%%%%

To prove that the construction phase eventually terminates and to
estimate its complexity, we use the concept of the extended closure of
an \ATL-formula.

\begin{definition}
  \label{def:atl-tableau-closure}
  Let $\theta$ be an \ATL-formula.  The \de{closure of $\theta$},
  denoted by \cl{\theta}, is the least set of formulae such that
  \begin{itemize}
  \item $\theta \in \cl{\theta}$;
  \item \cl{\theta} is closed under subformulae;
  \item if $\coal{A} (\vp \until \psi) \in \cl{\theta}$, then $\vp
    \con \coalnext{A} \coal{A} (\vp \until \psi) \in \cl{\theta}$;
  \item if $\neg \coal{A} (\vp \until \psi) \in \cl{\theta}$, then
    $\neg \psi \con \neg \vp, \neg \psi \con \neg \coalnext{A}
    \coal{A} (\vp \until \psi) \in \cl{\theta}$;
  \item if $\coalbox{A} \vp \in \cl{\theta}$, then $\vp \con
    \coalnext{A} \coalbox{A} \vp \in \cl{\theta}$.
  \end{itemize}
\end{definition}

\begin{definition}
  Let $\theta$ be an \ATL-formula. The \de{extended closure of
    $\theta$}, denoted by \ecl{\theta}, is the least set of formulae
  such that
  \begin{itemize}
  \item if $\vp \in \cl{\theta}$, then $\vp, \neg \vp \in
    \ecl{\theta}$;
  \item if $\neg \coalnext{\agents_{\theta}} \vp \in \cl{\theta}$,
    then $\coalnext{\emptyset} \neg \vp \in \ecl{\theta}$;
  \item $\top \in \ecl{\theta}$;
  \item $\coalnext{\agents} \truth \in \ecl{\theta}$.
  \end{itemize}
\end{definition}

We denote the cardinality of \ecl{\theta} by \card{\ecl{\theta}} and
the length of a formula $\theta$ by $|\theta|$.  When calculating the
length of a formula, we assume that every agent's name counts as one
symbol and that a pair of coalition braces is ``lumped together'' as
one symbol with the temporal operator that follows it; thus,
$|\coalnext{1, 2} p| = 4$.

\begin{lemma}
  \label{lm:size_of_ecl}
  Let $\theta$ be a \ATL-formula.  Then, \ecl{\theta} is finite; more
  precisely, $\card{\ecl{\theta}} \in \bigo{|\theta|}$, i.e,
  $\card{\ecl{\theta}} \leq c \cdot \card{\theta}$ for some $c \geq
  1$.
\end{lemma}

\begin{proof}
  Straightforward.
\end{proof}

To simplify notation, let us denote $\card{\theta}$ by $n$ and
$\card{\agents_{\theta}}$ by $k$; let also $c$ be the constant from
the statement of the preceding lemma. While building the pretableau
$\tableau{P}^{\theta}$, we create \bigo{2^{cn}} states and
\bigo{2^{cn}} prestates.  To create a state, we need no more than
\bigo{cn} steps, thus the creation of all the states takes not more
than \bigo{cn \times 2^{cn}} steps.  For a given state $\D$, to create
all the prestates in $\sucpr{\D}$, we first produce a $\G_{\move{}}$
associated with a given $\move{} \in \moves{}{\D}$, which costs
\bigo{cn} steps, and then check whether it is identical to a prestate
created earlier, which takes \bigo{(cn)^2 \times 2^{cn}} steps.  As
there are, all in all, $\bigo{(cn)^k}$ move vectors in $D(\D)$, the
whole procedure of creating prestates from a given state costs
\bigo{(cn)^k \times (cn + (cn)^2 \times 2^{cn})}.  Applying this
procedure to all \bigo{2^{cn}} states, i.e, creating all prestates can
thus be done in $\bigo{2^{cn} \times (cn)^k \times (cn + (cn)^2 \times
  2^{cn}} = \bigo{2^{(k+1)\log (cn) + cn} + 2^{(k+2) \log (cn) + 2cn
  }} = \bigo{2^{(k+2) \log(cn) + 2cn}}$. As this clearly dominates the
complexity of creating states, the cost of the construction phase as a
whole is \bigo{2^{(k+2) \log(cn) + 2cn}}.

%%%%%%%%%%%%%%%%%%%%%%%%%%%%%%%%%%
\subsection{Prestate elimination phase} \label{sec:prestate_elimination}
%%%%%%%%%%%%%%%%%%%%%%%%%%%%%%%%%%

At the second phase of the tableau procedure, we remove from
$\tableau{P}^{\theta}$ all the prestates and all the unmarked arrows,
by applying the following rule:

\bigskip

\Rule{PR} For every prestate $\G$ in $\tableau{P}^{\theta}$, do the
following:
\begin{enumerate}
\item remove $\G$ from $\tableau{P}^{\theta}$;
\item for all states $\D$ in $\tableau{P}^{\theta}$ with $\D
  \movearrow \G$ and all $\D' \in \st{\G}$, put $\D \movearrow \D'$.
\end{enumerate}

\bigskip

We call the graph obtained by applying \Rule{PR} to
$\tableau{P}^{\theta}$ the \emph{initial tableau}, which we denote by
$\tableau{T}_0^{\theta}$.  Note that if in $\tableau{P}^{\theta}$ we
have $\D \movearrow \G$ and $\st{\G}$ contains more than one state,
then in $\tableau{T}_0^{\theta}$ there is going to be more than one
edge labeled with $\move{}$ going out of $\D$.

\setcounter{example}{0}

\begin{example}[continued]
  \label{ex:atl_inittab_1}
  Here is the initial tableau $\tableau{T}_0^{\theta_1}$ for the formula
  $\theta_1 = \neg \coalbox{1} p \con \coalnext{1,2} p \con \neg
  \coalnext{2} \neg p$ (as before, some states are named for future
  reference):

\medskip
  
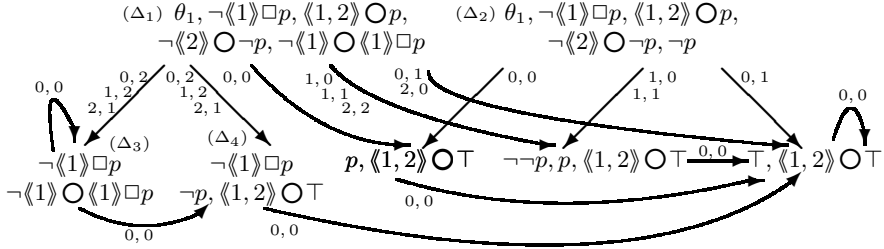
\begin{picture}(200,95)(0,380)
    \footnotesize
    \thicklines

    \put(36, 464){\tiny $(\D_1)$}

    % node 2; state
    \put(100,460){\makebox(0,0){
        $\begin{array}{c}
          \theta_1, \neg \coalbox{1} p, \coalnext{1,2} p, \\
          \neg \coalnext{2} \neg p, \neg \coalnext{1} \coalbox{1} p
        \end{array}$
      }}

    \put(162, 464){\tiny $(\D_2)$}

    % node 3; state
    \put(225,460){\makebox(0,0){
        $\begin{array}{c}
          \theta_1, \neg \coalbox{1} p, \coalnext{1,2} p, \\
          \neg \coalnext{2} \neg p, \neg p
        \end{array}$
      }}

    \put(52,446){\vector(-1,-1){30}}
    \put(35, 440){\tiny $0,2$}
    \put(29, 434){\tiny $1,2$}
    \put(23, 428){\tiny $2,1$}

    \put(62,446){\vector(1,-1){30}}
    \put(52.5, 440){\tiny $0,2$}
    \put(57.5, 434){\tiny $1,2$}
    \put(63, 428){\tiny $2,1$}

    \qbezier(10, 414)(5, 448)(18, 422)
    \put(18,422){\vector(0,-1){5}}
    \put(3, 436){\tiny $0,0$}

    \qbezier(20, 392)(42.5, 380)(65, 390)
    \put(64,389){\vector(1,1){5}}
    \put(37, 380){\tiny $0,0$}

    \qbezier(85, 446)(110, 418)(140, 416)
    \put(140,416){\vector(1,0){5}}
    \put(74, 440){\tiny $0,0$}

    % node 5; prestate
    \put(145,410){\makebox(0,0){
        $\begin{array}{c}
          p, \coalnext{1,2} \truth
        \end{array}$
      }}

    % \put(125,446){\vector(1,-1){30}}
    \qbezier(115, 446)(130, 428)(197, 416)
    \put(195,416){\vector(1,0){5}}
    \put(105, 439){\tiny $1, 0$}
    \put(111, 433.5){\tiny $1,1$}
    \put(119, 428){\tiny $2,2$}

    \qbezier(152, 444)(150, 428)(283, 416)
    \put(283,416){\vector(1,0){5}}
    \put(139, 441){\tiny $0, 1$}
    \put(141, 435.5){\tiny $2,0$}

    \put(180,445){\vector(-1,-1){30}}
    \put(182, 440){\tiny $0,0$}

    \put(233,445){\vector(-1,-1){30}}
    \put(235, 440){\tiny $1, 0$}
    \put(229, 433.5){\tiny $1,1$}

    \put(215,410){\makebox(0,0){
        $\begin{array}{c}
          \neg \neg p, p, \coalnext{1,2} \truth
        \end{array}$
      }}

    \put(298,410){\makebox(0,0){
        $\begin{array}{c}
          \truth, \coalnext{1,2} \truth
        \end{array}$
      }}

    \put(263,445){\vector(1,-1){30}}
    \put(270, 440){\tiny $0, 1$}

    \put(31, 415){\tiny $(\D_3)$}

    \put(20,403){\makebox(0,0){
        $\begin{array}{c}
          \neg \coalbox{1} p \\
          \neg \coalnext{1} \coalbox{1} p
        \end{array}$
      }}

    \put(68, 417){\tiny $(\D_4)$}

    \put(85,403){\makebox(0,0){
        $\begin{array}{c}
          \neg \coalbox{1} p \\
          \neg p, \coalnext{1,2} \truth
        \end{array}$
      }}

\put(145,410){\makebox(0,0){
  $\begin{array}{c}
   p, \coalnext{1,2} \truth
  \end{array}$
}}

\qbezier(305, 418)(312, 441)(317, 420)
\put(317,421){\vector(0,-1){5}}
\put(307, 433){\tiny $0,0$}

\put(250,410){\vector(1,0){23}}
\put(254, 412){\tiny $0,0$}

\qbezier(140, 403)(200, 385)(275, 403)
\put(274,402.5){\vector(1,0){5}}
\put(143, 393){\tiny $0,0$}

\qbezier(90, 392)(230, 360)(289, 402)
\put(288,401){\vector(1,1){5}}
\put(92, 382){\tiny $0,0$}
\end{picture}

Thus, our procedure for the formula $\neg \coalbox{1} p \con
\coalnext{1,2} p \con \neg \coalnext{2} \neg p$ creates 7 states.  For
the sake of comparison with the top-down tableau procedure
from~\cite{WLWW06}, we estimate how many states would be created using
that procedure.  As the running time of both procedures is roughly
proportional to the number of states created, this should give us an
idea as to how the two procedures compare in practice.

While we use the concept of extended closure of a formula for
metatheoretical purposes (to prove termination and estimate
complexity, see Section~\ref{sec:term_cosntr}), the top-down tableaux-like decision procedure from~\cite{WLWW06} uses it essentially.  Technically speaking, the procedure from~\cite{WLWW06} creates not states, but ``types''---maximal, propositionally consistent, saturated subsets of the extended closure of the input formula.  So, we estimate how many types the tableau procedure
from~\cite{WLWW06} would create for the formula $\neg \coalbox{1} p
\con \coalnext{1,2} p \con \neg \coalnext{2} \neg p$.  To that end, we
fist enumerate positive formulas of the extended closure for this
formula:
\medskip

(1) \ $\neg \coalbox{1} p \con \coalnext{1,2} p \con \neg
\coalnext{2} p$,

(2) \ $\neg \coalbox{1} p \con \coalnext{1,2} p$,

(3) \ $\coalnext{2} \neg p$,

(4) \ $\coalbox{1} p$,

(5) \ $\coalnext{1} \coalbox{1} p$,

(6) \ $\coalnext{1, 2} p$,

(7) \ $p$.

\medskip
For every formula from the above list, each type contains either that formula or its negation. However, not every such combination is allowed, as there
are dependencies between formulae as to their presence in a type.

First, if (1) is in a type, then that type must contain (2), $\neg (3)$,
$\neg (4)$ and $(6)$; so, there are $2^2$ distinct types containing formula
(1).  Second, if $\neg (1)$ and $\neg (2)$ are in a type, then we have
two cases: if the type contains (4), then it contains (5), generating
$2^3$ types; and if the type $\neg (6)$, then it contains $\neg (4)$,
generating $2^3$ more types.  Lastly, if $\neg (1)$ and $(2)$ are in a
type, then $(3)$, $\neg(4)$, and $(6)$ are also in the type, generating
$2^2$ types. Thus, all in all, the top-down tableau procedure
from~\cite{WLWW06} creates 24 types, as opposed to 7 states created by
incremental tableaux.
\end{example}

\begin{example}[continued]
  \label{ex:atl_inittab_2}
  Here is the initial tableau $\tableau{T}_0^{\theta_2}$ for the
  formula $\theta_2 = \coalbox{1} \neg q \con \coal{2} p \until q$ (as
  in the previous example, some states are named for future
  reference):

\medskip
  
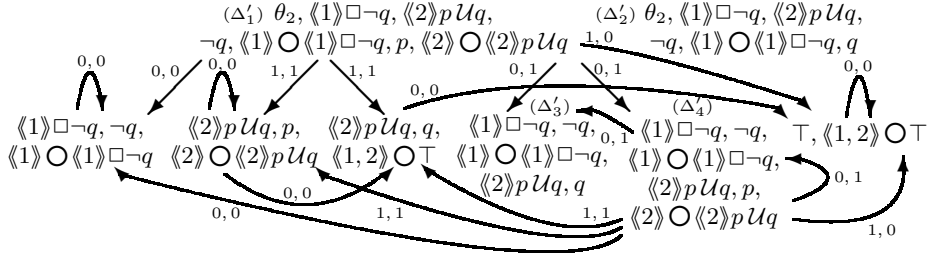
\begin{picture}(230,100)(0,370)
    \footnotesize
    \thicklines

    \qbezier(210, 452)(233, 450)(292.5, 431)
    \put(292,432){\vector(1,-1){5}}
    \put(211, 453){\tiny $1,0$}

    \put(74, 462){\tiny $(\D'_1)$}

    % node 2; state
    \put(135,458){\makebox(0,0){
        $\begin{array}{c}
          \theta_2, \coalbox{1} \neg q, \coal{2} p \until q,  \\
          \neg q, \coalnext{1} \coalbox{1} \neg q, p, \coalnext{2} \coal{2} p
          \until q
        \end{array}$
      }}

    \put(215, 462){\tiny $(\D'_2)$}

    % node 3; state
    \put(275,458){\makebox(0,0){
        $\begin{array}{c}
          \theta_2, \coalbox{1} \neg q, \coal{2} p \until q, \\
          \neg q, \coalnext{1} \coalbox{1} \neg q, q
        \end{array}$
      }}

    \put(66,446){\vector(-1,-1){20}}
    \put(47, 440){\tiny $0,0$}

    \put(110,446){\vector(-1,-1){20}}
    \put(92, 440){\tiny $1,1$}

    \put(115,446){\vector(1,-1){20}}
    \put(121.5, 440){\tiny $1,1$}

    % node 4; prestate
    \put(20,415){\makebox(0,0){
        $\begin{array}{c}
         \coalbox{1} \neg q, \neg q, \\
         \coalnext{1} \coalbox{1} \neg q
        \end{array}$
      }}

    \qbezier(19, 428)(24, 455)(28.5, 428.5)
    \put(28.4, 432){\vector(0,-1){5}}
    \put(19, 444){\tiny $0, 0$}

    \put(82,415){\makebox(0,0){
        $\begin{array}{c}
          \coal{2} p \until q, p, \\
          \coalnext{2} \coal{2} p \until q
        \end{array}$
      }}

    \qbezier(69, 428)(74, 455)(78.5, 428.5)
    \put(78.4, 432){\vector(0,-1){5}}
    \put(68.5, 443){\tiny $0,0$}

    \put(135,415){\makebox(0,0){
        $\begin{array}{c}
          \coal{2} p \until q, q, \\
          \coalnext{1,2} \truth
        \end{array}$
      }}

    \qbezier(143, 428)(193, 445)(285, 428)
     \put(283,429){\vector(1,-1){5}}
    \put(145, 433){\tiny $0,0$}

    \put(200,445){\vector(-1,-1){18}}
    \put(182.5, 441){\tiny $0,1$}

    \put(190, 427){\tiny $(\D'_3)$}

    \put(191,410){\makebox(0,0){
        $\begin{array}{c}
          \coalbox{1} \neg q, \neg q, \\
          \coalnext{1} \coalbox{1} \neg q, \\
          \coal{2} p \until q, q
        \end{array}$
      }}

    \put(210,445){\vector(1,-1){18}}
    \put(214.5, 441){\tiny $0,1$}

    \put(244, 427){\tiny $(\D'_4)$}

    \put(256,403){\makebox(0,0){
        $\begin{array}{c}
          \coalbox{1} \neg q, \neg q, \\
          \coalnext{1} \coalbox{1} \neg q, \\
          \coal{2} p \until q, p, \\
          \coalnext{2} \coal{2} p \until q
        \end{array}$
      }}

    \qbezier(290, 386)(335, 380)(331.5, 409)
    \put(331.5,405){\vector(0,1){5}}
    \put(318, 380){\tiny $1,0$}

    \qbezier(290, 393)(315, 402)(292, 409)
    \put(292,409){\vector(-1,0){5}}
    \put(305.5, 400){\tiny $0, 1$}

    \qbezier(229, 419)(213, 430)(210, 427)
    \put(213,427){\vector(-1,0){5}}
    \put(217, 415){\tiny $0,1$}

    \qbezier(225, 386)(200, 375)(151, 405)
    \put(154,402){\vector(-1,1){5}}
    \put(210, 386.5){\tiny $1,1$}

    \qbezier(225, 383)(205, 370)(111, 403)
    \put(114,401){\vector(-1,1){5}}
    \put(132, 386.5){\tiny $1,1$}

    \qbezier(225, 380)(205, 360)(35, 403)
    \put(38,401){\vector(-1,1){5}}
    \put(70, 386.5){\tiny $0,0$}

    \qbezier(75, 403)(105, 380)(135, 403)
    \put(134,401){\vector(1,1){5}}
    \put(97, 394){\tiny $0,0$}

    \put(315,418){\makebox(0,0){
        $\begin{array}{c}
          \truth, \coalnext{1,2} \truth
        \end{array}$
      }}

    \qbezier(310, 425)(315, 455)(319, 425)
    \put(319,430){\vector(0,-1){5}}
    \put(309, 442){\tiny $0,0$}
  \end{picture}

  Again, for the sake of comparison with tableaux form~\cite{WLWW06},
  we estimate the number of types created by those tableaux; a
  calculation similar to the one from the previous example shows that
  36 types are created by the top-down tableau-like procedure, as opposed
  to 8 states created by the incremental tableau procedure.
\end{example}

We briefly remark on the time required for this second phase.  Once
again, to simplify notation, let us denote $\card{\theta}$ by $n$.
Recall that $\card{\ecl{\theta}} \in \bigo{\card{\theta}}$, i.e,
$\card{\ecl{\theta}} = c \cdot \card{\theta}$ for some $c \geq 1$.  To
remove a single prestate, we need to delete from the memory its
\bigo{cn} formulae and redirect at most \bigo{2^{cn} \times 2^{cn}}
edges---having identified set-theoretically equal states as part of
the application of \Rule{Next} and having ``glued together'' arrows
having the same source and target, we do not have, at this stage, to
deal with \bigo{cn^{k}} outgoing edges for each state.  Hence, the
removal of a single prestate can be done in \bigo{2^{2cn}} steps. As
there are at most \bigo{2^{cn}} prestates, the whole procedure takes
\bigo{2^{3cn}} steps.

%%%%%%%%%%%%%%%%%%%%%%%%%%%%%%%%%%
\subsection{State elimination phase} \label{sec:state_elimination}
%%%%%%%%%%%%%%%%%%%%%%%%%%%%%%%%%%

During the state elimination phase, we remove those nodes of
$\tableau{T}_0^{\theta}$ that cannot be satisfied in any CGHS.  As
already mentioned, there are three reasons why a state $\D$ of
$\tableau{T}_0^{\theta}$ can turn out to be unsatisfiable in any CGHS.
First, $\D$ may contain a patent inconsistency\footnote{As states are
  downward-saturated, this is tantamount to saying that $\D$ contains
  a propositional inconsistency, even though in general these two
  concepts are not identical, as noted earlier.}. Secondly,
satisfiability of $\D$ may require that at least one state from a set
of tableau states $X$ is satisfiable as a successor of the state
$s_{\D}$ of a CGHS presumably satisfying $\D$, while all states of $X$
turn out to be unsatisfiable sets.  Thirdly, $\D$ may contain an
eventuality that is not realized in the tableau; that this implies
unsatisfiability of $\D$ is much less obvious than in the preceding
two cases---in fact, a major task within the soundness proof for our
procedure is to establish that this is indeed so. Accordingly, we have
three elimination rules, \Rule{E1}--\Rule{E3}, each taking care of
eliminating states of $\tableau{T}_0^{\theta}$ on one of the
above-mentioned counts.

Technically, the elimination phase is divided into stages; at stage
$n+1$, we remove from the tableau $\tableau{T}_n^{\theta}$ obtained at
the previous stage exactly one state, by applying one of the
elimination rules, thus obtaining the tableau
$\tableau{T}_{n+1}^{\theta}$. We now state the rules governing the
process.  The set of states of tableau $\tableau{T}_{m}^{\theta}$ is
denoted by $S_m^{\theta}$.

The rationale for the first rule is obvious.

\bigskip

\Rule{E1} If $\set{\vp, \neg \vp} \subseteq \D \in S_n^{\theta}$, then
obtain $\tableau{T}_{n+1}^{\theta}$ by eliminating $\D$ from
$\tableau{T}_n^{\theta}$.

\bigskip

The rationale behind the second rule is also intuitively clear: if
$\D$ is to be satisfiable, then for each $\move{} \in \moves{}{\D}$
there should exists a satisfiable $\D'$ with $\D \movearrow \D'$.  If
all such $\D'$s have been eliminated because they are unsatisfiable,
then $\D$ is itself unsatisfiable.

\bigskip

\Rule{E2} If, for some $\sigma \in D(\Delta)$, all states $\D'$ with
$\D \movearrow \D'$ have been eliminated at earlier stages, then
obtain $\tableau{T}_{n+1}^{\theta}$ by eliminating $\D$ from
$\tableau{T}_n^{\theta}$.

\bigskip

To formulate \Rule{E3}, we need the concepts of realization of an
eventuality in a tableau. To define that concept, we need some
auxiliary notation. Let $\D \in S_0^{\theta}$, and let $\coalnext{A}
\vp$ be the $p$-th formula in the linear ordering of the next-time formulae
of $\D$ induced as part of application of \Rule{Next} to $\D$; let,
finally, $\neg \coalnext{A'} \psi$ be the $q$-th formula in the same
ordering.  Then, we use the following notation:

\begin{center}
  \(
  \begin{array}{l}
    D (\D, \coalnext{A} \vp): = \crh{\move{} \in D(\D)}{\move{a} = p
      \text{ for every } a \in A}; \\

    D (\D, \neg \coalnext{A'} \psi):= \crh{\move{}
      \in D(\D)} {\mathbf{neg}(\move{}) = q \text{ and } \agents_{\theta}
      \setminus A' \subseteq N(\sigma)}.
  \end{array}
  \)
\end{center}

\noindent Intuitively, $D(\D, \chi)$ corresponds to an $A$-move (if
$\chi = \coalnext{A} \vp$) or a co-$A$-move (if $\chi = \neg
\coalnext{A'} \psi$) witnessing the ``satisfaction'' of $\chi$ at
state $\D$ (recall that $A$-moves and co-$A$-moves can be identified
with equivalence classes on the set of move vectors).

We now recursively define what it means for an eventuality of the
form $\coal{A} \vp \until \psi$ to be realized at a state $\D$ of
tableau $\tableau{T}_n^{\theta}$.

\begin{definition}[Realization of eventuality $\coal{A} \vp \until
  \psi$]
  \label{def:realization_until}
\mbox{}
\begin{enumerate}
\item If $\set{\psi, \coal{A} \vp \until \psi} \subseteq \D \in
  S_n^{\theta}$, then $\coal{A} \vp \until \psi$ is \de{realized} at
  $\D$ in $\tableau{T}_n^{\theta}$;
\item If $\set{\vp, \coalnext{A} \coal{A} \vp \until \psi, \coal{A}
    \vp \until \psi} \subseteq \D$ and for every $\move{} \in D (\D,
  \coalnext{A} \coal{A} \vp \until \psi)$, there exists $\D' \in
  S_n^{\theta}$ such that
  \begin{itemize}
  \item $\D \movearrow \D'$ and
  \item $\coal{A} \vp \until \psi$ is realized at $\D'$ in
    $\tableau{T}_n^{\theta}$,
  \end{itemize}
  then $\coal{A} \vp \until \psi$ is \de{realized} at $\D$ in
  $\tableau{T}_n^{\theta}$.
\end{enumerate}
\end{definition}

The definition of realization for eventualities of the form $\neg
\coalbox{A} \vp$ is analogous:

\begin{definition}[Realization of eventuality $\neg \coalbox{A} \vp$]
  \label{def:realization_neg_box}
  \mbox{}
  \begin{enumerate}
  \item If $\set{\neg \vp, \neg \coalbox{A} \vp} \subseteq \D \in
    S_n^{\theta}$, then $\neg \coalbox{A} \vp$ is \de{realized} at
    $\D$ in $\tableau{T}_n^{\theta}$;
  \item If $\set{\neg \coalnext{A} \coalbox{A} \vp, \neg \coalbox{A}
      \vp} \subseteq \D$ and, for every $\move{} \in D (\D, \neg
    \coalnext{A} \coalbox{A} \vp)$ there exists $\D' \in S_n^{\theta}$
    such that
    \begin{itemize}
    \item $\D \movearrow \D'$ and
    \item $\neg \coalbox{A} \vp$ is realized at $\D'$ in
      $\tableau{T}_n^{\theta}$,
    \end{itemize}
    then $\neg \coalbox{A} \vp$ is \de{realized} at $\D$ in
    $\tableau{T}_n^{\theta}$.
  \end{enumerate}
\end{definition}

We can now state our third elimination rule.

\bigskip

\Rule{E3} If $\D \in S_n^{\theta}$ contains an eventuality that is not
realized at $\D$ in $\tableau{T}_n^{\theta}$, then obtain
$\tableau{T}_{n+1}^{\theta}$ by removing $\D$ from
$\tableau{T}_n^{\theta}$.

\bigskip

While implementation of the rules \Rule{E1} and \Rule{E2} is
straightforward, implementation of \Rule{E3} is less so.  It can be
done by computing the set of states realizing a given eventuality
$\xi$ in tableau $\tableau{T}_n^{\theta}$, say, by marking those
states that realize $\xi$ in $\tableau{T}_n^{\theta}$.  To formally
describe the procedure, we need some extra notation.

First, given $\D \in S_n^{\theta}$ and $\move{} \in D (\D)$, we denote
by $\suctab{\D}{\move{}}$ the set \crh{\D' \in S_n^{\theta}}{\D
  \movearrow \D'}.  Secondly, given a formula $\chi$, we write,
abusing set-theoretic notation, $\chi \in \tableau{T}_n^{\theta}$ to
mean that $\chi \in \D$ for some $\D \in S_n^{\theta}$.

We now describe the marking procedure for $\tableau{T}_n^{\theta}$
with respect to eventuality $\xi$.  We first do so for eventualities
of the form $\coal{A} \vp \until \psi$.  Initially, we mark $\D$ if
$\psi\in \D$. Afterwards, we
repeat the following computation for every $\D \in S_n^{\theta}$ that
is still unmarked: mark $\D$ if, for every $\move{}\, \in D(\D,
\coalnext{A} \coal{A} \vp \until \psi)$, there exists at least one
$\D'$ such that $\D' \in \suctab{\D}{\move{}}$ and $\D'$ is marked.
The procedure is over when no more states can get marked.

The procedure for computing eventualities of the form $\neg
\coalbox{A} \vp$ is similar.  Initially, we mark $\D$ if
$\neg  \vp \in \D$.  Afterwards, we repeat the
following computation for every $\D \in S_n^{\theta}$ that is still
unmarked: mark $\D$ if, for every $\move{}\, \in D(\D, \neg
\coalnext{A} \coalbox{A} \vp \}$, there exists at least one $\D'$ such
that $\D' \in \suctab{\D}{\move{}}$ and $\D'$ is marked.  The
procedure is over when no more states can get marked.

\begin{lemma}
  \label{lem:realization_rank}
  Let $\D \in S_n^{\theta}$ and $\xi \in \tableau{T}_n^{\theta}$ be an
  eventuality.  Then, $\xi$ is realized at $\D$ in
  $\tableau{T}_n^{\theta}$ iff $\D$ is marked in
  $\tableau{T}_n^{\theta}$ with respect to $\xi$.
\end{lemma}

\begin{proof}
  Straightforward.
\end{proof}

Thus, the application of \Rule{E3} in tableau $\tableau{T}_n^{\theta}$
with respect to eventuality $\xi$ consists of carrying out the marking
procedure with respect to $\xi$ and then removing all the states that
contain $\xi$, but have not been marked with respect to $\xi$.

We have thus far described individual rules and how they can be
implemented.  To describe the state elimination phase as a whole, it
is crucial to specify the order of application of those rules.

First, we apply \Rule{E1} to all the states of
$\tableau{T}_0^{\theta}$; it is clear that, once it is done, we do not
need to go back to \Rule{E1} again.  The cases of \Rule{E2} and
\Rule{E3} are slightly more involved.  Having applied \Rule{E3} to the
states of the tableau, we could have removed, for some $\D$, all the
states accessible from it along the arrows marked by some $\move{} \in
\moves{}{\D}$; hence, we need to reapply \Rule{E2} to the resultant
tableau to get rid of such $\D$'s.  Conversely, having applied
\Rule{E2}, we could have removed some states that were instrumental in
realizing certain eventualities; hence, having applied \Rule{E2}, we
need to reapply \Rule{E3}.  Furthermore, we cannot stop the procedure
unless we have checked that \emph{all} eventualities are realized.
Thus, what we need is to apply \Rule{E3} and \Rule{E2} in a dovetailed
sequence that cycles through all the eventualities. More precisely, we
arrange all the eventualities occurring in the tableau obtained from
$\tableau{T}_0^{\theta}$ by having applied \Rule{E1} to
$\tableau{T}_0^{\theta}$ in the list $\xi_1, \ldots, \xi_m$.  Then, we
proceed in cycles. Each cycle consists of alternatingly applying
\Rule{E3} to the pending eventuality, and then applying \Rule{E2} to
the tableau resulting from that application, until all the
eventualities have been dealt with; once we reach $\xi_m$, we loop
back to $\xi_1$.  The cycles are repeated until, having gone through
the whole cycle, we have not had to remove any states.

Once that happens, the state elimination phase is over.  The resultant
graph we call the \fm{final tableau} for $\theta$ and denote by
$\tableau{T}^{\theta}$.

\begin{definition}
  The final tableau $\tableau{T}^{\theta}$ is \de{open} if $\theta \in
  \D$ for some $\D \in S^{\theta}$; otherwise, $\tableau{T}^{\theta}$
  is \de{closed}.
\end{definition}

The tableau procedure returns ``no'' if the final tableau is closed;
otherwise, it returns ``yes'' and, moreover, provides sufficient
information for producing a finite model satisfying $\theta$; that
construction is described in section \ref{sec:completeness}.

\setcounter{example}{0}

\begin{example}[continued]
  Consider the initial tableau for our formula $\theta_1$.  First, no
  states of that tableau contain patent inconsistencies.  Moreover,
  all four states containing the eventuality $\neg \coalbox{1} p$
  (which is the only eventuality in the tableau) get marked with
  respect to $\neg \coalbox{1} p$. Indeed, $\D_2$ and $\D_4$ get
  marked since they contain $\neg p$; $\D_1$ get marked since all the
  relevant move vectors (i.e, those for which $\mathbf{neg}(\move{}) =
  1$ and agent $2$ votes negatively; there are 3 such move vectors:
  $(0, 2), (1, 2), (2, 1)$) lead to a state $\D_4$ that is marked;
  finally $\D_3$ is marked as the only move vector going out of that
  state leads to a marked state, $\D_4$.  Lastly, all the states have
  all the required successors.  Therefore, no state of the initial
  tableau gets eliminated, hence, the final tableau  $\tableau{T}^{\theta_1}$ coincides with the initial tableau  $\tableau{T}_0^{\theta_1}$.  Thus, $\tableau{T}^{\theta_1}$ is open (it contains two states, $\D_1$ and $\D_2$, containing $\theta_1$);  therefore, $\theta_1 = \neg \coalbox{1} p \con \coalnext{1,2} p \con  \neg \coalnext{2} \neg p$ is satisfiable.
\end{example}

\begin{example}[continued]
  Consider the initial tableau for the formula $\theta_2$.  We have to
  eliminate state $\D'_2$ due to \Rule{E1}, as it contains a patent
  inconsistency.  For the same reason, we have to eliminate $\D'_3$.
  Furthermore, state $\D'_4$ gets eliminated due to \Rule{E3} since it
  contains an eventuality $\coal{2} p \until q$, but does not get
  marked with respect to it, as the only consistent state reachable
  from $\D'_4$ along the ``relevant'' move vector $(0, 1)$, which is
  $\D'_4$ itself, does not contain $q$.  Then, $\D'_1$ has to be
  eliminated, as both states reachable form it along the move vector
  $(0, 1)$ have been eliminated.  Thus, all the states containing the
  input formula, namely $\D'_1$ and $\D'_2$, are eliminated from the
  tableau.  Therefore, the final tableau for $\theta_2$ is closed and,
  hence, $\theta_2 = \coalbox{1} \neg q \con \coal{2} p \until q$ is
  unsatisfiable.
\end{example}

%%%%%%%%%%%%%%%%%%%%%%%%%%%%%%%%%%
\subsection{Incremental tableaux for CTL}
\label{sec:tableauCTL}
%%%%%%%%%%%%%%%%%%%%%%%%%%%%%%%%%%

The branching-time logic \CTL\ can be regarded as the one-agent
version of \ATL, where $\coal{\emptyset}$ is the universal path
quantifier and $\coal{1}$ is the existential path quantifier. Thus,
after due simplifications (notably, of the rule \Rule{Next}), our
tableau method produces an incremental tableau procedure for \CTL,
which is practically more efficient (in the average case) than Emerson
and Halpern's top-down tableau from~\cite{EmHal85}.

%%%%%%%%%%%%%%%%%%%%%%%%%%%%%%%%%%
\subsection{Complexity of the procedure}
\label{sec:complexity}
%%%%%%%%%%%%%%%%%%%%%%%%%%%%%%%%%%

We now estimate the complexity of the tableau procedure described
above.  As before, let $n = \card{\theta}$, $k =
\card{\agents_{\theta}}$, and let $c$ be the constant from the
equation $\card{\ecl{\theta}} = c \cdot \card{\theta}$ (recall
Lemma~\ref{lm:size_of_ecl}).

As we have seen, the costs of the construction phase and of the
prestate elimination phase are, respectively, \bigo{2^{(k+2) \log(cn)
    + 2cn}} and \bigo{2^{3cn}} steps.  It, thus, remains to estimate
the time required for the state elimination phase. During that phase,
we first apply \Rule{E1} to every state of the initial tableau.  To do
that, we need to go through \bigo{2^{cn}} states, and for each formula
$\vp$ of each state $\D$ check whether $\neg \vp \in \D$; this can be
done in time $\bigo{2^{cn} \times (cn)^2} = \bigo{2^{2 \log(cn) +
    cn}}$.

Next, we embark on the sequence of dovetailed applications of
\Rule{E3} and \Rule{E2}. We do it in cycles, whose number is bounded
by $\bigo{2^{cn}}$, each cycle involving going through all the
eventualities, whose number is bounded by \bigo{cn}.  For each
eventuality $\xi$, we have to, first, run the marking procedure with
respect to $\xi$ and then remove, as prescribed by \Rule{E3}, all the
relevant unmarked states; then, we apply the procedure implementing
\Rule{E2}.  The latter procedure can be carried out in $\bigo{2^{cn}
  \times (cn^k + cn)} = \bigo{2^{k \log(cn) + n} + 2^{\log (cn) + cn}}
= \bigo{2^{k \log(cn) + n}}$ steps, as we should go through
$\bigo{2^{cn}}$ states, doing the check for $\bigo{(cn)^k}$ moves
marking outgoing arrows, and possibly deleting \bigo{cn} formulas of
the state.  Since $k \leq n$, the cost of applying \Rule{E2} is
bounded by $\bigo{2^{n \log(cn) + n}} = \bigo{2^{n (\log (cn) + 1)}}$
steps.  As for the former, we need to compute the set of states
realizing $\xi$ in $\tableau{T}_n^{\theta}$, which can be done in
\bigo{2^{k \log (cn) + 3cn}} steps, as we do at most $\bigo{2^{cn}}$
``global'' status updates, each time updating the status for at most
\bigo{2^{cn}} states, each of these updates requiring looking at
$\bigo{(cn)^k}$ possible moves, which as several outgoing arrows can
be marked with the same move, can be repeated at most $\bigo{2^{cn}}$
times.  (For simplicity, we disregard the cost of applying deleting
states with unrealized eventualities, as its complexity, \bigo{2^{cn}
  \times cn}, is clearly dominated by the complexity of the marking
procedure.)  Thus, the whole sequence of dovetailed applications of
\Rule{E2} and \Rule{E3} requires $\bigo{(2^{cn} \times cn) \times
  (2^{k \log(cn) + n} + 2^{k \log (cn) + 3cn})} = \bigo{2^{(k+1) \log
    (cn) + 4cn}}$.

Thus, the overall complexity of our tableau procedure is
$\bigo{2^{(k+2) \log(cn) + 2cn}} + \bigo{2^{3cn}} + \bigo{2^{(k+1)
    \log (cn) + 4cn}}$.  As $k \leq n + 1$, this expression is bounded
by $\bigo{2^{n \log n + 5cn}} = \bigo{2^{2n \log n }} = \bigo{2^{2
    \card{\theta} \log \card{\theta}}}$.  This upper bound appears to
be better than the one claimed in~\cite{WLWW06} for the top-down
tableaux developed therein (namely, \bigo{2^{n^2}}); a more careful
analysis reveals, however, that the upper bound for tableaux
from~\cite{WLWW06} is within \bigo{2^{2n \log n }}, too.

%%%%%%%%%%%%%%%%%%%%%%%%%%%%%%%%%%
%%%%%%%%%%%%%%%%%%%%%%%%%%%%%%%%%%
\section{Soundness and completeness}
\label{sec:atl-tableaux-sound-and-complete}
%%%%%%%%%%%%%%%%%%%%%%%%%%%%%%%%%%
%%%%%%%%%%%%%%%%%%%%%%%%%%%%%%%%%%

We now prove that the tableau procedure described above is sound and
complete with respect to \ATL\ semantics as defined in
section~\ref{sec:semantics}; in algorithmic terminology, we show that
the procedure is correct.

%%%%%%%%%%%%%%%%%%%%%%%%%%%%%%%%%%
\subsection{Soundness}
\label{sec:soundness}
%%%%%%%%%%%%%%%%%%%%%%%%%%%%%%%%%%

Technically, soundness of a tableau procedure amounts to claiming that
if the input formula $\theta$ is satisfiable, then the final tableau
$\tableau{T}^{\theta}$ is open.

Before going into the technical details, we give an informal outline
of the proof.  The tableau procedure for the input formula $\theta$
starts off with creating a single prestate \set{\theta}.  Then, we
unwind \set{\theta} into states, each of which contains $\theta$.  To
establish soundness, it suffices to show that at least one these
states survives to the end of the procedure and is, thus, part of a
final tableau.

We start out by showing (Lemma~\ref{lem:propositional_consistency})
that if a prestate $\G$ is satisfiable, then at least one state
created from $\G$ using \Rule{SR} is also satisfiable.  In
particular, it ensures that if $\theta$ is satisfiable, then so is at
least one state obtained by \Rule{SR} form $\set{\theta}$.  To ensure
soundness, it is enough to show that this state never gets eliminated
from the tableau.

To that end, we first show (Lemma~\ref{lem:sat}) that, given a
satisfiable state $\D$, all the prestates created from $\D$ by
\Rule{Next}---each of which is associated with a move vector, say
$\sigma$---are satisfiable; according to
Lemma~\ref{lem:propositional_consistency}, each of these prestates
will give rise to at least one satisfiable state.  It follows that, if
a tableau state $\D$ is satisfiable, then for every move vector
$\sigma$ at $\D$, in the initial tableau, $\D$ will have at least one
satisfiable successor reachable by an arrow marked with $\sigma$;
hence, if $\D$ is satisfiable, it will not be eliminated on account of
\Rule{E2}.  Lastly, we show that no satisfiable states contain
unrealized eventualities (in the sense of
Definitions~\ref{def:realization_until} and
\ref{def:realization_neg_box}), and thus cannot be removed from the
tableau on account of \Rule{E3}.  Thus, we show that a satisfiable
state of the pretableau cannot be removed on account of any of the
state elimination rules and, therefore, survives to the end of the
procedure. In particular, this means that at least one state obtained
from the initial prestate $\theta$, and thus containing $\theta$,
survives to the end of the procedure---hence, the final tableau for
$\theta$ is open, as desired.

We start with the lemma that essentially asserts that the
``state-creation'' component of our tableaux preserves satisfiability.

\begin{lemma}
  \label{lem:propositional_consistency}
  Let $\G$ be a prestate of $\tableau{P}^{\theta}$ and let
  \sat{M}{s}{\G} for some CGM \mmodel{M} and some $s \in \mmodel{M}$.
  Then, \sat{M}{s}{\D} holds for at least one $\D \in \st{\G}$.
\end{lemma}

\begin{proof}
  Straightforward (see a remark at the end of
  section~\ref{sec:alpha-beta-atl}, though).
\end{proof}

The next lemma shows that \Rule{Next} creates from satisfiable states
satisfiable prestates (to see this, compare the condition of the lemma
with Remark~\ref{rem:at_most_one_neg_fma}).

\begin{lemma}
  \label{lem:sat}
  Let $\Phi = \set{\coalnext{A_1} \vp_1, \ldots, \coalnext{A_{m}}
    \vp_{m}, \neg \coalnext{A'} \psi}$ be a set of formulae such that
  $A_i \inter A_j = \emptyset$ for every $1 \leq i, j \leq {m}$ and
  $A_i \subseteq A'$ for every $1 \leq i \leq {m}$.  Let
  \sat{M}{s}{\Phi} for some CGM \mmodel{M} and $s \in \mmodel{M}$.
  Let, furthermore, $\move{A_i} \in \moves{A_i}{s}$ be an $A_i$-move
  witnessing the truth of $\coalnext{A_i} \vp_i$ at $s$, for each $1
  \leq i \leq {m}$, and let, finally, $\comove{A'} \in
  \comoves{A'}{s}$ be a co-$A'$-move witnessing the truth of $\neg
  \coalnext{A'} \psi$ at $s$.  Then, there exists $s' \in out (s,
  \move{A_1}) \inter \ldots \inter out (s, \move{A_m}) \inter out (s,
  \comove{A'})$ such that \sat{M}{s'}{\set{\vp_1, \ldots, \vp_{m},
      \neg \psi}}.
\end{lemma}

\begin{proof}
  As $A_i \inter A_j = \emptyset$ for every $1 \leq i, j \leq {m}$ and
  $A_i \subseteq A'$ for every $1 \leq i \leq {m}$, all the moves
  \move{A_i}, where $1 \leq i \leq m$, can be ``fused'' into a move
  \move{A_1 \union ... \union A_{m}}.  Then, the application of the
  co-move $\comove{A'}$ to any extension of \move{A_1 \union
    ... \union A_{m}} to a move of the coalition $\agents_{\theta}
  \setminus A' \supseteq A_1 \union ... \union A_{m}$ produces a move
  vector $\move{}$ such that $s' = \delta(s, \move{})$ satisfies both
  properties from the statement of the lemma.
\end{proof}

The preceding two lemmas show that from satisfiable (pre)states we
produce satisfiable (pre)states.  This, in particular, implies two
things: first, at least one of the states containing the input formula
$\theta$ is satisfiable and, second, satisfiable states never get
eliminated due to \Rule{E2}.  It is also clear that a satisfiable
state can not contain a propositional inconsistency and thus be
removed due to \Rule{E1}.

Therefore, all that remains to show is that \Rule{E3} does not
eliminate from tableaux satisfiable states.  To that end, we will need
some extra definitions and pieces of notations drawing analogies
between what happens in CGMs and tableaux
(Definition~\ref{def:tableau_outcome_set} through Notational
convention~\ref{not:tableaux}).

In what follows, we treat labels of the arrows of the tableaux as move
vectors; the concepts of $A$-move, and all the concomitant definitions
and notation are then used in exactly the same way as for CGFs (see
section~\ref{sec:CGS}); analogously for co-$A$-moves (see
section~\ref{sec:goranko-drimmelen-semantics}).  We only explicitly
mention what notion (i.e., the one relating to the semantics of \ATL\
or to tableaux) is referred to if the context leaves room for
ambiguity. The only notion that differs between \ATL-semantics and the
\ATL-tableaux is that of ``outcome'' of (CGF vs. tableau) moves and
co-moves. Unlike the former, the latter are generally not unique, as
there might be several outgoing arrows from a state $\D$ labeled with
the same ``move vector'' $\move{}$.  We, however, define an outcome
set of a tableau $A$-move \move{A} to contain exactly one state
obtained from $\D$ by following a given $\move{} \sqsupseteq \move{A}$
to make them resemble outcomes of $A$-moves in CGFs.

\begin{definition}
  \label{def:tableau_outcome_set}
  Let $\D \in S_n^{\theta}$ and $\move{A} \in D_A (\D)$. An
  \de{outcome set} of $\move{A}$ at $\D$ is a minimal set of states $X
  \subseteq S_n^{\theta}$ such that, for every $\sigma \sqsupseteq
  \move{A}$, there exists exactly one $\D' \in X$ such that $\D
  \movearrow \D'$.
\end{definition}

Outcome sets for tableau co-moves are defined analogously:

\begin{definition}
  \label{def:tableau_outcome_set_comoves}
  Let $\D \in S_n^{\theta}$ and $\comove{A} \in D^c_A (\D)$. An
  \de{outcome set} of $\comove{A}$ at $\D$ is a minimal set of states
  $X \subseteq S_n^{\theta}$ such that, for every $\move{A} \in
  \moves{A}{\D}$, there exists exactly one $\D' \in X$ such that $\D
  \stackrel{\comove{A}(\move{A})}{\longrightarrow} \D'$.
\end{definition}

\begin{notation}
  \label{not:tableaux}
  \mbox{}
  \begin{enumerate}
  \item Whenever we write $\coalnext{A_p} \vp_p \in \D \in
    S_n^{\theta}$, we mean that $\coalnext{A_p} \vp_p$ is the $p$-th
    formula in the linear ordering of the next-time formulae of $\D$
    induced as part of applying the \Rule{Next} rule to $\D$.  We use
    the notation $\neg \coalnext{A'_q} \psi_q \in \D \in S_n^{\theta}$
    in an analogous way.
  \item Given $\coalnext{A_p} \vp_p \in \D \in S_n^{\theta}$, by
    $\move{A_p}[\coalnext{A_p} \vp_p]$ we denote (the unique) tableau
    $A_p$-move $\move{A_p} \in \moves{A_p}{\D}$ such that
    $\move{A_p}(a) = p$ for every $a \in A_p$.
  \item Given a proper $\neg \coalnext{A'_q} \psi_q \in \D \in
    S_n^{\theta}$, by $\comove{A'_q}[\neg \coalnext{A'_q} \psi_q]$ we
    denote (the unique) tableau co-$A'_q$-move satisfying the
    following condition: $\mathbf{neg}(\comove{A'_q}(\move{A'_q})) =
    q$ and $\agents_{\theta} - A'_q \subseteq
    N(\comove{A'_q}(\move{A'_q}))$ for every $\move{A'_q} \in
    \moves{A'_q}{\D}$.
  \end{enumerate}
\end{notation}

We now get down to proving that \Rule{E3} does not eliminate any
satisfiable states.  We need to show that if a tableau
$\tableau{T}^{\theta}_n$ contains a state $\D$ that is satisfiable and
contains an eventuality $\xi$, then $\xi$ is realized at $\D$. This
will be accomplished by showing that $\tableau{T}^{\theta}_n$
``contains'' a structure (more precisely, a tree) that, in a sense to
be made precise, ``witnesses'' the realization of $\xi$ at $\D$ in
$\tableau{T}^{\theta}_n$.  This tree will, in a sense to be made
precise, emulate a tree of runs effected by a strategy or co-strategy
that ``realizes'' an eventuality in a model.  This simulation is going
to be carried out step-by-step, each step, i.e. $A$-move (in the case
of $\coal{A} \vp \until \psi$) or co-$A$-move (in the case of $\neg
\coalbox{A} \vp$) will be simulated by a tableau move or co-move
associated with a respective eventuality.  That this step-by-step
simulation can be done is proved in the next two lemmas (together with
their corollaries).

\begin{lemma}
  \label{lem:outcome-sat}
  Let $\coalnext{A_p} \vp_p \in \D \in S_n^{\theta}$ and let
  \sat{M}{s}{\D} for some CGM $\mmodel{M}$ and state $s \in
  \mmodel{M}$.  Let, furthermore, $\move{A_p} \in \moves{A_p}{s}$ be
  an $A_p$-move witnessing the truth of $\coalnext{A_p} \vp_p$ at
  $s$.  Then, there exists in $\tableau{T}_n^{\theta}$ an outcome set
  $X$ of $\move{A_p}[\coalnext{A_p} \vp_p]$ such that for each $\D'
  \in X$ there exists $s' \in out (s, \move{A_p})$ such that
  $\sat{M}{s'}{\D'}$.
\end{lemma}

\begin{proof}
  Consider the set of prestates $Y = \crh{\G \in \sucpr{\D}}{\D
    \movearrow \G \text{ for some } \sigma \sqsupseteq
    \move{A_p}[\coalnext{A_p} \vp_p]}$. Take an arbitrary $\G \in Y$.
  It follows immediately from the \Rule{Next} rule (see
  Remark~\ref{rem:at_most_one_neg_fma}) that $\G$ (which must contain
  $\vp_p$) is either of the form \set{\vp_1, \ldots, \vp_{m}, \neg
    \psi}, where
  $$\set{\coalnext{A_1} \vp_1, \ldots, \coalnext{A_{m}} \vp_{m}, \neg
    \coalnext{A'} \psi} \subseteq \D$$ satisfies the condition of
  Lemma~\ref{lem:sat}, or of the form \set{\vp_1, \ldots, \vp_{m}},
  where $$\set{\coalnext{A_1} \vp_1, \ldots, \coalnext{A_{m}} \vp_{m}}
  \subseteq \D$$ and $A_i \inter A_j = \emptyset$ for every $1 \leq i,
  j \leq {m}$.

  As \sat{M}{s}{\D}, in the former case, by Lemma~\ref{lem:sat}, there
  exists $s' \in out (s, \sigma_{A_p})$ with \sat{M}{s'}{\G}. Then
  $\G$ can be extended to a downward saturated set $\D'$ containing at
  least one next-time formula ($\coalnext{\agents_{\theta}} \truth$ if
  nothing else) such that \sat{M}{s'}{\D'}. This is done by choosing,
  for every $\beta$-formula to be dealt with, the ``disjunct'' that is
  actually true in \mmodel{M} at $s'$ (if both ``disjuncts'' happen to
  be true at $s'$, the choice is arbitrary).

  In the latter case, the same conclusion follows from
  Lemma~\ref{lem:sat} again, by adding to $\D$ the valid formula $\neg
  \coalnext{\agents_{\theta}} \bot$.

  To complete the proof of the lemma, take $X$ to be the set of all
  tableau states $\D'$ obtainable from the prestates in $Y$ in the way
  described above.
\end{proof}

\begin{corollary}
  \label{cor:beta_out_sat}
  Let $\coalnext{A_p} \vp_p \in \D \in S_n^{\theta}$ and let
  \sat{M}{s}{\D} for some CGM $\mmodel{M}$ and state $s \in
  \mmodel{M}$.  Let, furthermore, $\move{A_p} \in \moves{A_p}{s}$ be
  an $A_p$-move witnessing the truth of $\coalnext{A_p} \vp_p$ at $s$
  and let $\chi \in \ecl{\theta}$ be a $\beta$-formula, whose
  $\beta_i$-associate ($i \in \set{1, 2}$) is $\chi_i$.  Then, there
  exists in $\tableau{T}_n^{\theta}$ an outcome set $X_{\chi_i}$ of
  $\move{A_p}[\coalnext{A_p} \vp_p]$ such that for every $\D' \in
  X_{\chi_i}$ there exists $s' \in out (s, \move{A_p})$ such that
  $\sat{M}{s'}{\D'}$, and moreover, if $\sat{M}{s'}{\chi_i}$, then
  $\chi_i \in \D'$.
\end{corollary}

\begin{proof}
  Construct $X_{\chi_i}$ in a way $X$ was constructed in the proof of
  the preceding lemma, with a single modification: when dealing with
  the formula $\chi$, instead of choosing arbitrarily between $\chi_1$
  and $\chi_2$, choose $\chi_i$ whenever it is true at $s'$.
\end{proof}

\begin{lemma}
  \label{lem:outcome-sat-comoves}
  Let $\neg \coalnext{A'_q} \psi_q \in \D \in S_n^{\theta}$ and let
  \sat{M}{s}{\D} for some CGM $\mmodel{M}$ and state $s \in
  \mmodel{M}$.  Let, furthermore, $\comove{A'_q} \in
  \comoves{A'_q}{s}$ be a co-$A'_q$-move witnessing the truth of $\neg
  \coalnext{A'_q} \psi_q$ at $s$.  Then, there exists in
  $\tableau{T}_n^{\theta}$ an outcome set $X$ of $\comove{A'_q}[\neg
  \coalnext{A'_q} \psi_q]$ such that for each $\D' \in X$ there exists
  $s' \in out (s, \comove{A'_q})$ such that $\sat{M}{s'}{\D'}$.
\end{lemma}

\begin{proof}
  Consider the set of prestates $Y = \crh{\G \in \sucpr{\D}}{\D
    \movearrow \G, \sigma = \comove{A'_q}[\neg \coalnext{A'_q} \psi_q]
    (\move{A'_q}) \linebreak \text{ for some } \move{A'_q} \in
    \moves{A'_q}{\D}}$. Take an arbitrary $\G \in Y$.  It follows
  immediately from the \Rule{Next} rule (see
  Remark~\ref{rem:at_most_one_neg_fma}) that $\G$ (which must contain
  $\neg \psi_q$) is either of the form \set{\vp_1, \ldots, \vp_{m},
    \neg \psi_q}, where
  $$\set{\coalnext{A_1} \vp_1, \ldots, \coalnext{A_{m}} \vp_{m}, \neg
    \coalnext{A'_q} \psi_q} \subseteq \D$$ satisfies the condition of
  Lemma~\ref{lem:sat}, or of the form \set{\neg \psi_q}.

  As \sat{M}{s}{\D}, in the former case, by Lemma~\ref{lem:sat}, there
  exists $s' \in out (s, \comove{A'_q})$ with \sat{M}{s'}{\G}. Then
  $\G$ can be extended to a downward saturated set $\D'$ containing at
  least one next-time formula ($\coalnext{\agents_{\theta}} \truth$ if
  nothing else) such that \sat{M}{s'}{\D'}. This is done by choosing,
  for every $\beta$-formula to be dealt with, the ``disjunct'' that is
  actually true in \mmodel{M} at $s'$ (if both ``disjuncts'' are true,
  choose arbitrarily).

  In the latter case, the same conclusion follows from
  Lemma~\ref{lem:sat} again, by adding to $\D$ the valid formula
  $\coalnext{\emptyset}\truth.$

  To complete the proof of the lemma, take $X$ to be the set of all
  tableau states $\D'$ obtainable from the prestates in $Y$ in the way
  described above.
\end{proof}

\begin{corollary}
  \label{cor:beta_out_sat_comoves}
  Let $\neg \coalnext{A'_q} \psi_q \in \D \in S_n^{\theta}$ and let
  \sat{M}{s}{\D} for some CGM $\mmodel{M}$ and state $s \in
  \mmodel{M}$.  Let, furthermore, $\comove{A'_q} \in
  \comoves{A'_q}{s}$ be a co-$A'_q$-move witnessing the truth of $\neg
  \coalnext{A'_q} \psi_q$ at $s$ and let $\chi \in \ecl{\theta}$ be a
  $\beta$-formula, whose $\beta_i$-associate ($i \in \set{1, 2}$) is
  $\chi_i$.  Then, there exists in $\tableau{T}_n^{\theta}$ an outcome
  set $X_{\chi_i}$ of $\comove{A'_q}[\neg \coalnext{A'_q} \psi_q]$
  such that for every $\D' \in X_{\chi_i}$ there exists $s' \in out
  (s, \comove{A'_q})$ such that $\sat{M}{s'}{\D'}$, and moreover, if
  $\sat{M}{s'}{\chi_i}$, then $\chi_i \in \D'$.
\end{corollary}

\begin{proof}
  Analogous to the proof of Corollary~\ref{cor:beta_out_sat}.
\end{proof}

We now show that the tableau moves (for eventualities of the form
$\coal{A} \vp \until \psi$) and co-moves (for eventualities of the
form $\neg \coalbox{A} \vp$) whose existence was established in the
preceding two lemmas can be stitched together into what we call
eventuality realization witness trees\footnote{In the context of this
  paper, by a tree we mean any directed, connected, and acyclic graph,
  each node of which, except one, the root, has exactly one incoming
  edge.}.  Theses trees, as already mentioned, simulate trees of
runs effected in models by (co-)strategies. It will then follow that
the existence of such a tree for a state $\D$ means that it cannot be
removed from a tableau due to \Rule{E3}.

\begin{definition}
  \label{def:colored_trees}
  Let $\tree{R} = (R, \rightarrow)$ be a tree and $X$ be a non-empty
  set.  An \de{$X$-coloring} of \tree{R} is a mapping $c: R \mapsto
  X$. When such mapping is fixed, we say that \tree{R} is
  \de{$X$-colored}.
\end{definition}

\begin{definition}
  \label{def:realisation_witness_tree_until}
  A \de{realization witness tree for the eventuality $\coal{A} \vp
    \until \psi$ at state $\D \in S_n^{\theta}$} is a finite
  $S_n^{\theta}$-colored tree $\mathcal{R} = (R, \rightarrow)$ such
  that
  \begin{enumerate}
  \item the root of $\mathcal{R}$ is colored with $\D$;
  \item if an interior node of $\mathcal{R}$ is colored with $\D'$,
    then $\set{\vp, \coalnext{A} \coal{A} \vp \until \psi, \coal{A}
      \vp \until \psi} \subseteq \D'$;
  \item for every interior node $w$ of $\mathcal{R}$ colored with
    $\D'$, the children of $w$ are colored bijectively with
    the states from an outcome set of $\move{A}[\coalnext{A} \coal{A}
    \vp \until \psi] \in D_A(\D')$;
  \item if a leaf of $\mathcal{R}$ is colored with $\D'$, then
    $\set{\psi, \coal{A} \vp \until \psi} \subseteq \D'$.
  \end{enumerate}
\end{definition}

\begin{definition}
  \label{def:realisation_witness_tree_negbox}
  A \de{realization witness tree for the eventuality $\neg \coalbox{A}
    \vp$ at state $\D \in S_n^{\theta}$} is a finite
  $S_n^{\theta}$-colored tree $\mathcal{R} = (R, \rightarrow)$ such
  that
  \begin{enumerate}
  \item the root of $\mathcal{R}$ is colored with $\D$;
  \item if an interior node of $\mathcal{R}$ is colored with $\D'$,
    then $\set{\neg \coalnext{A} \coalbox{A} \vp, \neg \coalbox{A}
      \vp} \subseteq \D'$;
  \item for every interior node $w$ of $\mathcal{R}$ colored with
    $\D'$, the children of $w$ are colored bijectively with the
    states from an outcome set of $\comove{A}[\neg \coalnext{A}
    \coalbox{A} \vp]$;
  \item if a leaf of $\mathcal{R}$ is colored with $\D'$, then
    $\set{\neg \vp, \neg \coalbox{A} \vp} \subseteq \D'$.
  \end{enumerate}
\end{definition}

\begin{lemma}
  \label{lm:realization_in_the_tree}
  Let $\mathcal{R} = (R, \rightarrow)$ be a realization witness tree
  for an eventuality $\xi$ at $\D \in S_n^{\theta}$.  Then, $\xi$ is
  realized in $\tableau{T}_n^{\theta}$ at every $\D'$ coloring a node
  of $R$---in particular, at $\D$ in $\tableau{T}_n^{\theta}$.
\end{lemma}

\begin{proof}
  Straightforward induction on the length of the longest path from a
  node colored by $\D'$ to a leaf of $\mathcal{R}$ (recall that
  realization of eventualities was defined in
  Definitions~\ref{def:realization_until} and
  \ref{def:realization_neg_box}).
\end{proof}

We now prove the existance of realization witness trees for
satisfiable states of tableaux containing eventualities.

\begin{lemma}
  \label{lem:existence_of_rwt_soundness}
  Let $\xi \in \D$ be an eventuality formula and $\D \in S_0^{\theta}$
  be satisfiable.  Then there exists a realization witness tree
  $\tree{R} = (R, \rightarrow)$ for $\xi$ at $\D \in S_0^{\theta}$.
  Moreover, every $\D'$ coloring a node of $R$ is satisfiable.
\end{lemma}

\begin{proof}
  We only supply the full proof for eventualities of the form
  $\coal{A} \vp \until \psi$; we then indicate how to obtain the proof
  for eventualities of the form $\neg \coalbox{A} \vp$.

  If $\psi \in \D$, then we are done straight off---the realization
  witness tree is made up of a single node, the root, colored with
  $\D$. Hence, we only need to consider the case when $\psi \notin
  \D$. As $\Delta$ is downward saturated, then $\set{\vp, \coalnext{A}
    \coal{A} \vp \until \psi} \subseteq \D$.

  So, suppose that \sat{M}{s}{\D}; in particular, \sat{M}{s}{\vp} and
  \sat{M}{s}{\coalnext{A} \coal{A} \vp \until \psi}.  The latter means
  that there exists $\move{A} \in \moves{A}{s}$ such that $s' \in out
  (s, \move{A})$ implies \sat{M}{s'}{\coal{A} \vp \until \psi}.  Now,
  $\coalnext{A} \coal{A} \vp \until \psi$ is a positive next-time
  formula.  Since $\D$ is satisfiable, it does not contain a patent
  inconsistency; hence, the \Rule{Next} rule has been applied to it.
  As part of that application, $\coalnext{A} \coal{A} \vp \until \psi$
  has been assigned a place, say $p$, in the linear ordering of the
  next-time formulae of $\D$.  Furthermore, $\coal{A} \vp \until \psi$
  is a $\beta$-formula whose $\beta_2$ is $\psi$.  Therefore,
  Corollary~\ref{cor:beta_out_sat} is applicable to $\D$, $\chi =
  \coal{A} \vp \until \psi$, $\chi_1 = \coalnext{A} \coal{A} \vp
  \until \psi$, and $\chi_2 = \psi$. According to that corollary,
  there exists an outcome set $X_{\psi}$ of $\move{A} [\coalnext{A}
  \coal{A} \vp \until \psi]$ at $\D$ such that, for every $\D' \in
  X_{\psi}$, there exists $s' \in out (s, \move{A})$ such that
  \sat{M}{s'}{\D'} and, moreover, if \sat{M}{s'}{\psi}, then $\psi \in
  \D'$.  We start building the witness tree \tree{R} by constructing a
  simple tree (i.e., one with a single interior node, the root) whose
  root $r$ is colored with $\D$ and whose leaves are colored, in the
  way prescribed by
  Definition~\ref{def:realisation_witness_tree_until}, by the states
  from $X_{\psi}$.

  Next, since \sat{M}{s'}{\coal{A} \vp \until \psi} for every $s' \in
  out (s, \move{A})$, it follows that for every such $s'$ there exists
  a (perfect-recall) $A$-strategy $\str{A}^{s'}$ such that for every
  $\lambda \in out (s', \str{A}^{s'})$ there exists $i \geq 0$ with
  \sat{M}{\lambda[i]}{\psi} and \sat{M}{\lambda[j]}{\vp} holds for all
  $0 \leq j < i$.  Then, playing $\move{A}$ followed by
  playing $\str{A}^{s'}$ for the $s' \in out (s, \move{A})$ ``chosen''
  by the counter-coalition $\agents_{\theta} \setminus A$ constitutes
  a (perfect-recall) strategy \str{A} witnessing the truth of
  $\coal{A} \vp \until \psi$ at $s$.

  We, then, continue the construction of \tree{R} as follows.  For
  every $s' \in out (s, \move{A})$ (each such $s'$ has been matched by
  a node of \tree{R} at the initial stage of the construction of
  \tree{R}), we follow the (perfect-recall) strategy $\str{A}^{s'}$,
  matching every state $s''$ appearing as part of a run
  compliant with $\str{A}^{s'}$ and satisfying the requirement that
  $\notsat{M}{s''}{\psi}$ with a node $w''$ of \tree{R} and matching
  every $A$-move of $\str{A}^{s'}$ at $s''$ with the $A$-move in the
  tableau $\move{A} [\coalnext{A} \coal{A} \vp \until \psi] \in
  \moves{A}{\D''}$ for the state $\D''$ coloring the node $w''$.  In
  this way, we follow each $\str{A}^{s'}$ along each run, up to the
  point when we reach a state $t$ where $\psi$ is true; at that point
  we reach the leaf of the respective branch of the tree we are
  building, as by construction, the node associated with $t$ will be
  colored with a state containing both $\psi$ and $\coal{A} \vp
  \until \psi$.

  In the manner outlined above, we are guaranteed to build a tree
  satisfying all conditions of
  Definition~\ref{def:realisation_witness_tree_until}.  Indeed, the
  very way the tree is built guarantees that conditions (1-4) of that
  definition hold.  As for finiteness, assuming that the resultant
  tree is infinite implies that it contains an infinite branch
  colored with sets not containing $\psi$, which in turn implies the
  existence of $\lambda \in out (s, \str{A})$ such that for every $i
  \geq 0$ we have \sat{M}{\lambda[i]}{\lnot \psi}, which contradicts
  the fact that $\str{A}$ is a strategy witnessing the truth of
  $\coal{A} \vp \until \psi$ at $s$.

  Thus, we have obtained a realization witness tree \tree{R} for
  $\coal{A} \vp \until \psi$ at $\D$ in $\tableau{T}_n^{\theta}$.
  Moreover, it is clear from the way this tree has been built that
  every state coloring a node of \tree{R} is satisfiable (in
  \mmodel{M}).

  The proof for eventualities of the form $\neg \coalbox{A} \vp$ is
  completely analogous, with reference to
  Corollary~\ref{cor:beta_out_sat_comoves} rather than
  Corollary~\ref{cor:beta_out_sat}, using the fact that $\neg
  \coalbox{A} \vp$ is a $\beta$-formula, with $\beta_1 = \neg \vp$
  and $\beta_2 = \neg \coalnext{A} \coalbox{A} \vp$.
\end{proof}

\begin{theorem}[Soundness]
\label{thr:soundness}
  If $\theta$ is satisfiable, then $\tableau{T}^{\theta}$ is open.
\end{theorem}

\begin{proof}
  We will prove that no satisfiable states are eliminated in the state
  elimination phase of the construction of the tableau.  The statement
  of the lemma will then follow immediately from
  Lemma~\ref{lem:propositional_consistency}, which implies that if the
  initial prestate $\set{\theta}$ is satisfiable, then at least one
  state of $\tableau{T}^{\theta}$ containing $\theta$ is also
  satisfiable.

  As the elimination process proceeds in stages, we will prove by
  induction on the number $n$ of stages that, for every $\D \in
  S_0^{\theta}$, if $\D$ is satisfiable, then \D\ will not be
  eliminated at stage $n$.

  The base case is trivial: when $n = 0$, no eliminations have yet
  been done, hence no satisfiable $\D$ has been eliminated.

  Now inductively assume that, if $\D' \in S_0^{\theta}$ is
  satisfiable, it has not been eliminated during the previous $n$
  stages of the elimination phase, and thus $\D' \in
  S_n^{\theta}$. Consider stage $n+1$ and a satisfiable $\D \in
  S_0^{\theta}$.  By inductive hypothesis, $\D \in S_n^{\theta}$.  We
  will now show that no elimination rule allows elimination of $\D$
  from $\tableau{T}_n^{\theta}$; hence, $\D$ will remain in
  $\tableau{T}_{n+1}^{\theta}$.

  \Rule{E1} As $\D$ is satisfiable, it clearly cannot contain
  both $\vp$ and $\neg \vp$; therefore, it cannot be eliminated from
  $\tableau{T}_n^{\theta}$ due to \Rule{E1}.

  \Rule{E2} Due to the form of the \Rule{Next}-rule (see
  Remark~\ref{rem:at_most_one_neg_fma}), it immediately follows from
  Lemma~\ref{lem:sat} that if $\D$ is satisfiable, then all the
  prestates in $\sucpr{\D}$ are satisfiable, too.  By virtue of
  Lemma~\ref{lem:propositional_consistency}, $\tableau{T}_0^{\theta}$
  contains for every $\sigma \in D(\D)$ at least one satisfiable $\D'$
  with $\D \movearrow \D'$.  By the inductive hypothesis, all such
  $\D'$ belong to $\tableau{T}_n^{\theta}$; thus, $\D$ can not be
  eliminated from $\tableau{T}_n^{\theta}$ due to \Rule{E2}.

  \Rule{E3} We need to show that if $\D$ is satisfiable and contains
  an eventuality $\xi$, then $\xi$ is realized at $\D$ in
  $\tableau{T}_n^{\theta}$.

  According to Lemma~\ref{lem:existence_of_rwt_soundness}, there
  exists a realization witness tree $\tree{R} = (R, \rightarrow)$ for
  $\xi$ at $\D$ in $\tableau{T}_0^{\theta}$ and every $\D'$ coloring
  a node of $R$ is satisfiable.  Therefore, by inductive hypothesis,
  each such $\D'$ belongs to $S^{\theta}_n$. Then, it is clear from
  the construction of $\tree{R}$ in the proof of
  Lemma~\ref{lem:existence_of_rwt_soundness}, that $\tree{R}$ will
  still be a realization witness tree for $\D$ in
  $\tableau{T}_n^{\theta}$.  Then, by virtue of
  Lemma~\ref{lm:realization_in_the_tree}, $\xi$ is realized at $\D$ in
  $\tableau{T}_n^{\theta}$, hence cannot be eliminated due to
  \Rule{E3}.
\end{proof}

%%%%%%%%%%%%%%%%%%%%%%%%%%%%%%%%%%
\subsection{Completeness}
\label{sec:completeness}
%%%%%%%%%%%%%%%%%%%%%%%%%%%%%%%%%%

Completeness of a tableau procedure means that if the final tableau
for the input formula $\theta$ is open, than $\theta$ is satisfiable.
The completeness proof presented in this section boils down to
building a Hintikka structure $\hintikka{H}_{\theta}$ for the input
formula $\theta$ out of the open tableau $\tableau{T}^{\theta}$.
Theorem~\ref{thr:from_hintikka_to_models} then guarantees the
existence of a model for $\theta$.

Our construction of a Hintikka structure $\hintikka{H}_{\theta}$ for
$\theta$ out of $\tableau{T}^{\theta}$ is going to resemble building a
house, when bricks are assembled into prefab blocks that are then
assembled into walls that are finally assembled into a complete
structure.  We will use analogues of all of those in our producing a
Hintikka structure for $\theta$.  Larger and larger components of our
construction will satisfy more and more conditions required by
Definition~\ref{def:cghs}, so that by the end, we are going to get a
fully-fledged Hintikka structure.

The ``bricks'' of $\hintikka{H}_{\theta}$ are going to be the states
of $\tableau{T}^{\theta}$. Being downward-saturated sets containing no
patent inconsistencies (otherwise, they would have been eliminated due
to \Rule{E1}), they satisfy conditions (H1)--(H3) of
Definition~\ref{def:cghs}.

The ``prefab blocks'' are going to be \fm{locally consistent simple
  $\tableau{T}^{\theta}$-trees}, which it is our next task to define.
Intuitively, these trees are one-step components of the Hintikka
structure we are building.

\begin{definition}
  \label{def:labelld_trees}
  Let $\tree{W} = (W, \arc)$ be a tree and $Y$ be a non-empty set.
  A $Y$-\de{labeling} of \tree{W} is a mapping $l$ from the set of
  edges of \tree{W} to the set of non-empty subsets of $Y$.
   When such mapping is fixed, we say that \tree{W} is \de{labeled by $Y$}.
\end{definition}

\begin{definition}
  A tree $\tree{W} = (W, \arc)$ is a \de{$\tableau{T}^{\theta}$-tree}
  if the following conditions hold:
  \begin{itemize}
  \item \tree{W} is $S^{\theta}$-colored (recall
    Definition~\ref{def:colored_trees}), by some coloring mapping
    $c$;
  \item \tree{W} is labeled by $\union_{\D \in S^{\theta}}
    \moves{}{\D}$, by some labeling mapping $l$;
  \item $l (w \arc w') \subseteq \moves{}{\D}$ for every $w \in W$ with
    $c(w) = \D$.
  \end{itemize}
\end{definition}

\begin{definition}
  \label{def:locally_consistent}
  A $\tableau{T}^{\theta}$-tree $\tree{W} = (W, \arc)$ is \de{locally
    consistent} if the following condition holds:
    \begin{quotation}
      For every interior node $w \in W$ with $c(w) = \D$ and every
      $\D$-successor $\D' \in S^{\theta}$, there exists exactly one
      $w' \in W$ such that $l (w \arc w')= \crh{\move{}}{\D \movearrow
        \D'}$.
    \end{quotation}
\end{definition}

That is, a locally consistent tree can not have two distinct
successors $w' = c(\D')$ and $w'' = c(\D'')$ of an interior node $w =
c(\D)$ such that $\crh{\move{}}{\D \movearrow \D'}= \crh{\move{}}{\D
  \movearrow \D''}$.  Note that we label edges of
$\tableau{T}^{\theta}$-trees with sets of move vectors as each edge in
a tableau can be marked by more than one move vector.

\begin{definition}
  \label{def:simple_trees}
  A tree $\tree{W} = (W, \arc)$ is \de{simple} if it has no interior
  nodes other than the root.
\end{definition}

Locally consistent simple $\tableau{T}^{\theta}$-trees will be our
building blocks for the construction of a Hintikka structure from an
open tableau $\tableau{T}^{\theta}$.  Essentially, we are extracting
from tableaux one-step structures that resemble CGMs in that every
interior node of these structures has exactly one outcome state
associated with a given move vector.  In other words, while an open
tableau encodes all possible Hintikka structures for the input
formula, we are extracting only one of them, by choosing the outcome
states associated with move vectors at each state out of possibly
several such outcomes.

We now prove the existence of locally consistent simple
$\tableau{T}^{\theta}$-trees associated with each state $\D$.

\begin{definition}
  Let $\D \in S^{\theta}$.  A $\tableau{T}^{\theta}$-tree \tree{W} is
  \de{rooted at $\D$} if the root of \tree{W} is colored with $\D$,
  i.e., $c(r) = \D$, where $r$ is the root of \tree{W}.
\end{definition}

\begin{lemma}
  \label{lem:existence_of_lcst}
  Let $\D \in S^{\theta}$.  Then, there exists a locally consistent
  simple $\tableau{T}^{\theta}$-tree rooted at $\D$.
\end{lemma}

\begin{proof}
  Such a tree can be built as follows: consider all successor states
  $\D'$ of $\D$ in $\tableau{T}^{\theta}$. With each of them is
  associated a non-empty set of ``move vectors'' \crh{\move{}}{\D
    \movearrow \D'}. The $\tableau{T}^{\theta}$-tree will then consist
  of a root $r$ colored with $\D$ and a leaf associated with each such
  set of move vectors, colored with any of the successor states $\D'$
  with which this particular set of moves is associated (note that, in
  general, a tableau can contain more than one such $\D'$); the edge
  between the root $r$ and a leaf $t$ is then labeled by the set of
  moves \crh{\move{}}{\D \movearrow c(t)}. Note that, by construction
  of the tableau, different successor states $\D'$ of $\D$ are
  reachable from $\D$ by pairwise disjoint sets of moves.
\end{proof}

The next lemma essentially asserts that, in addition to conditions
(H1)--(H3), locally consistent simple $\tableau{T}^{\theta}$-trees
also satisfy conditions (H4)--(H5) of Definition~\ref{def:cghs}, where
outcomes of $A$-moves and co-$A$-moves are defined for such trees as
for CGFs; recall that edges of these trees are labeled with sets of
move vectors. Thus, locally consistent simple
$\tableau{T}^{\theta}$-trees are closely approximating Hintikka
structures, but so far only \emph{locally}.

\begin{lemma}
  \label{lem:local_consistency}
  Let \tree{S} be a locally consistent simple
  $\tableau{T}^{\theta}$-tree rooted at $\D$.  Then, the following
  hold:
  \begin{enumerate}
  \item If $\coalnext{A} \vp \in \D = c(w)$, then there exists an
    $A$-move $\move{A} \in \moves{A}{w}$ such that $\vp \in \D'$ for
    all $\D' = c(w') \in out(\D, \move{A})$.
  \item If $\neg \coalnext{A} \vp \in \D = c(w)$, then there exists a
    co-$A$-move $\comove{A} \in D^c_A (w)$ such that $\neg \vp \in
    \D'$ for all $\D' = c(w') \in out(\D, \comove{A})$.
  \end{enumerate}
\end{lemma}

\begin{proof}
  Note that every $\D \in S^{\theta}$ is not patently
  inconsistent. Therefore, we can assume throughout the proof that all
  next-time formulae of $\D$ have been linearly ordered as part of
  applying the \Rule{Next} rule to $\D$.

  (1) Suppose that $\coalnext{A} \vp \in \D$.  Then the required
  $A$-move is $\move{A} [\coalnext{A} \vp]$ (recall Notational
  convention~\ref{not:tableaux}).  Indeed, it immediately follows
  from the rule \Rule{Next} that for every $\move{} \sqsupseteq
  \move{A} [\coalnext{A} \vp]$ in the pretableau
  $\tableau{P}^{\theta}$, if $\D \movearrow \G$ then $\vp \in \G$.
  Now, in $\tableau{T}^{\theta}$ we have $\D \movearrow \D'$ only if
  in $\tableau{P}^{\theta}$ we had $\D \movearrow \G$ for some $\G
  \subseteq \D'$.  Therefore, $\vp \in \D'$ for every $\D'$ in any
  outcome set of $\move{A} [\coalnext{A} \vp]$ at $\D$, and the
  statement of the lemma follows.

  (2) Suppose that $\neg \coalnext{A} \vp \in \D$.  We have two cases
  to consider.

  Case 1: $A \ne \agents_{\theta}$.  Therefore, there exists $b \in
  \agents_{\theta} \setminus A$ and, furthermore, $\neg \coalnext{A}
  \vp$ occupies some place, say $q$, in the linear ordering of the
  next-time formulae of $\D$.  Consider an arbitrary $\move{A} \in D_A
  (\D)$. We claim that \move{A} can be extended to $\move{}' \supmove
  \move{A}$ such that $\D \stackrel{\move{}'}{\longrightarrow} \D'$
  and $\neg \vp \in \D'$ for some $\D'$. To show that, denote by
  $N(\move{A})$ the set \crh{i}{\move{A} (i) \geq m}, where $m$ is the
  number of positive next-time formulae in $\D'$, and by $\mathbf{neg}
  (\move{A})$ the number $\left(\sum_{i \in N(\move{A})} (\move{A}(i)
    - m))\right)\!\!\!\!\mod l$, where $l$ is the number of negative
  next-time formulae in $\D$.  Now, consider $\move{}' \supmove
  \move{A}$ defined as follows: $\move{b}' = (\left(q - \mathbf{neg}
    (\move{A})\right)\!\!\!\! \mod l) + m$ and $\move{a'}' = m$ for
  any $a' \in \agents_{\theta} \setminus (A \union \set{b})$.  It is
  easy to see that $\agents_{\theta} \setminus A \subseteq
  N(\move{}')$, and moreover, that $\mathbf{neg} (\move{}') =
  \left(\mathbf{neg} (\move{A}) + (q - \mathbf{neg}
    (\move{A}))\right)\!\!\!\!\mod l = q$.  We conclude that in the
  pretableau $\tableau{P}^{\theta}$, if $\D
  \stackrel{\move{}'}{\longrightarrow} \G$, then $\neg \vp \in \G$.
  But, \tree{S} contains at least one leaf colored with such $\D'$
  that $\D \stackrel{\move{}'}{\longrightarrow} \D'$, and this $\D'$
  was obtained by extending a $\G$ with $\D
  \stackrel{\move{}'}{\longrightarrow} \G$; hence, $\neg \vp \in \D'$,
  and the statement of the lemma follows.

  Case 2: $A = \agents_{\theta}$.  Then, by virtue of (H2),
  $\coalnext{\emptyset} \neg \vp \in \D$ and thus, by the rule
  \Rule{Next}, $\neg \vp \in \G$ for every $\G \in \sucpr{\D}$.  Then,
  $\neg \vp \in \D'$ for every $\D'$ that is a successor of $\D$ in
  $\tableau{T}^{\theta}$ and hence in the coloring set of every leaf
  of \tree{S}.  Then, the (unique) co-$\agents_{\theta}$-move, which
  is an identity function, has the required properties.
\end{proof}

Now, we come to the ``walls'' of our building---the components of the
future Hintikka structure that take care of single eventualities.
Following~\cite{GorDrim06}, we call them \fm{final tree components}.
Each final tree component is built around a realization witness tree
for the corresponding eventuality (recall
Definitions~\ref{def:realisation_witness_tree_until} and
\ref{def:realisation_witness_tree_negbox}), the existence of which is
proved in the forthcoming lemma.

\begin{lemma}
  \label{lem:existence_of_rwt}
  Let $\xi$ be an eventuality realized at $\D$ in
  $\tableau{T}_n^{\theta}$.  Then, there exists a realization witness
  tree $\mathcal{R}$ for $\xi$ at $\D$ in $\tableau{T}_n^{\theta}$.
\end{lemma}

\begin{proof}
  To build $\tree{R}$, we use the concept of the realization rank of
  $\D$ in $\tableau{T}_n^{\theta}$ with respect to an eventuality
  $\xi$, which we define as the shortest path from $\D$ to a state
  witnessing the realization of $\xi$ at $\D$ (if $\xi = \coal{A} \vp
  \until \psi$, such a state contains $\psi$; if $\xi = \neg
  \coalbox{A} \vp$, then such a state contains $\neg \vp$) and denoted
  by $\mathbf{rank}(\D, \xi, \tableau{T}_n^{\theta})$.  If such a path
  does not exists, then $\mathbf{rank}(\D, \xi,
  \tableau{T}_n^{\theta}) = \infty$.
  Clearly, if $\xi$ is realized at $\D$ in $\tableau{T}_n^{\theta}$,
  then $\mathbf{rank}(\D, \xi, \tableau{T}_n^{\theta})$ is finite.

  Suppose, first, that $\xi$ is of the form $\coal{A} \vp \until
  \psi$.  We start building $\tree{R}$ by taking a root node and
  coloring it with $\D$.  Afterwards, for every $w' \in \tree{R}$
  colored with $\D'$, we do the following: for every $\move{}
  \sqsupseteq \move{A}[\coalnext{A} \coal{A} \vp \until \psi] \in
  D(\D')$, we pick the $\D'' \in \suctab{\D'}{\move{}}$ with the least
  $\mathbf{rank}(\D'', \coal{A} \vp \until \psi,
  \tableau{T}_n^{\theta})$ and add to \tree{R} a child $w''$ of $w'$
  colored with $\D''$.  As $\coal{A} \vp \until \psi$ is realized at
  $\D$, it follows that $\mathbf{rank}(\D, \coal{A} \vp \until \psi,
  \tableau{T}_n^{\theta})$ is finite. By construction of \tree{R} and
  definition of the rank, each child of every node of the so
  constructed tree has a smaller realization rank than the
  parent.  Therefore, along each branch of the tree we are bound to
  reach in a finite number of steps a node colored with a state whose
  realization rank with respect to $\coal{A} \vp \until \psi$ is
  $0$; such nodes are taken to be the leaves of $\tree{R}$.  As every
  node of $\tree{R}$ has finitely many children, due to K\"{o}nig's
  lemma, $\tree{R}$ is finite.  Therefore, so constructed $\tree{R}$
  is indeed a realization witness tree for $\coal{A} \vp \until \psi$
  at $\D$ in $\tableau{T}_n^{\theta}$.

  Suppose, next, that $\xi$ is of the form $\neg \coalbox{A} \vp$.
  Again, we begin by taking a root node and coloring it with $\D$.
  Afterwards, for every $w' \in \tree{R}$ colored with $\D'$, we do
  the following: for every $\move{} = \comove{A}[\neg \coalnext{A}
  \coalbox{A} \vp] (\move{A}) \in D(\D')$, we pick the $\D'' \in
  \suctab{\D'}{\move{}}$ with the least $\mathbf{rank}(\D'', \neg
  \coalbox{A} \vp, \tableau{T}_n^{\theta})$ and add to \tree{R} a
  child $w''$ of $w'$ colored with $\D''$. The rest of the argument
  is analogous to the one for the other eventuality.
\end{proof}

Now, we are going to use realization witness trees to build
$\tableau{T}^{\theta}$-trees doing the same job for eventualities as
realization witness trees do, i.e., ``realizing'' them in a certain
sense.  The problem with realization witness trees is that their nodes
might lack successors along some ``move vectors''; the next definition
and lemma show that this shortcoming can be easily remedied, by giving
each interior node $\D$ of a realization witness tree a successor
associated with every move vector $\sigma \in \D$.

\begin{definition}
  \label{def:realisation_by_T_trees}
  Let $\tree{W} = (W, \arc)$ be a locally consistent
  $\tableau{T}^{\theta}$-tree rooted at $\D$ and $\xi \in \D$ be an
  eventuality formula.  We say that $\tree{W}$ \de{realizes} $\xi$ if
  there exists a subtree\footnote{By a subtree, we mean a graph
    obtained from a tree by removing some of its nodes together with
    all the nodes reachable from them.}  $\tree{R}_{\xi}$ of \tree{W}
  rooted at $\D$ such that $\tree{R}_{\xi}$ is a realization witness
  tree for $\xi$ at $\D$ in $\tableau{T}^{\theta}$.
\end{definition}

\begin{lemma}
  \label{lem:final_components_for_eventualities}
  Let $\xi \in \D \in S^{\theta}$ be an eventuality formula.  Then,
  there exist a finite locally consistent $\tableau{T}^{\theta}$-tree
  $\tree{W}_{\xi}$ rooted at $\D$ realizing $\xi$.
\end{lemma}

\begin{proof}
  Take the realization witness tree $\tree{R}_{\xi}$ for $\xi$ at $\D$
  in $\tableau{T}^{\theta}$, which exists by
  Lemma~\ref{lem:existence_of_rwt}.  The only reason why
  $\tree{R}_{\xi}$ may turn out not to be a locally consistent
  $\tableau{T}^{\theta}$-tree is that some of its interior nodes do
  not have a successor node along every move vector $\move{}$ (recall
  that, in realization witness trees, every interior node has just
  enough successors to witness realization of the corresponding
  eventuality, and no more).  Therefore, to build a locally consistent
  $\tableau{T}^{\theta}$-tree out of $\tree{R}_{\xi}$, we simply add
  to its interior nodes just enough ``colored'' successors so that (1)
  for every interior node $w'$ of $\tree{W}_{\xi}$ and every $\move{}
  \in D(w')$, the tree $\tree{W}_{\xi}$ contains a $w''$ such that
  $c(w'') = \D''$ for some $\D'' \in \suctab{\D'}{\move{}}$ (where
  $\D' = c(w')$) and (2) $\tree{W}_{\xi}$ satisfies the condition of
  Definition~\ref{def:locally_consistent}.  It is then obvious that
  $\tree{W}_{\xi}$ is a locally consistent
  $\tableau{T}^{\theta}$-tree, by definition realizing $\xi$.
  Moreover, as according to Lemma~\ref{lem:existence_of_rwt},
  $\tree{R}_{\xi}$ is finite, $\tree{W}_{\xi}$ is finite, too.
\end{proof}

We want to build Hintikka structures our of locally consistent
$\tableau{T}^{\theta}$-trees.  Hintikka structures are based on CGFs;
therefore, we need to be able to ``embed'' such trees into CGFs.  The
following definition formally defines such an embedding.

\begin{definition}
  Let $\tree{W} = (W, \arc)$ be a locally consistent
  $\tableau{T}^{\theta}$-tree and $\kframe{F} = (\agents_{\theta}, S,
  d, \delta)$ be a CGF.  We say that \tree{W} is \de{contained} in
  \kframe{F}, denoted $\tree{W} \ll \kframe{F}$, if the following
  conditions hold:
  \begin{itemize}
  \item $W \subseteq S$;
  \item if $\move{} \in l (w \arc w')$, then $w' = \delta(w,
    \move{})$.
  \end{itemize}
\end{definition}

Locally consistent $\tableau{T}^{\theta}$-trees realizing an
eventuality $\xi$ are meant to represent run trees in CGMs effected by
(co-)strategies.  We now show that if we embed the former variety of
tree into a CGM then, as expected, this gives rise to a positional
(co-)strategy witnessing the truth of $\xi$ under an ``appropriate
valuation''.  (Intuitively, this (co-)strategy is extracted out of a
locally consistent $\tableau{T}^{\theta}$-tree when it is embedded
into a CGF and can, thus, be viewed as a run tree).  The following two
lemmas prove this for two types of eventualities we have in the
language.

\begin{lemma}
  \label{lem:realisation_within_final_tree_components_until}
  Let, $\coal{A} \vp \until \psi \in \D \in S^{\theta}$ and let
  $\tree{W} = (W, \arc)$ be a locally consistent
  $\tableau{T}^{\theta}$-tree rooted at $\D$ and realizing $\coal{A}
  \vp \until \psi$.  Let, furthermore, $\kframe{F} =
  (\agents_{\theta}, S, d, \delta)$ be a CGF such that $\tree{W} \ll
  \kframe{F}$.  Then, there exists a positional $A$-strategy \str{A}
  in \kframe{F} such that, if $\lambda \in out (w, \str{A})$, where
  $c(w) = \D$, then there exists $i \geq 0$ such that $\psi \in
  \lambda[i] \in W$ and $\vp \in \lambda[j] \in W$ holds for all $0
  \leq j < i$.
\end{lemma}

\begin{proof}
  At every node $w'$ of the realization witness tree for $\coal{A} \vp
  \until \psi$, which is contained in \tree{W}, take the $A$-move
  $\move{A} [\coalnext{A} \coal{A} \vp \until \psi] \in
  \moves{A}{w'}$.  At any other node, for definiteness' sake, take the
  lexicographically first $A$-move.  This strategy is clearly
  positional and has the required property.
\end{proof}

\begin{lemma}
  \label{lem:realisation_within_final_tree_components_box}
  Let, $\neg \coalbox{A} \vp \in \D \in S^{\theta}$ and let $\tree{W}
  = (W, \arc)$ be a locally consistent $\tableau{T}^{\theta}$-tree
  rooted at $\D$ and realizing $\neg \coalbox{A} \vp$.  Let,
  furthermore, $\kframe{F} = (\agents_{\theta}, S, d, \delta)$ be a
  CGF such that $\tree{W} \ll \kframe{F}$.  Then, there exists a
  positional co-$A$-strategy \costr{A} in \kframe{F} such that, if
  $\lambda \in out (w, \costr{A})$, where $c(w) = \D$, then $\neg \vp
  \in \lambda[i] \in W$ for every $i \geq 0$.
\end{lemma}

\begin{proof}
  At every node $w'$ of the realization witness tree for $\neg
  \coalbox{A}$, which is contained in \tree{W}, take the co-$A$-move
  $\comove{A} [\neg \coalnext{A} \coalbox{A} \vp] \in \moves{A}{w'}$.
  At any other node, for definiteness' sake, take the
  lexicographically first co-$A$-move.  This co-$A$-strategy is
  clearly positional and has the required property.
\end{proof}

Our next big step in the completeness proof is to assemble locally
consistent $\tableau{T}^{\theta}$-trees realizing eventualities as
well as locally consistent simple $\tableau{T}^{\theta}$-trees into a
Hintikka structure for $\theta$. To do that, we need the concept of
partial concurrent game frame that generalizes that of CGF.  Partial
CGFs are different from CGFs in that they have ``deadlocked'' states,
i.e., states for which the transition function $\delta$ is not defined
(the analog in Kripke frames would be ``dead ends''---the nodes that
cannot ``see'' any other node); however, each deadlocked state of a
partial CGF is an image of a transition function $\delta$ for some
(ordinary) state.  We need partial CGFs as we will be building a
Hintikka structure for $\theta$ step-by-step, all but the final step
producing partial CGFs having deadlocked states that will be given
successors at the next stage of the construction.  Put another way,
the motivation for introducing partial CGFs is that locally consistent
$\tableau{T}^{\theta}$-trees are partial CGFs, and we want to build a
Hintikka structure for $\theta$ out of such trees.

\begin{definition}
  \label{def:cgs}
  A \de{partial concurrent game frame} (partial CGF, for short) is a
  tuple $\kframe{S} = (\agents, S, Q, d, \delta)$, where
  \begin{itemize}
  \item $\agents$ is a finite, non-empty set of \de{agents};
  \item $S \ne \emptyset$ is a set of \de{states};
  \item $Q \subseteq S$ is a set of \de{deadlock states};
  \item $d$ is a function assigning to every $a \in \agents$ and every
    $s \in S \setminus Q$ a natural number $d_a(s) \geq 1$ of
    \de{moves} available to agent $a$ at state $s$; notation
    $\moves{a}{s}$ and $\moves{}{s}$ has the same meaning as in the
    case of CGFs (see Definition~\ref{def:cgf});
  \item $\delta$ is a \de{transition function} satisfying the
    following requirements:
    \begin{itemize}
    \item $\delta (s, \move{}) \in S$ for every $s \in S \setminus Q$
      and every $\move{} \in \moves{}{s}$;
    \item for every $q \in Q$, there exist $s \in S \setminus Q$ and
      $\move{} \in \moves{}{s}$ such that $q = \delta(s, \move{})$.
    \end{itemize}
  \end{itemize}
\end{definition}

The concept of $A$-move is defined for partial CGFs in a way analogous
to the way it is defined for CGFs; the only difference is that, in the
former case, $A$-moves are only defined for states in $S \setminus Q$.
The set of all $A$-moves at state $s \in S \setminus Q$ is denoted by
$\moves{A}{s}$.  Outcomes of $A$-moves are defined exactly as for
CGFs.  Analogously for co-$A$-moves.

\begin{definition}
  \label{def:pos_strategies_for_cgs}
  Let $\kframe{S} = (\agents, S, Q, d, \delta)$ be a partial CGF and
  $A \subseteq \agents$. A \de{positional $A$-strategy in \kframe{S}}
  is a mapping $\str{A}: S \mapsto \bigunion \crh{\moves{A}{s}}{s \in
    S \setminus Q}$ such that $\str{A} (s) \in \moves{A}{s}$ for all
  $s \in S \setminus Q$.
\end{definition}

\begin{definition}
  \label{def:pos_costrategies_for_cgs}
  Let $\kframe{S} = (\agents, S, Q, d, \delta)$ be a partial CGF and
  $A \subseteq \agents$. A \de{positional co-$A$-strategy in
    \kframe{S}} is a mapping $\costr{A}: S \mapsto \bigunion
  \crh{\comoves{A}{s}}{s \in S \setminus Q}$ such that $\costr{A} (s)
  \in \comoves{A}{s}$ for all $s \in S \setminus Q$.
\end{definition}

We now establish a fact that will be crucial to our ability to stitch
partial CGFs that are locally consistent $\tableau{T}^{\theta}$-trees
together. Intuitively, given such a partial CGF \kframe{S} and a state
$w$ of \kframe{S} colored with a set $\D''$ containing an eventuality
$\coal{A} \vp \until \psi$, coalition $A$ has a strategy such that
every (finite) run compliant with that strategy either realizes
$\coal{A} \vp \until \psi$ or postpones its realization until a
deadlocked state (Lemma~\ref{lem:duration}).  Analogously for
eventualities of the form $\neg \coalbox{A} \vp$ and co-$A$-strategies
(Lemma~\ref{lem:postponement}).  First, a technical definition.

\begin{definition}
  Let $\kframe{S} = (\agents, S, Q, d, \delta)$ be a partial CGF and
  let $s \in S$.  An $s$-\de{fullpath} in \kframe{S} is a finite
  sequence $\rho = s_0, \ldots, s_n$ of elements of $S$ such that
  \begin{itemize}
  \item $s_0 = s$;
  \item for every $0 \leq i < n$, there exists $\move{} \in
    \moves{}{s_i}$ such that $s_{i+1} = \delta(s_i, \move{})$;
  \item $s_{n} \in Q$.
  \end{itemize}
  The fullpath $\rho = s_0, \ldots, s_n$ is compliant with the
  strategy \str{A}, denoted $\rho \in out(\str{A})$, if $s_{i+1} \in
  out (\str{A} (s_i))$ for all $0 \leq i < n$.  Analogously for
  co-strategies. The length of $\rho$ (defined as the number of
  positions in $\rho$) is denoted by $|\rho|$.
\end{definition}

\begin{lemma}
  \label{lem:duration}
  Let $\kframe{S} = (\agents_{\theta}, S, Q, d, \delta)$ be a partial
  CGF such that
  \begin{enumerate}
  \item $S \subseteq S^{\theta}$;
  \item for every $w \in S$, the set $\set{w} \union \crh{w'}{w' =
      \delta (w, \move{}), \text{ for some } \move{} \in
      \moves{}{w}}$ is a set of nodes of a locally consistent simple
    $\tableau{T}^{\theta}$-tree;
  \item $\coal{A} \vp \until \psi \in \D''$, where $\D'' = c(w'')$ for
    some $w'' \in S$;
  \end{enumerate}
  Then, there exists a positional $A$-strategy \str{A} in \kframe{S}
  such that, for every $w''$-fullpath $\rho \in out (\str{A})$, either
  of the following holds:
  \begin{itemize}
  \item there exists $0 \leq i < |\rho|$ such that $\psi \in
    c(\rho[i])$ and $\vp \in c(\rho[j])$ for every $0 \leq j < i$;
  \item $\vp \in c(\rho[i])$ for every $0 \leq i < |\rho|$.
  \end{itemize}
\end{lemma}

\begin{proof}
  Straightforward.
\end{proof}

\begin{lemma}
  \label{lem:postponement}
  Let $\kframe{S} = (\agents_{\theta}, S, Q, d, \delta)$ be a partial
  CGF such that
  \begin{enumerate}
  \item $S \subseteq S^{\theta}$;
  \item for every $w \in S$, the set $\set{w} \union \crh{w'}{w' =
      \delta (w, \move{}), \text{ for some } \move{} \in \moves{}{w}}$
    is a set of nodes of a locally consistent simple
    $\tableau{T}^{\theta}$-tree;
  \item $\neg \coalbox{A} \vp \in \D''$, where $\D'' = c(w'')$ for
    some $w'' \in S$;
  \end{enumerate}
  Then, there exists a positional co-$A$-strategy \costr{A} in
  \kframe{S} such that $\neg \vp \in c(\rho[i])$ for every
  $\D''$-fullpath $\rho \in out (\costr{A})$ and every $i \geq 0$.
\end{lemma}

\begin{proof}
  Straightforward.
\end{proof}

Now, we define the building blocks, referred to as \de{final tree
  components}, from which a Hintikka structure for $\theta$ will be
built; the construction is essentially taken from~\cite{GorDrim06}.

\begin{definition}
  \label{def:final_tree component}
  Let $\D \in S^{\theta}$ and $\xi \in \tableau{T}^{\theta}$ be an
  eventuality formula.  Then, the \de{final tree component for $\xi$
    and \D}, denoted $\tree{F}_{(\xi, \D)}$, is defined as follows:
  \begin{itemize}
  \item if $\xi \in \D$, then $\tree{F}_{(\xi, \D)}$ is a finite locally
    consistent $\tableau{T}^{\theta}$-tree $\tree{W}_{\xi}$ rooted at
    $\D$ realizing $\xi$; the existence of such a tree being
    guaranteed by Lemma~\ref{lem:final_components_for_eventualities};
  \item if $\xi \notin \D$, then $\tree{F}_{(\xi, \D)}$ is a locally
    consistent simple $\tableau{T}^{\theta}$-tree rooted at $\D$; the
    existence of such a tree being guaranteed by
    Lemma~\ref{lem:existence_of_lcst}.
  \end{itemize}
\end{definition}

We are now ready to define what we will prove to be a (positional)
Hintikka structure for the input formula $\theta$, which we denote by
$\hintikka{H}_{\theta}$.  We start by defining the CGF \kframe{F}
underlying $\hintikka{H}_{\theta}$.

To that end, we first arrange all states of $\tableau{T}^{\theta}$ in
a list $\D_0, \ldots, \D_{n-1}$ and all eventualities occurring in the
states of $\tableau{T}^{\theta}$ in a list $\xi_{0}, \ldots,
\xi_{m-1}$.  We then think of all the final tree components (see
Definition~\ref{def:final_tree component}) as arranged in an
$m$-by-$n$ grid whose rows are marked with the correspondingly
numbered eventualities of $\tableau{T}^{\theta}$ and whose columns are
marked with the correspondingly numbered states of
$\tableau{T}^{\theta}$.  The final tree component found at the
intersection of the $i$th row and the $j$th column will be denoted by
$\tree{F}_{(i, j)}$. The building blocks for \kframe{F} will all come
from the grid, and we build \kframe{F} incrementally, at each state of
the construction producing a partial CGF realizing more and more
eventualities.  The crucial fact here is that if an eventuality $\xi$
is not realized within a partial CGF used in the construction, then
$\xi$ is ``passed down'' to be realized later, in accordance with
Lemmas~\ref{lem:duration} and \ref{lem:postponement}.

We start off with a final tree component that is uniquely determined
by $\theta$, in the following way.  If $\theta$ is an eventuality,
i.e., $\theta = \xi_p$ for some $0 \leq p < m$, then we start off with
the component $\tree{F}_{(p, q)}$ where, for definiteness, $q$ is the
least number $< n$ such that $\theta \in \D_q$; as
$\tableau{T}^{\theta}$ is open, such a $q$ exists.  If, on the other
hand, $\theta$ is not an eventuality, then we start off with
$\tree{F}_{(0, q)}$, where $q$ is as described above.  Let us denote
this initial partial CGF by $\kframe{S}_0$.

Henceforth, we proceed as follows.  Informally, we think of the above
list of eventualities as a queue of customers waiting to be served.
Unlike the usual queues, we do not necessarily start serving the queue
from the first customer (if $\theta$ is an eventuality, then it gets
served first; otherwise we start from the beginning of the queue), but
then we follow the queue order, curving back to the beginning of the
queue after having served its last eventuality if we started in the
middle.  Serving an eventuality $\xi$ amounts to appending to
deadlocked states of the partial CGF constructed so far final tree
components realizing $\xi$.  Thus, we keep track of what eventualities
have already been ``served'' (i.e., realized), take note of the one
that was served the last, say $\xi_i$, and replace every deadlocked
state $w$ such that $c(w) = \D_j$ of the partial CGF so far
constructed with the final tree component $\tree{F}_{((i + 1)\!\!\!
  \mod m, j) )}$. The process continues until all the eventualities
have been served, at which point we have gone the full cycle through
the queue.

After that, the cycle is repeated, but with a crucial modification
that will guarantee that the CGHS we are building is going to be
finite: whenever the component we are about to attach, say
$\tree{F}_{(i,j)}$, is already contained in the partial CGF we have
constructed thus far, instead of replacing the deadlocked state $w$
(such that $c(w) = \D_j$) with that component, we connect every
``predecessor'' $v$ of $w$ to the root of $\tree{F}_{(i,j)}$ by an
arrow $\arc$ marked with the set $l(v \arc w)$.  This modified
version of the cycle is repeated until we come to a point when no more
components get added.  This result in a finite CGF \kframe{F}.  Now,
to define $\hintikka{H}_{\theta}$, we simply put $H (w) = c(w)$, for
every $w \in \kframe{F}$.

\begin{theorem}
  The above defined $\hintikka{H}_{\theta}$ is a (positional) Hintikka
  structure for $\theta$.
\end{theorem}

\begin{proof}
  The ``for $\theta$'' part immediately follows the construction of
  $\hintikka{H}_{\theta}$ (recall the very first step of the
  construction).  It, thus, remains to argue that
  $\hintikka{H}_{\theta}$ is indeed a Hintikka structure.

  Conditions (H1)--(H3) of Definition~\ref{def:cghs} hold since states
  of $\hintikka{H}_{\theta}$ are consistent downward saturated sets.

  Conditions (H4) and (H5) essentially follow from
  Lemma~\ref{lem:local_consistency}.

  Condition (H6) follows from the way $\hintikka{H}_{\theta}$ is
  constructed together with
  lemmas~\ref{lem:realisation_within_final_tree_components_until} and
  \ref{lem:duration}.  Lastly, condition (H7) follows from the way
  $\hintikka{H}_{\theta}$ is constructed together with
  lemmas~\ref{lem:realisation_within_final_tree_components_box} and
  \ref{lem:postponement}.

  Lastly, $\hintikka{H}_{\theta}$ is positional by construction.
  Indeed, it is built from final tree components, which are locally
  consistent simple $\tableau{T}^{\theta}$-trees; as we have seen in
  Lemmas~\ref{lem:realisation_within_final_tree_components_until} and
  \ref{lem:realisation_within_final_tree_components_box}, when
  embedded into CGFs, these trees give rise to positional strategies.
\end{proof}

The positionality of $\hintikka{H}_{\theta}$ gives us the following,
stronger, version of the completeness theorem for our tableau
procedure:

\begin{theorem}[Positional completeness]
  \label{thr:completeness} Let $\theta$ be an $\ATL$ formula and let
  $\tableau{T}^{\theta}$ be open. Then, $\theta$ is satisfiable in a
  CGM based on a frame with positional strategies.
\end{theorem}

\begin{corollary}
  \label{cor:positionality}
  If an $\ATL$-formula $\theta$ is tightly satisfiable, then it is
  tightly satisfiable in a positional CGM.
\end{corollary}

\begin{proof} Suppose that $\theta$ is tightly satisfiable in a CGM
  based on a CGF with perfect recall strategies.  Then, by
  Theorem~\ref{thr:soundness}, the tableau $\tableau{T}^{\theta}$ for
  $\theta$ is open.  It then follows form
  Theorem~\ref{thr:completeness} that $\theta$ is satisfiable in a
  positional CGM.
\end{proof}

%%%%%%%%%%%%%%%%%%%%%%%%%%%%%%%%%%
%%%%%%%%%%%%%%%%%%%%%%%%%%%%%%%%%%
\section{Some variations of the method}
\label{sec:flexibility}
%%%%%%%%%%%%%%%%%%%%%%%%%%%%%%%%%%
%%%%%%%%%%%%%%%%%%%%%%%%%%%%%%%%%%

In the present section, we sketch some immediate adaptations of the
tableau method described above for testing other strains of
satisfiability, such as loose \ATL-satisfiability and
\ATL-satisfiability over some special classes of frames. Other, less
straightforward, adaptations will be developed in follow-up work.

%%%%%%%%%%%%%%%%%%%%%%%%%%%%%%%%%%
\subsection{Loose satisfiability for \ATL}
\label{sec:loose_sat_atl}
%%%%%%%%%%%%%%%%%%%%%%%%%%%%%%%%%%

The procedure described above is easily adaptable to testing
$\ATL$-formulae for loose satisfiability, which the reader will
recall, is satisfiability over frames with exactly one agent not
featuring in the formula.  All that is necessary to adapt the
above-described procedure to testing for this strain of satisfiability
is the modification of the \Rule{Next} rule in such a way that it
accommodates $\card{\agents_{\theta}} + 1$ agent rather than
$\card{\agents_{\theta}}$.  As such a modification is entirely
straightforward, we omit the details.  The complexity of the procedure
is not affected.

%%%%%%%%%%%%%%%%%%%%%%%%%%%%%%%%%%
\subsection{\ATL\ over special classes of frames}
\label{sec:special_classes_of_frames}
%%%%%%%%%%%%%%%%%%%%%%%%%%%%%%%%%%

Some classes of concurrent game frames are of particular interest (for
motivation and examples, see~\cite{AHK02}).

%%%%%%%%%%%%%%%%%%%%%%%%%%%%%%%%%%
\subsubsection{Turn-based synchronous frames}
\label{sec:tbs}
%%%%%%%%%%%%%%%%%%%%%%%%%%%%%%%%%%

In turn-based synchronous frames, at every state, exactly one agent
has ``real choices''.  Thus, agents take it in turns to act.

\begin{definition}
  \label{def:tbs}
  A concurrent game frame $\kframe{F} = (\agents, S, d, \delta)$ is
  \de{turn-based synchronous} if, for every $s \in S$, there exists
  agent $a_s \in \agents$, referred to as the \de{owner of $s$}, such
  that $\nmoves{a}{s} = 1$ for all $a \in \agents \setminus
  \set{a_s}$.
\end{definition}

To tests formulae for satisfiability over turn-based synchronous
frames, we need to make the following adjustments to the above tableau
procedure (we are assuming that we are testing for tight
satisfiability; loose satisfiability is then straightforward).  All
the states of the tableau are now going to be ``owned'' by individual
agents.  Intuitively, if $\D$ is ``owned'' by $a \in
\agents_{\theta}$, it is agent's $a$ turn to act at $\D$; we indicate
ownership by affixing the name of the owner as a subscript of the
state.  The rule \Rule{SR} now looks as follows:

\bigskip

\Rule{SR} Given a prestate $\G$, do the following:
\begin{enumerate}
\item for every $a \in \agents_{\theta}$, add to the pretableau all
  the minimal downward saturated extensions of $\G$, marked with $a$
  (all thus created sets $\D_a$ are \de{$a$-states});

\item for each of the so obtained states $\D_a$, if $\D_a$ does not
  contain any formulae of the form $\coalnext{A} \vp$ or $\neg
  \coalnext{A} \vp$, add the formula $\coalnext{\agents_{\theta}}
  \truth$ to $\D_a$;

\item for each state $\D_a$ obtained at steps 1 and 2, put $\G
  \brancharrow \D_a$;

\item if, however, the pretableau already contains a state $\D'_a$
  that coincides with $\D_a$, do not create another copy of $\D'_a$,
  but only put $\G \brancharrow \D'_a$.
\end{enumerate}

Moreover, when creating prestates from $a$-states, all agents except
$a$ get exactly one vote, while $a$ can still vote for any next-time
formula in the current state.  The rule \Rule{Next}, therefore, now
looks as follows:

\bigskip

\Rule{Next} Given a state $\D_a$ such that for no $\chi$ we have
$\chi, \neg \chi \in \D_a$, do the following:
\begin{enumerate}
\item order linearly all positive and proper negative next-time formulae of
  $\D_a$ in such a way that all the positive next-time formulae precede all
  the negative ones; suppose the result is the list
   \[\mathbb{L} = \coalnext{A_0} \vp_0, \ldots, \coalnext{A_{m-1}}
   \vp_{m-1},\lnot \coalnext{A'_0} \psi'_0, \ldots, \lnot
   \coalnext{A'_{l-1}} \psi_{l-1}.\] (Due to step 2 of \Rule{SR},
   $\mathbb{L}$ is non-empty.) Let $r_{\D} = m + l$; denote by $D
   (\D_a)$ the set $\crh{\move{} \in \nat^{\card{\agents_{\theta}}}}{0
     \leq \move{a} < r_{\D} \text{ and } \move{b} = 0, \text{ for all
     } b \ne a}$;

 \item consider the elements of $D(\D_a)$ in the lexicographic order
   and for each $\sigma \in D(\D_a)$ do the following:
   \begin{enumerate}
 \item create a prestate
   \begin{eqnarray*}
     \G_{\move{}} & = & \crh{\vp_p}{\coalnext{A_p} \vp_p \in \D_a
       \text{ and } a \in A_p \text{ and }  \move{a} = p }\\
     & \union & \crh{\vp_p}{\coalnext{A_p} \vp_p \in \D_a
       \text{ and } a \notin A_p}  \\
     & \union & \crh{\neg \psi_q}{\neg \coalnext{A'_q}
       \psi_q \in \D_a \text{ and } a \in A'_q}  \\
     & \union & \crh{\neg \psi_q}{\neg \coalnext{A'_q} \psi_q \in \D_a
       \text{ and } a \notin A'_q \text{ and } \move{a} = q}
     \end{eqnarray*}
     put $\G_{\move{}} = \{\top\}$ if all four sets above above are
     empty.
   \item connect $\D_a$ to $\G_{\move{}}$ with $\movearrow$;
   \end{enumerate}
   If, however, $\G_{\move{}} = \G$ for some prestate $\G$ that has
   already been added to the pretableau, only connect $\D$ to $\G$
   with $\movearrow$.
 \end{enumerate}

 Otherwise, tableaux testing for satisfiability over turn-based
 synchronous frames are no different from those for satisfiability
 over all frames.

%%%%%%%%%%%%%%%%%%%%%%%%%%%%%%%%%%
\subsubsection{Moore synchronous frames}
\label{sec:Moore}
%%%%%%%%%%%%%%%%%%%%%%%%%%%%%%%%%%

In Moore synchronous frames over the set of agents $\agents$, the set
of states $S$ can be represented as a Cartesian product of sets of
local states $S_{a \in \agents}$, one for each agent.  The actions of
agents are determined by the current ``global'' state $s \in S$; each
action $\move{a}$ of agent $a$ at state $s \in S$, however, results in
a local state determined by a function $\delta_a$ mapping pairs
$\langle$ global state, $a$-move$\rangle$ into $S_a$.  Then, given a
move vector $\move{} \in \moves{}{s}$, representing simultaneous
actions of all agents at $s$, the $\move{}$-successor of $s$ is
determined by the local states of agents produced by their
actions---namely, it is a $k$-tuple (where $k = {\card{\agents}}$) of
respective local states $(\delta_1 (s, \move{1}), \ldots, \delta_{k}
(s, \move{k}))$, one for each agent. This intuition can be formalized
as follows (see~\cite{AHK02}):

\begin{definition}
  A CGF $\kframe{F} = (\agents, S, d, \delta)$ is \de{Moore
    synchronous} if the following two conditions are satisfied, where $k =
  \card{\agents}$:
  \begin{itemize}
  \item $S = S_1 \times \dots \times S_{k}$;
  \item for each state $s \in S$, move vector $\move{}$, and agent $a
    \in \agents$, there exists a local state $\delta_a(s, \move{a})$
    such that $\delta(s, \sigma) = (\delta_1 (s, \move{1}), \ldots,
    \delta_{k} (s, \move{k}))$.
  \end{itemize}
\end{definition}

\begin{definition}
  A CGF $\kframe{F} = (\agents, S, d, \delta)$ is \de{bijective}, if
  $\delta(s, \move{}) \ne \delta(s, \move{}')$ for every $s \in S$ and
  every $\move{}$ and $\move{}'$ such that $\move{} \ne \move{}'$.
\end{definition}

It is easy to see that every bijective frame is isomorphic to a Moore
synchronous one.  Therefore, if---for whatever reason---using our
tableau procedure, we want to produce a Moore synchronous model for
the input formula, we simply never identify the states created in the
course of applying the \Rule{Next} rule.  This clearly produces a
bijective, and hence Moore synchronous, model.  By inspecting the
tableau procedure, it can be noted that identification or otherwise of
the states never affects the output of the procedure.  Therefore, an
analysis of our tableau procedure leads to the following claim:

\begin{theorem}[\cite{Goranko01}]
  Let $\theta$ be an $\ATL$-formula.  Then, $\theta$ is satisfiable in
  the class of all CGFs iff it is satisfiable in the class of Moore
  synchronous CGFs.
\end{theorem}

%%%%%%%%%%%%%%%%%%%%%%%%%%%%%%%%%%
%%%%%%%%%%%%%%%%%%%%%%%%%%%%%%%%%%
\section{Concluding remarks}
\label{sec:concluding}
%%%%%%%%%%%%%%%%%%%%%%%%%%%%%%%%%%
%%%%%%%%%%%%%%%%%%%%%%%%%%%%%%%%%%

We have developed a complexity-efficient terminating
incremental-tableau-based decision procedure for $\ATL$ and some of
its variations. This style of tableaux for $\ATL$, while having the
same worst-case upper bound as the other known decision procedures,
including the top-down tableaux-like procedure presented in \cite{WLWW06}, is
expected to perform better in practice because, as it has been shown in
the examples, it creates much fewer tableau states.

We believe that the tableau method developed herein is not only of
more immediate practical use, but also is more flexible and adaptable
than any of the decision procedures developed earlier in
\cite{vanDrimmelen03}, \cite{GorDrim06}, and \cite{WLWW06}. In
particular, this method can be suitably adapted to variations of
$\ATL$ with committed strategies \cite{AgotnesTARK2007} and with
incomplete information, which is the subject of a follow-up work.

%%%%%%%%%%%%%%%%%%%%%%%%%%%%%%%%%%
%%%%%%%%%%%%%%%%%%%%%%%%%%%%%%%%%%
\section{Acknowledgments}
%%%%%%%%%%%%%%%%%%%%%%%%%%%%%%%%%%
%%%%%%%%%%%%%%%%%%%%%%%%%%%%%%%%%%

This research was supported by a research grant of the National
Research Foundation of South Africa and was done during the second
author's post-doctoral fellowship at the University of the
Witwatersrand, funded by the Claude Harris Leon Foundation. We
gratefully acknowledge the financial support from these institutions.
We also gratefully acknowledge the detailed and useful referees' comments, which have helped us improve significantly the presentation of the paper.

\bibliographystyle{plain}

\begin{thebibliography}{10}

\bibitem{AGF07}
Pietro Abate, Rajeev Gor\`{e}, and Florian Widmann.
\newblock One-pass tableaux for {C}omputation {T}ree {L}ogic.
\newblock In {\em Lecture Notes in Computer Science}, pages 32--46.
  Springer-Verlag, 2007.
\newblock Proc. LPAR 2007.

\bibitem{AgotnesTARK2007}
Thomas {\AA}gotnes, Valentin Goranko, and Wojciech Jamroga.
\newblock Alternating-time temporal logics with irrevocable strategies.
\newblock In D.~Samet, editor, {\em Proceedings of the 11th International
  Conference on Theoretical Aspects of Rationality and Knowledge (TARK XI)},
  pages 15--24, Univ. Saint-Louis, Brussels, 2007. Presses Universitaires de
  Louvain.

\bibitem{AHK97}
Rajeev Alur, Thomas~A. Henzinger, and Orna Kuperman.
\newblock Alternating-time temporal logic.
\newblock In {\em Proceedings of the 38th IEEE Symposium on Foundations of
  Computer Science}, pages 100--109, October 1997.

\bibitem{AHK98}
Rajeev Alur, Thomas~A. Henzinger, and Orna Kuperman.
\newblock Alternating-time temporal logic.
\newblock In {\em Lecture Notes in Computer Science}, volume 1536, pages
  23--60. Springer-Verlag, 1998.

\bibitem{AHK02}
Rajeev Alur, Thomas~A. Henzinger, and Orna Kuperman.
\newblock Alternating-time temporal logic.
\newblock {\em Journal of the ACM}, 49(5):672--713, 2002.

\bibitem{AHMQR98}
Rajeev Alur, Thomas~A. Henzinger, F.~Y.~C. Mang, Shaz Qadeer, Sriram~K.
  Rajamani, and Serdar Tasiran.
\newblock Mocha: Modularity in model-checking.
\newblock In {\em Lecture Notes in Computer Science}, volume 1427, pages
  521--525. Springer-Verlag, 1998.

\bibitem{BradStir07}
Julian Bradfield and Colin Stirling.
\newblock Modal $\mu$-calculi.
\newblock In Patrick~Blackburn et~al., editor, {\em Handbook of Modal Logic},
  pages 721--756. Elsevier, 2007.

\bibitem{Emerson90}
E.~Allen Emerson.
\newblock Temporal and modal logics.
\newblock In J.~van Leeuwen, editor, {\em Handbook of Theoretical Computer
  Science}, volume~B, pages 995--1072. MIT Press, 1990.

\bibitem{EmHal85}
E.~Allen Emerson and Joseph Halpern.
\newblock Decision procedures and expressiveness in the temporal logic of
  branching time.
\newblock {\em Journal of Computation and System Sciences}, 30(1):1--24, 1985.

\bibitem{Fagin95knowledge}
Ron Fagin, Joseph Halpern, Yoram Moses, and Moshe Vardi.
\newblock {\em Reasoning about Knowledge}.
\newblock MIT Press: Cambridge, MA, 1995.

\bibitem{Fitting83}
Melvin Fitting.
\newblock {\em Proof Methods for Modal and Intuitionistic Logics}.
\newblock D. Reidel, 1983.

\bibitem{Fitting07}
Melvin Fitting.
\newblock Modal proof theory.
\newblock In P.~Blackburn et~al., editor, {\em Handbook of Modal Logic}, pages
  85--138. Elsevier, 2007.

\bibitem{Goranko01}
Valentin Goranko.
\newblock Coalition games and alternating temporal logics.
\newblock In Johan van Benthem, editor, {\em Proceedings of the 8th conference
  on Theoretical Aspects of Rationality and Knowledge (TARK VIII)}, pages
  259--272. Morgan Kaufmann, 2001.

\bibitem{GorJam04}
Valentin Goranko and Wojciech Jamroga.
\newblock Comparing semantics of logics for multi-agent systems.
\newblock {\em Synthese}, 139(2):241--280, 2004.

\bibitem{GorShkat3}
Valentin Goranko and Dmitry Shkatov.
\newblock Deciding satisfiability in the full coalitional multiagent epistemic
  logic with a tableau-based procedure.
\newblock Submitted, 2008.

\bibitem{GorShkat2}
Valentin Goranko and Dmitry Shkatov.
\newblock Tableau-based decision procedure for the multiagent epistemic logic
  with operators of common and distributed knowledge.
\newblock In A.~Cerone and S.~Gruner, editors, {\em Proc. of the Sixth IEEE
  conference on Software Engineering and Formal Methods (SEFM 2008)}. IEEE
  Computer Society Press, 2008, to appear.

\bibitem{GorDrim06}
Valentin Goranko and Govert van Drimmelen.
\newblock Complete axiomatization and decidablity of {A}lternating-time
  temporal logic.
\newblock {\em Theoretical Computer Science}, 353:93--117, 2006.

\bibitem{Gore98}
Rajeev Gore.
\newblock Tableau methods for modal and temporal logics.
\newblock In M.~D'Agostino et~al., editor, {\em Handbook of Tableau Methods}.
  Kluwer, 1998.

\bibitem{Hansen04}
Helle~Hvid Hansen.
\newblock Tableau games for {C}oalition {L}ogic and {A}lternating-time
  {T}emporal {L}ogic.
\newblock Master's thesis, University of Amsterdam, 2004.

\bibitem{Hewitt90}
Carl Hewitt.
\newblock The challenge of open systems.
\newblock In Derek Partridge and Yorick Wilks, editors, {\em The Foundations of
  Artificial Intelligence -- a Sourcebook}, pages 383--395. Cambridge
  University Press, 1990.

\bibitem{MMR00}
Maarten Marx, Szabolcs Mikul\'{a}s, and Mark Reynolds.
\newblock The mosaic method for temporal logics.
\newblock In {\em Lecture Notes in Computer Science}, volume 1847, pages
  324--340. Springer-Verlag, 2000.

\bibitem{Pauly01}
Marc Pauly.
\newblock {\em Logic for {S}ocial {S}oftware}.
\newblock PhD thesis, University of Amsterdam, 2001.
\newblock {ILLC} Dissertation Series 2001-10.

\bibitem{Pauly01a}
Marc Pauly.
\newblock A logical framework for coalitional effectivity in dynamic
  procedures.
\newblock {\em Bulletin of Economic Research}, 53(4):305--324, October 2001.

\bibitem{Pauly02}
Marc Pauly.
\newblock A modal logic for coalitional power in games.
\newblock {\em Journal of Logic and Computation}, 12(1):149--166, February
  2002.

\bibitem{PaulyPar03}
Marc Pauly and Rohit Parikh.
\newblock Game logic---an overview.
\newblock {\em Studia Logica}, 75(2):165--182, 2003.

\bibitem{Shoham08}
Yoav Shoham and Kevin Leyton-Brown.
\newblock {\em Multi-agent systems: Algorithmic, Game-Theoretic, and Logical
  Foundations}.
\newblock CUP, 2008.

\bibitem{Smullyan68}
Raymond~M. Smullyan.
\newblock {\em First-order Logic}.
\newblock Springer-Verlag, 1968.

\bibitem{Thomas95}
Wolfgang Thomas.
\newblock On the synthesis of strategies in infinite games.
\newblock In E.W. Mayr and C.~Puech, editors, {\em Proceedings of the 12th
  Annual Symposium on Theoretical Aspects of Computer Science, STACS '95},
  volume LNCS 900, pages 1--13. Springer, 1995.

\bibitem{vanDrimmelen03}
Govert van Drimmelen.
\newblock Satisfiability in alternating-time temporal logic.
\newblock In {\em Proceedings of 18th IEEE Symposium on Logic in Computer
  Science (LICS)}, pages 208--217, 2003.

\bibitem{WLWW06}
Dirk Walther, Carsten Lutz, Frank Wolter, and Michael Wooldridge.
\newblock {A}{T}{L} satisfiability is indeed {E}xp{T}ime-complete.
\newblock {\em Journal of Logic and Computation}, 16(6):765--787, 2006.

\bibitem{Weiss99}
G\"{u}nter Weiss, editor.
\newblock {\em Multiagent Systems}.
\newblock MIT Press, 1999.

\bibitem{Wolper85}
Pierre Wolper.
\newblock The tableau method for temporal logic: an overview.
\newblock {\em Logique et Analyse}, 28(110--111):119--136, 1985.

\bibitem{Wooldridge02}
Michael Wooldridge.
\newblock {\em An Introduction to Multiagent Systems}.
\newblock John Willey and Sons, 2002.
\end{thebibliography}

\end{document}